\documentclass[12pt,oneside,english]{article}

\usepackage{amstext,amsthm,amssymb} 
\usepackage{mathtools,mathrsfs} 
\usepackage[T1]{fontenc} 
\usepackage{lipsum} 
\usepackage{titlesec} 
\usepackage{xcolor} 
\usepackage{enumitem} 
\usepackage{simpler-wick} 
\usepackage{bbm,bm,dsfont} 
\usepackage[dvipsnames]{xcolor} 
\usepackage{url} 
\usepackage{pdfpages} 
\usepackage{relsize} 
\usepackage{esint} 
\usepackage{upgreek} 
\usepackage[hidelinks]{hyperref} 
\usepackage[justification=centering,labelfont=bf,sc]{caption} 
\usepackage[percent]{overpic} 
\usepackage{enumitem} 
\usepackage{simpler-wick} 
\usepackage{tikz} 
	\usetikzlibrary{decorations.pathmorphing} 
\usepackage{geometry} 
\numberwithin{equation}{section}
\numberwithin{figure}{section}
\newtheoremstyle{theorem} 
	{\topsep}                    
	{\topsep}                    
	{\itshape}                   
	{}                           
	{\scshape\bfseries}                   
	{.}                          
	{.5em}                       
	{}  
\theoremstyle{theorem}
\newtheorem{thm}{Theorem}[section]
\newtheorem*{thm*}{Theorem}
\newtheorem{lemma}[thm]{Lemma}
\newtheorem{prop}[thm]{Proposition}
\newtheorem{coro}[thm]{Corollary}

\theoremstyle{definition}

\newtheoremstyle{remark} 
	{\topsep}                    
	{\topsep}                    
	{}                           
	{}                           
	{\scshape\bfseries}          
	{.}                          
	{.5em}                       
	{}  
\theoremstyle{remark}
\newtheorem{rmk}[thm]{Remark}
\newtheorem{ex}[thm]{Example}
\numberwithin{thm}{section}

\makeatletter
\newcommand{\pushright}[1]{\ifmeasuring@#1\else\omit\hfill$\displaystyle#1$\fi\ignorespaces}
\newcommand{\pushleft}[1]{\ifmeasuring@#1\else\omit$\displaystyle#1$\hfill\fi\ignorespaces}
\makeatother

\def\N{{\mathbb N}}
\def\Z{{\mathbb Z}}

\def\R{{\mathbb R}}
\def\C{{\mathbb C}}
\newcommand{\Znneg}{\Z_{\geq0}}
\newcommand{\Zpos}{\Z_{>0}}
\newcommand{\One}{\mathbbm{1}}
\newcommand{\identity}{\mathbb{I}}

\newcommand{\E}{{\mathbb E}}
\renewcommand{\P}{{\mathbb P}}

\newcommand{\indicator}{\mathbbm{1}}

\newcommand{\beq}{\begin{align}}
\newcommand{\eeq}{\end{align}}
\newcommand{\nnbeq}{\begin{align*}}
\newcommand{\nneeq}{\end{align*}}

\newcommand{\black}{\color{black}}

\newcommand{\niceBlue}{\color{NavyBlue}}
\newcommand{\niceRed}{\color{RedOrange}}

\newcommand{\term}[1]{\textbf{#1}}
\newcommand{\mybar}[3]{%
    \mathrlap{\hspace{#2}\overline{\scalebox{#1}[1]{\phantom{\ensuremath{#3}}}}}\ensuremath{#3}
}

\newcommand{\bk}{\mathbf{k}}
\newcommand{\bl}{\bm{\ell}}
\newcommand{\thectt}{\alpha}
\newcommand{\EnvSymFer}{\mathcal{SF}}
\newcommand{\EnvSymFerBar}{\mybar{1}{1.5pt}{\mathcal{SF}}}
\newcommand{\FullEnvSymFer}{\EnvSymFer \otimes \EnvSymFerBar}
\newcommand{\FullSymFer}{\SymFer \oplus \SymFerBar}
\newcommand{\SymFer}{\textnormal{\textbf{SF}}}
\newcommand{\Vir}{\textnormal{\textbf{Vir}}}

\newcommand{\SymFerBar}{\overline{\SymFer}}
\newcommand{\AlgEta}[1]{{\upeta}_{#1}}
\newcommand{\AlgChi}[1]{{\upchi}_{#1}}
\newcommand{\AlgEtaBar}[1]{\mybar{0.97}{-0.6pt}{\upeta}_{#1}}
\newcommand{\AlgChiBar}[1]{\mybar{1}{-0.6pt}{\upchi}_{#1}}
\newcommand{\AlgK}{\mathsf{k}}
\newcommand{\AlgKBar}{\mybar{1}{-0.5pt}{\mathsf{k}}}
\newcommand{\UniEnve}{\mathcal{U}}

\newcommand{\FullFock}{\mathscr{F}{\!\!\textnormal{\footnotesize\itshape ock}}}
\newcommand{\GroundPartner}{\boldsymbol{\omega}}
\newcommand{\Ground}{\boldsymbol{\mathbbm{1}}}
\newcommand{\GroundEta}{\boldsymbol{\theta}}
\newcommand{\GroundChi}{\boldsymbol{\xi}}
\newcommand{\HolCurrentEta}{\boldsymbol{\eta}}
\newcommand{\HolCurrentChi}{\boldsymbol{\chi}}
\newcommand{\AntiHolCurrentEta}{\overline{\boldsymbol{\eta}}}
\newcommand{\AntiHolCurrentChi}{\overline{\boldsymbol{\chi}}}
\newcommand{\AlgL}[1]{\textnormal{L}_{#1}}
\newcommand{\AlgLBar}[1]{\mybar{0.85}{-1pt}{\textnormal{L}}_{#1}}
\newcommand{\AlgC}{\textnormal{C}}
\newcommand{\id}{\textnormal{id}}
\newcommand{\SugL}[1]{{\mathsf{L}}_{#1}}
\newcommand{\SugLBar}[1]{\mybar{0.6}{-0.3pt}{\mathsf{L}}_{#1}}
\newcommand{\dL}[1]{{\mathbf{L}}_{#1}}
\newcommand{\dLBar}[1]{\mybar{0.85}{-1pt}{\mathbf{L}}_{#1}}

\newcommand{\End}{\textnormal{End}}

\newcommand{\spn}{\mathrm{span}}

\newcommand{\ii}{\mathbbm{i}}
\newcommand{\domain}{\Omega}
\newcommand{\SFtag}{\textnormal{SyFe}}
\newcommand{\bigSFCorrFun}[3]{\big\langle #3 \big\rangle_{#1; #2}^{\SFtag}}
\newcommand{\BigSFCorrFun}[3]{\Big\langle #3 \Big\rangle_{#1; #2}^{\SFtag}}

\newcommand{\AnalyFun}[2]{C^{\,\omega}\big( #1 , #2 \big)}
\newcommand{\Conf}[2]{\textnormal{Conf}_{#1}( #2 )}

\newcommand{\Green}{\mathbf{G}}

\newcommand{\sgntr}{\textnormal{sgn}}
\newcommand{\bdry}{\partial}
\newcommand{\HarmOfGreen}{\mathbf{g}}

\newcommand{\field}{\boldsymbol{\varphi}}

\newcommand{\cd}{\textnormal{d}}

\newcommand{\Bij}{\mathfrak{B}}

\newcommand{\Zprimary}{\Z^2_\primary}
\newcommand{\Zdual}{\Z^2_\dual}
\newcommand{\Zdiamond}{\Z^2_\diamond}
\newcommand{\Zmedial}{\Z^2_\medial}
\newcommand{\Zcorner}{\Z^2_\corner}
\newcommand{\fDGFFtheta}{\theta}
\newcommand{\fDGFFxi}{\xi}
\newcommand{\meshsize}{\delta}
\newcommand{\ddomain}[1]{\Omega^{#1}}
\newcommand{\zd}{\mathbf{z}}
\newcommand{\zdbis}{\mathbf{w}}
\newcommand{\xd}{\mathbf{x}}
\newcommand{\xdbis}{\mathbf{y}}
\newcommand{\fDGFFtag}{\textnormal{fDGFF}}
\newcommand{\bigfDGFFCorrFun}[2]{\big\langle #2 \big\rangle_{#1}^{\fDGFFtag}}
\newcommand{\BigfDGFFCorrFun}[2]{\Big\langle #2 \Big\rangle_{#1}^{\fDGFFtag}}
\newcommand{\ActionfDGFF}{\mathbf{S}}
\newcommand{\Grass}{\textnormal{Gr}}
\newcommand{\FieldPoly}{\mathcal{P}}
\newcommand{\NullFieldPoly}{\mathcal{N}}
\newcommand{\Fields}{\mathcal{F}}
\newcommand{\ev}{\textnormal{ev}}
\newcommand{\zu}{\mathbf{u}}
\newcommand{\zubis}{\mathbf{v}}
\newcommand{\medial}{\textnormal{m}}
\newcommand{\primary}{\bullet}
\newcommand{\dual}{*}
\newcommand{\corner}{\textnormal{c}}
\newcommand{\antiholcurrtheta}{\mybar{0.9}{-0.2pt}{\boldsymbol{\upeta}}}
\newcommand{\antiholcurrxi}{\mybar{1.1}{-0.7pt}{\boldsymbol{\upchi}}}
\newcommand{\holcurrtheta}{\boldsymbol{\upeta}}
\newcommand{\holcurrxi}{\boldsymbol{\upchi}}
\newcommand{\pddeebar}{\mybar{0.8}{1pt}{\partial}^\bullet_\sharp}
\newcommand{\pddee}{\partial^\bullet_\sharp}
\newcommand{\dgrad}{\nabla_{\!\sharp}}
\newcommand{\ddeebar}{\mybar{0.8}{1pt}{\partial}_\sharp}
\newcommand{\partialBar}{\mybar{0.8}{1pt}{\partial}}
\newcommand{\ddee}{\partial_\sharp}
\newcommand{\dlaplacian}{\Delta_\sharp}
\newcommand{\Fieldtheta}{\boldsymbol\uptheta}
\newcommand{\Fieldxi}{\boldsymbol\upxi}
\newcommand{\Null}{\textnormal{\small Null}}
\newcommand{\HolCurrModetheta}[1]{\boldsymbol{\eta}_{#1}}
\newcommand{\HolCurrModexi}[1]{\boldsymbol{\chi}_{#1}}
\newcommand{\AntiHolCurrModetheta}[1]{\mybar{1}{0pt}{\boldsymbol{\eta}}_{#1}}
\newcommand{\AntiHolCurrModexi}[1]{\mybar{1}{0pt}{\boldsymbol{\chi}}_{#1}}
\newcommand{\dd}{\textnormal{d}^{\sharp}}
\newcommand{\dint}{\int^{\sharp}}
\newcommand{\dsqint}{\sqint^{\sharp}}
\newcommand{\Kronecker}{\boldsymbol{\delta}}
\newcommand{\Fock}{\FullFock}
\newcommand{\supp}{\textnormal{supp}\,}
\newcommand{\interior}{\textnormal{int}\,}
\newcommand{\interiordiam}{\textnormal{int}_\diamond}
\newcommand{\interiorprim}{\textnormal{int}_\primary}
\newcommand{\interiormed}{\textnormal{int}_\medial}
\newcommand{\fDGFFGround}{\mathbf{1}}
\newcommand{\fDGFFGroundTheta}{\boldsymbol{\uptheta}}
\newcommand{\fDGFFGroundXi}{\boldsymbol{\upxi}}
\newcommand{\fDGFFGroundPartner}{\boldsymbol{\upomega}}
\newcommand{\LinearFields}{\Fields_{\textnormal{lin}}}
\newcommand{\LinearFieldPoly}{\FieldPoly_{\textnormal{lin}}}
\newcommand{\RepHolChi}[1]{{X}_{#1}}
\newcommand{\RepAntiHolChi}[1]{\mybar{0.9}{1.5pt}{{X}}_{#1}}
\newcommand{\RepHolEta}[1]{{H}_{#1}}
\newcommand{\RepAntiHolEta}[1]{\mybar{0.9}{1.5pt}{{H}}_{#1}}
\newcommand{\RepHolChiNoIndex}{{X}}
\newcommand{\RepAntiHolChiNoIndex}{\mybar{0.9}{1.5pt}{{X}}}
\newcommand{\RepHolEtaNoIndex}{{H}}
\newcommand{\RepAntiHolEtaNoIndex}{\mybar{0.9}{1.5pt}{{H}}}
\newcommand{\LattSFAsso}{\mathcal{SF}\oplus\mybar{0.89}{2.5pt}{\mathcal{SF}}}

\newcommand{\noQuo}[1]
        {\, \rotatebox[]{90}{\scalebox{.8}{$\ \circ\,\circ$}}\,#1\,\rotatebox[]{90}{\scalebox{.8}{$\ \circ\,\circ$}}\,}
\newcommand{\no}[1]{\, \rotatebox[]{90}
        {\scalebox{.8}{$\ \bullet\,\bullet$}}\,#1\,\rotatebox[]{90}{\scalebox{.8}{$\ \bullet\,\bullet$}}\,}
\newcommand{\Pair}{\mathfrak{P}}

\newcommand{\degreeField}{\mathsf{deg}}
\newcommand{\Heights}[1]{\mathcal{H}_{#1}}
\newcommand{\Recurrent}[1]{\mathcal{R}_{#1}}
\newcommand{\heightOne}{\mathsf{h}_1}
\newcommand{\dissipation}{\mathsf{d}}
\newcommand{\bigASMseparator}{\rotatebox[]{90}{\scalebox{.65}{\textbf{ ~- - - }}}}
\newcommand{\ASMseparator}{\rotatebox[]{90}{\scalebox{.55}{\hspace{-2pt} \textbf{ ~- - - }}}}
\newcommand{\dheightOne}{\mathsf{H}_1}
\newcommand{\dissField}{\mathsf{D}}
\newcommand{\contHeightOne}{\mathbf{H}_1}
\newcommand{\dDegree}{\mathsf{T}}
\newcommand{\contDegree}{\mathbf{T}}
\newcommand{\OEright}{\!\bullet\!\mathbf{-}}
\newcommand{\OEleft}{\!\mathbf{-}\mspace{-7mu}\bullet\!}
\newcommand{\OEleftsubs}{\!\mathbf{-}\!\bullet}
\newcommand{\contHorEdgeField}{\mathbf{H}}
\newcommand{\horEdgeField}{\mathsf{H}}

\titleformat{\subsubsection}[runin]
    {\normalfont\bfseries}{\thesubsubsection.}{0pt}{\hspace{0.5em}}[.]
\titleformat{\subsection}
    {\normalfont\bfseries}{\thesubsection}{0pt}{\hspace{0.5em}}[\vspace{-5pt}]
\titleformat{\section}
    {\normalfont\LARGE\bfseries}{\thesection}{0pt}{\hspace{1em}}[\vspace{5pt}]

\setlist[description]{leftmargin=4.3cm,labelindent=1cm,rightmargin=2.5cm}

\begin{document}

\title{\vspace{1cm}
    \Large\scshape\bfseries
    The fermionic DGFF and its scaling limit logCFT
    \\
    \vspace{30pt}
        }

\author{
	David Adame$\,$-$\hspace{2pt}$Carrillo\footnote{\texttt{david.adame-carrillo@helsinki.fi}}
	$\ $ and $\ $
	Wioletta M. Ruszel\footnote{\texttt{w.m.ruszel@uu.nl}}\vspace{0.1cm}}

\date{\small\itshape
$^*$Department of Mathematics and Statistics, University of Helsinki
\\
$^\dagger$Mathematical Institute, Utrecht University
}

\maketitle

\vspace{30pt}

\begin{abstract}
In this paper, we identify the scaling limit of the fermionic discrete Gaussian free field (fDGFF) as a logarithmic conformal field theory (CFT) in two dimensions.
We first establish a one-to-one correspondence between the space of local observables of the fDGFF and the space of local fields of the symplectic fermions CFT, a logarithmic CFT with central charge $c = - 2$.
This correspondence is meaningful in the sense that, when appropriately renormalised, the fDGFF correlation functions converge to corresponding CFT correlation functions in the scaling limit.
As an application to these results, we interpret (the scaling limit of) certain local observables in the uniform spanning tree and the Abelian sandpile model as local fields of the symplectic fermions.
\end{abstract}

\vfill
{\small
\noindent
Keywords:~%
{\it%
fermionic Gaussian free field,
symplectic fermions,
conformal field theory,
local fields, Virasoro structure,
discrete complex analysis,
scaling limit of corelation functions,
uniform spanning tree,
Abelian sandpile model.
}}

\vspace{20pt}

\pagenumbering{gobble}

\newpage
\tableofcontents

\newpage
\pagenumbering{arabic}
\setcounter{page}{2}

\setlength{\parskip}{5pt}
\section{Introduction}
\paragraph{Lattice models and (logarithmic) conformal field theory.}
In the past twenty-five years, substantial progress has been made in the rigorous understanding of (the scaling limit of) critical models of statistical mechanics in the framework of conformal field theory (CFT) in two dimensions.
Notable milestones include Kenyon's analysis of the height function of the dimer model ---\cite{Kenyon00}--- as well as the works of Chelkak, Hongler, Izyurov and Smirnov on the Ising model ---\cite{Smi10,Hon10,ChSm12,HoSm13,CHI15,CHI21}---.
In those works, the authors identify the scaling limit of the ---suitably renormalised--- correlation functions of certain observables in the corresponding lattice model as CFT correlation functions.
It is worth noting that, in some cases, the scaling limit of local observables can be made sense of random distributions ---\cite{Kenyon01,CGN-magnetisation_ising}--- but this is not always the case ---\cite{garban_kup}---.
These achievements constituted rigorous confirmation of a prediction made by the physicists Belavin, Polyakov and Zamolodchikov in the 1980s, asserting that probabilistic models at criticality can be described by conformal field theories ---\cite{belA,belB}---.
Despite the extensive body of work in the Physics literature inspired by this conjecture, most models of statistical mechanics continue to lack a complete rigorous treatment of their scaling-limit CFT picture.

In the breakthrough work \cite{HKV}, the authors address this conjecture from a new angle in the context of the Ising model and the discrete Gaussian free field.
The key insight is the following:
employing techniques of discrete complex analysis,
the arguments typically used in the CFT literature can be reproduced at the discrete level in models that are sufficiently integrable.
This way, one can render the space of local observables of the discrete model a representation of the relevant symmetry algebra.
This approach permits a more systematic and exhaustive analysis of the space of local observables of a discrete model.
This program was further applied to fermionic observables of the double-dimer model in~\cite{AdameCarrillo-discrete_symplectic_fermions}.
Furthermore, in~\cite{ABK-DGFF_local_fields}, the authors used this approach to establish a one-to-one correspondence between local observables of the DGFF and local fields in the free boson CFT---hence finding for the first time the full structure of a CFT in a lattice model. 

In certain critical models of statistical mechanics, the correlation functions 
of specific observables exhibit logarithmic divergencies when their insertion points are brought together.
While this behaviour appears to spoil the conformal covariance of the theory, it can still be incorporated in the framework of CFT---theories featuring this phenomenon are referred to as logarithmic conformal field theories (logCFTs).
This class of models has recently attracted considerable attention.
Notably, Ruelle and others found natural observables in the Abelian sandpile model that can be accommodated in a logCFT---see \cite{Ruelle-SM in the large} for a review.
Also critical (site Bernoulli) percolation has seen a notorious breakthrough in its description as a logarithmic CFT in the series of works \cite{camia_perco,CF-percolation_log_OPEa,CF-percolation_log_OPEb,CF-percolation_energy}.
\pagebreak

One of the most well-understood examples of a logCFT is the symplectic fermions, which was first described by Kausch in \cite{Kau00}---see \cite{AdaCar-symplectic_fermions} for a review of the theory in general domains.
This logCFT has been connected to the uniform spanning tree ---\cite{LPW-UST_CFT}--- and the Abelian sandpile model ---\cite{Ruelle-SM in the large}---.
In turn, certain correlation functions in these two probabilistic models are well-known to be computable via a discrete model known as the fermionic discrete Gaussian free field (fDGFF)---see, e.g., \cite{BCHS-fDGFF_forests,CCRR-fDGFF_sandpile_UST}.

For a given mesh size $\meshsize>0$, the fermionic DGFF on a discrete domain $\ddomain{\meshsize}\subset\meshsize\Z^2$ can be roughly understood as generated by two fermionic fiedlds $\xi$ and $\theta$ whose correlation functions are discrete harmonic as long as the fields are inserted away from each other.
The fermionic nature of the fields yields a determinantal structure of the multipoint correlation functions of the model.
See Section~\ref{subsec: fDGFF def} for the precise definition.

\vspace{-10pt}
\paragraph{Contribution and discussion.}
In this paper, inspired by \cite{HKV, ABK-DGFF_local_fields}, we aim at conceiving the fermionic discrete Gaussian free field as a discretised version of the symplectic fermions conformal field theory.
The way this statement is made precise is by establishing a one-to-one correspondence between local observables of the fDGFF and local fields in the CFT such that ---when appropriately renormalised--- the discrete correlation functions of the former converge to the CFT correlation functions of the latter in the scaling limit.

Our first main result concerning the discrete-local-observable/CFT-local-field is stated algebraically.
To that end, let us briefly review some algebraic facts of two-dimensional CFT.
The conformal symmetries of the theory are encoded into the space of fields by the Virasoro algebra,
that is, the infinite-dimensional Lie algebra
\begin{align*}
    \Vir
    \,\coloneqq\,
    \Big( \, \bigoplus_{n\in \mathbb{Z}} \,\mathbb{C}\,\AlgL{n} \Big)
    \oplus
    \mathbb{C} \, \AlgC
\end{align*}
equipped with the Lie brackets
\begin{align*}
    \big[ \AlgL{n}, \AlgL{m} \big]
    \,=\,
    (n-m) \, \AlgL{n+m} + \delta_{n+m} \frac{n^3-n}{12} \, \AlgC\,,
    \mspace{70mu}
    \big[ \AlgC, \AlgL{n} \big] \,=\, 0,
\end{align*}
for all $n, m\in \mathbb{Z}$.
In particular, the space of fields of a two-dimensional CFT is a representation of (two commuting copies of) the Virasoro algebra.
In CFT, the central element $\AlgC$ acts as a fixed scalar $c$ multiple of the identity operator---in the model at hand, we have $c=-2$.

A characteristic feature of logCFTs is their more exotic Virasoro structure; unlike in standard QFT, the energy operator~$(\AlgL{0}+\AlgLBar{0})$ is not diagonalisable.
Geometrically, the operator~$(\AlgL{0}+\AlgLBar{0})$ is the generator of dilations---its (generalised) eigenvalues are called (generalised) scaling dimensions and they determine the behaviour of a given field under rescaling.

On the other hand, in the discrete setting, one should think of local observables of the fDGFF as polynomials on the fields $\xi$ and $\theta$ and their derivatives at a point.
However, at any given positive mesh size $\meshsize > 0$, these polynomials involve the field values on a finite neighborhood of the insertion point---see Section~\ref{subsec: local fields} for their definition.
This space can be rendered a $\Vir\oplus\overline\Vir$ representation adapting the techniques developed in \cite{AdameCarrillo-discrete_symplectic_fermions,ABK-DGFF_local_fields}.
Then, our first main result can be informally stated as follows.

\noindent
\textsc{\textbf{Theorem~1.}}
{\itshape
As representations of $\Vir\oplus\overline\Vir$, we have the isomorphism
\begin{align*}
\bigg\{
    \begin{array}{c}
    \text{local observables of} \\
    \text{the fermionic DGFF}
    \end{array}
\bigg\}
\; \cong \;
\bigg\{
    \begin{array}{c}
    \text{local fields of the} \\
    \text{symplectic fermions}
    \end{array}
\bigg\}\,.
\end{align*}
}

More precisely, we prove that those spaces are isomorphic as representations of the symplectic fermion symmetry algebra, which implies the above statement.
See Theorem~\ref{thm: isomorphism} for the precise statement.

The second main result concerns the scaling limit of correlation functions.
The definition of the discrete correlation functions of the fermionic DGFF can be found in Section~\ref{subsec: fDGFF def}; 
the CFT correlation functions are defined in Section~\ref{subsec: corr fun symplectic fermions}.

The setup of the scaling limit is as follows.
Let $\ddomain{\meshsize}\subset\meshsize\Z^2$ be discrete domains that approximate ---in an appropriate sense--- the simply-connected domain $\domain \subset \C$ as $\meshsize \downarrow 0$, and let $z_1^{\meshsize} , \ldots , z_n^{\meshsize}$ be the points in $\ddomain{\meshsize}$ that are closest to the fixed points $z_1 , \ldots , z_n \in \domain$ respectively.
We also fix a positive constant $\lambda > 0$, whose interpretation from the physics side we shall discuss after the statement of the theorem.
Our second main result can then be stated informally as follows.

\noindent
\textsc{\textbf{Theorem~2.}}
{\itshape
Let $\field_1 , \ldots , \field_n$ be symplectic-fermions local fields with generalised scaling \linebreak[3] dimensions $D_1 , \ldots , D_n$ respectively.
Let $F_i$ and $\tilde{F}_i$ be the fDGFF local observables corres\-ponding ---via Theorem~1--- to the local field $\field_i$ and $[(\AlgL{0}+\AlgLBar{0})-D_i]\field_i$, respectively.
Then, we have
\begin{align*}
    \frac{
    \Big\langle
    \Big(
    F_1(z^\meshsize_1) - \log ( \lambda \meshsize ) \tilde{F}_1(z^\meshsize_1)
    \Big)
    \!\cdots\!
    \Big(
    F_n(z^\meshsize_n) - \log ( \lambda \meshsize ) \tilde{F}_n(z^\meshsize_n)
    \Big)
    \Big\rangle^{\textnormal{fDGFF}}_{\ddomain{\meshsize}}
    \!\!\!\!\!\!\!\!
    }{
    \mathlarger{\mathlarger{\meshsize}}^{\,D_1}
    \;\cdots\;
    \mathlarger{\mathlarger{\meshsize}}^{\,D_n}
    }
    \mspace{4mu}
    \xrightarrow{\mspace{20mu} \meshsize \downarrow 0 \mspace{20mu}}
    \mspace{1mu}
    \Big\langle
    \field_1(z_1) \cdots \field_n(z_n)
    \Big\rangle^{\textnormal{SyFe}}_{\domain;\lambda}
\end{align*}
uniformly for $z_1, \ldots , z_n$ on compacts of $\domain$ and away from one another.
}

The constant $\lambda > 0$ should be interpreted as an arbitrary choice of a linear scale---there is no good reason to pick the square side length $\meshsize$ instead of the square diagonal length $\sqrt{2}\meshsize$ as the characteristic linear scale of our approximation $\ddomain{\meshsize}$.
This fact leads to the scaling limit not being unique, but rather unique up to a choice of a linear scale.
Indeed, notice that we added a dependence on $\lambda$ in the CFT correlation functions.
This non-uniqueness is explained in CFT by the existence of non-trivial self-isomorphism of the space of fields of the symplectic fermions;
see \cite{AdaCar-symplectic_fermions} for more details.
We give explicit examples of this phenomenon in Example~\ref{ex: 1pt fctn log partner} and Section~\ref{subsubsec: ASM}.

These results have applications to the the uniform spanning tree (UST) and the Abelian sandpile model (ASM) via their connection to the fermionic DGFF.
In particular, any local observable in those models whose correlation functions can be computed with fDGFF correlation functions can be identified with a local field in the symplectic fermions CFT.
Moreover, the (appropriately renormalised) scaling limit of its correlation functions can be explicitly computed.
In Section~\ref{subsec: apps to stat mech}, we discuss these matters for the following observables:
the degree field and the open-horizontal-edge field in the UST,
and the height-one field and the dissipation field in the ASM.

\vspace{-10pt}
\paragraph{Organisation of the paper.}
In Section~\ref{sec:background}, we provide the specific background that is required for the precise formulation of our results.
In particular, we discuss the symplectic fermions CFT and present the tools of discrete complex analysis that we employ in our treatment.
Section~\ref{sec: local_fields} is devoted to the definition of the fDGFF and its space of local fields.
Furthermore, we define certain operators ---the current modes--- in the space of local fields that are key in our approach.
Our two main results ---Theorems~1~and~2 above--- are stated in precise terms and discussed in Section~\ref{sec:results}.
We complement these results with concrete applications to two probabilistic models---namely, the Abelian sandpile model and the uniform spanning tree.
Most proofs throughout the text are postponed to Section~\ref{sec: proofs}.
\vfill
\noindent
\textbf{Acknowledgments:}
The authors thank Kalle~Kyt\"ol\"a for many insightful discussions and Aapo~Pajala for the code to make certain ---otherwise long and tedious--- computations in this paper.
The authors also enjoyed and benefited from discussing topics  related to this work with Philippe~Ruelle and Dirk~Schuricht. Furthermore the authors are thankful to Federico Camia for providing useful references and for the discussions.
DAC is grateful for the hospitality of the Utrecht University during his visits.
DAC was supported by the Research Council of Finland;
project 346309:~Finnish Centre of Excellence in Randomness and Structures (FiRSt).
Part of this research was performed while the authors were visiting the Institute for Pure and Applied Mathematics (IPAM), which is supported by the National Science Foundation (Grant No.~DMS-1925919).
WMR is supported by OCENW.KLEIN.083 grant and the Vidi grant VI.Vidi.213.112 from the Dutch Research Council (NWO).

\newpage
\section{Background}
\label{sec:background}
In this section, we review notions that, while possibly familiar to some readers, are essential in the derivation of our results.
In addition, this section is used to establish the notation used throughout the remainder of the text.

In Section~\ref{subsec: symplectic fermions} we review the symplectic fermions CFT, this is based on \cite{AdaCar-symplectic_fermions}.
In particular, in Section~\ref{subsec: log fock space}, we give an explicit construction of its space of local fields as the logarithmic Fock space and, in Section~\ref{subsec: corr fun symplectic fermions}, we characterise its correlation functions on general domains of the complex plane.
The tools of discrete complex analysis that we use in our analysis are presented in Section~\ref{subsec: tools discret C-anal}.

\vspace{-10pt}
\subsection{The symplectic fermions CFT}
\label{subsec: symplectic fermions}
The symplectic fermions is a conformal field theory that possesses a symmetry algebra that is more elementary than the Virasoro algebra.
Its symmetry algebra we call the \term{symplectic fermions algebra}---it is the vector space
\begin{align*}
    \FullSymFer
    \,\coloneqq\,
    \left( \bigoplus_{k\in\Z} \C\, \AlgEta{k} \right)
    \oplus
    \left( \bigoplus_{k\in\Z} \C\, \AlgChi{k} \right)
    \oplus
    \C\, \AlgK
    \oplus
    \left( \bigoplus_{k\in\Z} \C\, \AlgEtaBar{k} \right)
    \oplus
    \left( \bigoplus_{k\in\Z} \C\, \AlgChiBar{k} \right)
    \oplus
    \C\, \AlgKBar
\end{align*}
equipped with the Lie superbrackets
\begin{align*}
    \big\{ \AlgEta{k},\AlgChi{\ell}\big\}
    \,=\,
    k\, \delta_{k+\ell} \, \AlgK\,, \phantom{\Big\vert}
    \mspace{100mu}
    \big\{ \AlgEta{k},\AlgEta{\ell} \big\}
    \,=\,
    \big\{ \AlgChi{k},\AlgChi{\ell} \big\}
    \,=\, & \,
    0 
    \,=\,
    \big[\, \AlgK , \FullSymFer \,\big]\,,
    \\
    \phantom{\bigg\vert}
    \big\{ \AlgEtaBar{k},\AlgChiBar{\ell}\big\}
    \,=\,
    k\, \delta_{k+\ell} \, \AlgKBar\,, \phantom{\Big\vert}
    \mspace{100mu}
    \big\{ \AlgEtaBar{k},\AlgEtaBar{\ell} \big\}
    \,=\,
    \big\{ \AlgChiBar{k},\AlgChiBar{\ell} \big\}
    \,=\, & \,
    0 
    \,=\,
    \big[\, \AlgKBar , \FullSymFer \,\big]\,,
\end{align*}
and
\begin{align*}
    \big\{ \AlgChi{k},\AlgChiBar{\ell} \big\}
    \,=\,
    \big\{ \AlgChi{k},\AlgEtaBar{\ell} \big\}
    \,=\,
    \big\{ \AlgEta{k},\AlgChiBar{\ell} \big\}
    \,=\,
    \big\{ \AlgEta{k}, & \,\AlgEtaBar{\ell} \big\}
    \,=\,
    0
\end{align*}
for $k,\ell \in \mathbb{Z}$.

\vspace{-10pt}
\subsubsection{Space of fields}
\label{subsec: log fock space}
The space of fields of the theory is constructed as a Fock space of the symmetry algebra.
We first consider the associative algebra
\begin{align*}
\FullEnvSymFer
\,=\,
\UniEnve \big(\, \SymFer \oplus \SymFerBar \,\big)
\,\big/\,
(\, \AlgK-1\,,\,\AlgKBar-1 \,)
\end{align*}
where the generators $\AlgK$ and $\AlgKBar$ are set to $1$.
The \term{logarithmic Fock space} of the symplectic fermions is then defined as the quotient
\begin{align*}
\FullFock
\, \coloneqq \,
\FullEnvSymFer
\, \big/ \,
\big(\, \AlgEta{0} - \AlgEtaBar{0}\,,\; \AlgChi{0} - \AlgChiBar{0}\,,\; \AlgEta{k}\,,\; \AlgChi{k}\,,\; \AlgEtaBar{k}\,,\; \AlgChiBar{k} \,\colon\, k>0  \,\big)
\end{align*}
of the associative algebra $\FullEnvSymFer$ by the left ideal generated by $\AlgEta{k}$, $\AlgChi{k}$, $\AlgEtaBar{k}$ and $\AlgChiBar{k}$ with positive indices $k$, and the elements $\AlgEta{0} - \AlgEtaBar{0}$ and $\AlgChi{0} - \AlgChiBar{0}$.
Note that $\FullFock$ is canonically a representation of $\FullSymFer$ and a ($\FullEnvSymFer$)-module.

The logarithmic Fock space serves as the space of local fields of the symplectic fermions CFT; in Section~\ref{subsec: corr fun symplectic fermions} below we will construct the correlation functions of this theory.
There is a number of fields that play distinguished roles in the construction.
First, we have the \term{ground states}.
These are the bosonic states
\begin{align*}
    \phantom{\Big\vert}
    \Ground
    \,&\,\coloneqq \,
    [\,\AlgChi{0}\AlgEta{0}\,]\,,
    & \textnormal{(\textbf{identity field})}
    \\
    \phantom{\Big\vert}
    \GroundPartner
    \,&\,\coloneqq \,
    [\,1\,]\,,
    & \textnormal{ (\textbf{logarithmic partner} of $\Ground$)}
\end{align*}
and the \term{ground fermions}
\begin{align*}
    \GroundEta
    \,\coloneqq \,
    -[\,\AlgEta{0}\,]\,,
    \mspace{50mu}
    \textnormal{ and }
    \mspace{50mu}
    \GroundChi
    \,\coloneqq \,
    -[\,\AlgChi{0}\,]\,.
\end{align*}
Further, we have the \term{holomorphic currents}
\begin{align*}
    \HolCurrentChi
    \,\coloneqq\,
    \AlgChi{-1}\Ground\,,
    \mspace{50mu}
    \textnormal{ and }
    \mspace{50mu}
    \HolCurrentEta
    \,\coloneqq\,
    \AlgEta{-1}\Ground\,,
\end{align*}
and the \term{antiholomorphic currents}
\begin{align*}
    \AntiHolCurrentChi
    \,\coloneqq\,
    \AlgChiBar{-1}\Ground\,,
    \mspace{50mu}
    \textnormal{ and }
    \mspace{50mu}
    \AntiHolCurrentEta
    \,\coloneqq\,
    \AlgEtaBar{-1}\Ground\,.
\end{align*}

\noindent
The vector space $\FullFock$ admits the Poincar\'e--Birkhoff--Witt-type basis 
\begin{align}\label{eq: fock basis}
\AlgEta{-k_r}\cdots\AlgEta{-k_2}\AlgEta{-k_1}
\AlgChi{-\ell_s}\cdots\AlgChi{-\ell_2}\AlgChi{-\ell_1}
\AlgEtaBar{-\bar{k}_{\bar{r}}}\cdots\AlgEtaBar{-\bar{k}_2}\AlgEtaBar{-\bar{k}_1}
\AlgChiBar{-\bar{\ell}_{\bar{s}}}\cdots\AlgChiBar{-\bar{\ell}_2}\AlgChiBar{-\bar{\ell}_1}
\GroundPartner
\end{align}
\vspace{-25pt}
\begin{align}
\nonumber
\text{ with } \qquad
& \phantom{\Big\vert}r,s,\bar{r},\bar{s}\in \Znneg \\
\nonumber
\text{ and the orderings } \qquad 
& \phantom{\Big\vert}{0 \leq k_1 < k_2 < \cdots < k_r}\,,\\
\nonumber
& \phantom{\Big\vert}{0 \leq \ell_1 < \ell_2 < \cdots < \ell_s}\,,\\
\nonumber
& \phantom{\Big\vert}{0 < \bar{k}_1 < \bar{k}_2 < \cdots < \bar{k}_{\bar{r}}}\,, \text{ and }\\
\nonumber
& \phantom{\Big\vert}{0 < \bar{\ell}_1 < \bar{\ell}_2 < \cdots < \bar{\ell}_{\bar{s}}}\,.
\end{align}

\begin{rmk}\label{rmk: subrepresentations}
Although the space~$\FullFock$ is not an irreducible representation of $\FullSymFer$, it is cyclically generated by the action of the generators of $\FullSymFer$ on the vector~$\GroundPartner$.
From the explicit basis~\eqref{eq: fock basis} and the anticommutation  relations ---Lie superbrackets--- of the generators,
one can argue that the only proper subrepresentations of $\FullFock$ are the ones cyclically generated by the other ground states ---$\Ground$, $\GroundChi$, and $\GroundEta$---.
\hfill$\diamond$
\end{rmk}

The logarithmic Fock space can be rendered a Virasoro algebra representation via the so-called Sugawara construction.
That is the object of the following lemma, which can be found as a special case of \cite[Proposition 5.1]{Kac-VAs_for_beginners}; a computational proof can be found in~\cite[Theorem~5.2]{AdameCarrillo-discrete_symplectic_fermions}.
Let $[\,\cdot,\cdot\,]$ denote the usual commutator of linear operators, that is $[A,B]\coloneqq AB-BA$.

\begin{lemma}
\label{lemma: sugawara}
The linear operators on $\FullFock$ given by the formal sums
\begin{align*}
\SugL{n}
\coloneqq
\sum_{k \,\geq\, n/2}
\AlgChi{n-k}\AlgEta{k}
-
\sum_{k \,<\, n/2}
\AlgEta{k}\AlgChi{n-k}
\mspace{30mu}
\text{ and }
\mspace{30mu}
\SugLBar{n}
\coloneqq
\sum_{k \,\geq\, n/2}
\AlgChiBar{n-k}\AlgEtaBar{k}
-
\sum_{k \,<\, n/2}
\AlgEtaBar{k}\AlgChiBar{n-k}\,,
\end{align*}
for $n\in\Z$, are well-defined and constitute two commuting representations of the Virasoro algebra with central charge $-2$.
That is, they satisfy
\begin{align*}
    & \,
    \big[ \,\SugL{n}, \SugL{m} \big]
    \,=\,
    (n-m) \, \SugL{n+m} +  \frac{c}{12} \,(n^3-n) \, \delta_{n+m}  \, \id_{\FullFock}\,,
    \\
    \phantom{\Bigg\vert}
    & \,
    \big[ \,\SugLBar{n}, \SugLBar{m} \big]
    \,=\,
    (n-m) \, \SugLBar{n+m} +  \frac{c}{12} \,(n^3-n) \, \delta_{n+m}  \, \id_{\FullFock}\,,
    \mspace{20mu} \text{and}
    \\
    & \,
    \big[ \,\SugL{n}, \SugLBar{m} \big]
    \,=\,
    0
\end{align*}
with $c=-2$.
\end{lemma}
\vspace{-5pt}

The operators $\SugL{n}$ are called \term{holomorphic Virasoro modes} and the operators $\SugLBar{n}$ are called \term{antiholomorphic Virasoro modes}.

\vspace{-10pt}
\subsubsection{Correlation functions}
\label{subsec: corr fun symplectic fermions}
In a field theory, the physically significant quantities are correlation functions.
Given ${n\geq1}$ fields~$\field_1,\ldots,\field_n$ in a field theory, their correlation function is a $\C$-valued  multivariable function denoted by
\begin{align*}
    (z_1,\ldots,z_n)
    \longmapsto
    \Big\langle
    \field_1(z_1)\cdots\field_n(z_n)
    \Big\rangle
    \,
\end{align*}
where the points $z_i$ live in some ambient space where the theory takes place.
In this section, we present the characterisation of these functions for the symplectic fermions CFT in general domains of $\C$---we refer the reader to \cite{AdaCar-symplectic_fermions} for a more detailed discussion.
For the rest of the paper, we take $\domain\subset\C$ to be a (non-empty) proper, simply-connected open subset of the complex plane.

A notable feature of this theory is the existence of a family ---parametrised by a complex number $\thectt\in\C$--- of non-trivial isomorphisms $\FullFock\to\FullFock$.
This fact suggests that the correlation functions of the theory cannot be unique, but rather unique up to the action of such isomorphism on the space of fields.
Although this fact seems a priori burdensome, we argue that it is a distinctive and celebrated property of logarithmic theories---we will see the necessity of this ambiguity when considering scaling limits in Section~\ref{subsec: scaling limit}, in which an arbitrary choice of scale appears naturally.
An instance of this feature will be discussed in the context of a particular probabilistic model ---the Abelian sandpile--- in Example~\ref{ex: dissipation field}.

Fixing a constant $\thectt\in\C$, the \term{correlation functions} of the symplectic fermions are a collection of linear maps
\begin{align*}
    \Big\{\;
    \bigSFCorrFun{\domain}{\thectt}{\cdots}
    \; \colon \;
    \FullFock^{\,\otimes n} \longrightarrow \AnalyFun{\Conf{n}{\domain}}{\C}
    \,\Big\}_{n\in\Zpos}\,,
\end{align*}
where $\AnalyFun{\Conf{n}{\domain}}{\C}$ is the set of $\C$-valued real-analytic functions on the \term{$n$-point configuration space}
\begin{align*}
	\Conf{n}{\domain}
	\coloneqq
	\Big\{
	(z_1,\ldots,z_n)\in\domain^n \,\big\vert\, z_i \neq z_j \text{ for } i\neq j
	\Big\}\,
\end{align*}
of the domain~$\domain$.

The correlation functions of the symplectic fermions can be obtained following a bootstrap procedure---this is exploited in the Theorem~\ref{thm: characterization of CF} below to characterise them.
The idea of the bootstrap approach is to rely on general consistency conditions that are valid for general two-dimensional CFTs ---rather than on a specific Lagrangian--- to construct the correlation functions of the theory.
To start the bootstrap approach, we fix the correlation functions of the ground states.
In particular, following the dictates of CFT ---see \cite{AdaCar-symplectic_fermions} for more details--- a sensitive choice is
\begin{align*}
    \BigSFCorrFun{\domain}{\thectt}{\GroundChi(z) \GroundEta(w)}
    \,=\,
    4\pi\,\Green_\domain(z,w)\,,
\end{align*}
where $\Green_\domain$ is Green's function of the Laplacian operator~$\Delta=\partial^2_x+\partial^2_y$ with Dirichlet boundary conditions.
We take the normalisation of $\Green_\domain$ such that we have
\begin{align*}
    \Green_\domain(z,w)
    \,=\,
    - \frac{\log \vert\, z-w \,\vert}{2\pi}
    + \HarmOfGreen_\domain(z,w)
\end{align*}
with $\HarmOfGreen_\domain \colon \domain^2\to\R$ harmonic on both variables separately.
Then, the rest of correlation functions can be obtained by identifying the fields that carry the symplectic-fermion algebraic structure in their Fourier modes.
These fields are the currents.

In what follows, we use the Wirtinger derivatives $\partial$ and $\partialBar$.
That is, for a function of complex variable $ z = x + \ii y$, with $x,y\in\R$, we have $\partial_z = \partial_x - \ii \partial_y$ and $\partialBar_z = \partial_x + \ii \partial_y$.



\setlist[description]{leftmargin=2.6cm,labelindent=0.1cm,rightmargin=0cm}
\begin{thm}[{\cite[Theorem 3.1]{AdaCar-symplectic_fermions}}]\label{thm: characterization of CF}
For each complex number $\thectt\in\C$, there exists a unique collection
of correlation functions with the following properties:
\begin{description}[style=sameline]

	\item[(FER)]
	The correlation functions of the ground fermions are
	\begin{align*}
		\BigSFCorrFun{\domain}{\thectt}{\GroundChi(z_1)\,\GroundEta(w_1)\cdots\GroundChi(z_n)\,\GroundEta(w_n)}
		\, = \,
		(4\pi)^n
		\displaystyle\sum_{\sigma\in\mathcal{S}_n} \sgntr(\sigma)
		\displaystyle\prod_{i=1}^n
		\Green_\domain\big(z_i,w_{\sigma(i)}\big)
	\end{align*}
	for $(z_1,\ldots,z_n,w_1\ldots,w_n)\in\Conf{2n}{\domain}$,
	they vanish when the number of
	\linebreak[4]
	$\GroundChi$-inser\-tions and $\GroundEta$-insertions do not coincide,
	and they satisfy
	\begin{align*}
		\BigSFCorrFun{\domain}{\thectt}{\cdots\GroundChi(z)\,\GroundEta(w)\cdots}
		\,=
		-\,\BigSFCorrFun{\domain}{\thectt}{\cdots\GroundEta(w)\,\GroundChi(z)\cdots}\,.
	\end{align*}

	\item[(DER)]
	For the Virasoro modes~$\SugL{-1}$ and $\SugLBar{-1}$ we have
	\begin{align*}
		\BigSFCorrFun{\domain}{\thectt}{\big[\SugL{-1}\field\big](z)\cdots}
		\  = \ 
		\partial_z\,\BigSFCorrFun{\domain}{\thectt}{\field(z)\cdots}	
	\end{align*}
	and
	\begin{align*}
		\;\BigSFCorrFun{\domain}{\thectt}{\big[\mspace{1mu}\SugLBar{-1}\field\big](z)\cdots}
		\  = \ 
		\partialBar_z\,\BigSFCorrFun{\domain}{\thectt}{\field(z)\cdots}\,,
	\end{align*}
	for any field~$\field\in\FullFock$.

\item[(CUR)] 
Within correlation functions, we have
\begin{align*}
    \HolCurrentChi(z)\field(w)
    \,=\,
    \sum_{k\in\Z}\,
    \frac{\,\big[\AlgChi{k}\field\big](w)\,}{\,(z-w)^{k+1}}\,,
    \mspace{20mu}
    &
    \mspace{59mu}
    \AntiHolCurrentChi(z)\field(w)
    \,=\,
    \sum_{k\in\Z}\,
    \frac{\,\big[\AlgChiBar{k}\field\big](w)\,}{\,(\overline{z}-\overline{w})^{k+1}}\,,
    \\
    \phantom{.}
    \\
    \HolCurrentEta(z)\field(w)
    \,=\,
    \sum_{k\in\Z}\,
    \frac{\,\big[\AlgEta{k}\field\big](w)\,}{\,(z-w)^{k+1}}\,,
    \mspace{20mu}
    &
    \text{ and }
    \mspace{20mu}
    \AntiHolCurrentEta(z)\field(w)
    \,=\,
    \sum_{k\in\Z}\,
    \frac{\,\big[\AlgEtaBar{k}\field\big](w)\,}{\,(\overline{z}-\overline{w})^{k+1}}
\end{align*}
for any field $\field\in\FullFock$ for small enough $\vert z - w \vert$.
\item[(LOG)]
We have
\begin{align*}
    \BigSFCorrFun{\domain}{\thectt}{ \GroundChi(z)\GroundEta(w) \cdots }
    \,=\, &\,
    \log \frac{1}{\vert z-w\vert^2}
    \BigSFCorrFun{\domain}{\thectt}{ \Ground(w) \cdots }
    \,-\,
    \BigSFCorrFun{\domain}{\thectt}{ \big[\GroundPartner + \thectt \, \Ground\big](w) \cdots }
    \\
    &\,
    \,+\, 
    o\big(\vert z - w\vert\big) \phantom{\Bigg\vert}
    \,
\end{align*}
as $\vert z - w \vert \to 0$.
\end{description}
\end{thm}

\noindent
A proof of this theorem can be found in \cite{AdaCar-symplectic_fermions}.
The following remarks are corollaries to this result and will be used later in the text.

\begin{rmk}
\label{rmk: integral formula}
Consider, for $i=1,\ldots,n$, the $\FullFock$ elements
\begin{align*}
\field_i
\,\coloneqq\,
\AlgChi{-k_{\alpha_i}^{(i)}}\cdots\AlgChi{-k_1^{(i)}}
\AlgEta{-\ell_{\beta_i}^{(i)}}\cdots\AlgEta{-\ell_1^{(i)}}
\AlgChiBar{-\bar{k}_{\bar\alpha_i}^{(i)}}\cdots\AlgChiBar{-\bar{k}_1^{(i)}}
\AlgEtaBar{-\bar\ell_{\bar\beta_i}^{(i)}}\cdots\AlgEtaBar{-\bar\ell_1^{(i)}}
\GroundPartner
\in \FullFock\,,
\end{align*}
in the basis~\eqref{eq: fock basis}.
By property \textbf{(CUR)}, their correlation functions can be obtained as
\begin{align*}
    &
    \BigSFCorrFun{\domain}{\thectt}{
    \field_1(z_1)
    \,\cdots\,
    \field_n(z_n)
    }
    \,=\,
    \oint \!\cdots\! \oint
    \prod_{i=1}^n
    \Bigg(
    \prod_{a=1}^{\alpha_i}
    \frac{\cd \zeta_{i;a}^{\chi}}{2\pi\ii}
    \prod_{b=1}^{\beta_i}
    \frac{\cd \zeta_{i;b}^{\eta}}{2\pi\ii}
    \prod_{\bar a=1}^{\bar\alpha_i}
    \frac{\cd \overline{\zeta}_{i;\bar a}^{\bar \chi}}{\overline{2\pi\ii}}
    \prod_{\bar b=1}^{\bar\beta_i}
    \frac{\cd \overline{\zeta}_{i;\bar b}^{\bar \eta}}{\overline{2\pi\ii}}
    \Bigg)
    \times
    \\
    &
    \mspace{60mu}
    \times
    \prod_{i=1}^n
    \Bigg(
    \prod_{a=1}^{\alpha_i}
    \, \big( \zeta_{i;a}^{\chi} - z_i \big)^{-k_a^{(i)}}
    \prod_{b=1}^{\beta_i}
    \, \big( \zeta_{i;b}^{\eta} - z_i \big)^{-\ell_b^{(i)}}
    \prod_{\bar a=1}^{\bar\alpha_i}
    \, \Big( \overline{ \zeta_{i;\bar a}^{\bar \chi} - z_i}\Big)^{-\bar{k}_{\bar a}^{(i)}}
    \prod_{\bar b=1}^{\bar\beta_i}
    \, \Big(\overline{\zeta_{i;\bar b}^{\bar\eta} - z_i}\Big)^{-\bar{\ell}_{\bar b}^{(i)}}
    \Bigg)
    \times
    \\
    &
    \times\!
    \bigg\langle
    \HolCurrentChi\big( \zeta_{1;1}^{\chi} \big) \cdots \HolCurrentChi\big( \zeta_{1;\alpha_1}^{\chi} \big)
    \HolCurrentEta\big( \zeta_{1;1}^{\eta} \big) \cdots \HolCurrentEta\big( \zeta_{1;\beta_1}^{\eta} \big)
    \AntiHolCurrentChi\big( \zeta_{1;1}^{\bar \chi} \big) 
    \cdots
    \AntiHolCurrentChi\big( \zeta_{1;\bar\alpha_1}^{\bar \chi} \big)
    \AntiHolCurrentEta\big( \zeta_{1;1}^{\bar \eta} \big)
    \cdots
    \AntiHolCurrentEta\big( \zeta_{1;\bar\beta_1}^{\bar \eta} \big)\,
    \GroundPartner(z_1)
    \cdots\phantom{\Bigg\vert}
    \\
    &
    \mspace{10mu}
    \cdots
    \HolCurrentChi\big( \zeta_{n;1}^{\chi} \big) \cdots \HolCurrentChi\big( \zeta_{n;\alpha_n}^{\chi} \big)
    \HolCurrentEta\big( \zeta_{n;1}^{\eta} \big) \cdots \HolCurrentEta\big( \zeta_{n;\beta_n}^{\eta} \big)
    \AntiHolCurrentChi\big( \zeta_{n;1}^{\bar \chi} \big) 
    \cdots
    \AntiHolCurrentChi\big( \zeta_{n;\bar\alpha_n}^{\bar \chi} \big)
    \AntiHolCurrentEta\big( \zeta_{n;1}^{\bar \eta} \big)
    \cdots
    \AntiHolCurrentEta\big( \zeta_{n;\bar\beta_n}^{\bar \eta} \big)\,
    \GroundPartner(z_n)
    \bigg\rangle_{\domain;\thectt}^{\SFtag}
\end{align*}
where the integrals are performed along non-intersecting contours in the following way:
for each $i = 1 , \ldots , n$, the variables $\zeta_{i;a}^{\chi}$, $\zeta_{i;b}^{\eta}$, $\zeta_{i;\bar a}^{\bar \chi}$ and $\zeta_{i;\bar b}^{\bar \eta}$ are integrated along nested contours around the point $z_i$ in the radial order
\begin{align*}
    \zeta_{i;1}^{\chi} , \ldots , \zeta_{i;\alpha_i}^{\chi},
    \zeta_{i;1}^{\eta} , \ldots , \zeta_{i;\beta_i}^{\eta},
    \zeta_{i;1}^{\bar \chi} , \ldots , \zeta_{i;\bar\alpha_i}^{\bar \chi},
    \zeta_{i;1}^{\bar \eta} , \ldots , \zeta_{i;\bar\beta_i}^{\bar \eta}
\end{align*}
in the inward direction.
\hfill$\diamond$
\end{rmk}

\begin{ex}\label{ex: 1pt fctn log partner}
We have
\begin{align*}
    \BigSFCorrFun{\domain}{\thectt}{\GroundPartner(z)}
    \,=\,
    - \, 4\pi\,\HarmOfGreen_{\domain}(z,z) - \thectt \,,
\end{align*}
where $\HarmOfGreen_\domain$ is the harmonic part of the Green's function~$\Green_\domain$.
Note that this one-point function is related to the conformal radius~$\mathrm{R}(z;\domain)$ of the domain~$\domain$ from the point~$z\in\domain$ since we have $\mathrm{R}(z;\domain) = \exp\big(2\pi\HarmOfGreen_\domain(z,z)\big)$.
\hfill$\diamond$
\end{ex}

\begin{rmk}
\label{rmk: general wick formula}
The correlation functions of the ground states $\GroundChi$, $\GroundEta$, and $\GroundPartner$ satisfy the following formula:
Take three non-intersecting collections $\underline{z}  \coloneqq\{z_1,\ldots,z_n\}$, $\underline{w}  \coloneqq\{w_1,\ldots,w_n\}$ and $\underline{x}  \coloneqq\{x_1,\ldots,x_k\}$ of distinct points in the domain~$\domain$.
Let ${\Bij(\underline{ z}  ,\underline{ w}  ;\underline{ x}  )}$ denote the set of bi\-jec\-tions from the set~${\underline{ z}  \cup\underline{ x}  }$ onto the set~${\underline{w}  \cup\underline{ x}  }$.
We have
\begin{align*}
    \BigSFCorrFun{\domain}{\thectt}
    {\GroundChi(z_1)\GroundEta(w_1)&\cdots\GroundChi(z_n)\GroundEta(w_n)
    \GroundPartner(x_k)\cdots\GroundPartner(x_1)}
    \\
    & = \phantom{\bigg\vert}
    (-1)^k\sum_{b\in\Bij(\underline{z}  ,\underline{w}  ;\underline{x}  )}
    (-1)^b(-1)^{F_b}
    \prod_{\underset{b(y)\neq y}{y\in\underline{z}  \cup\underline{x}  }}
    \BigSFCorrFun{\domain}{\thectt}
    {\GroundChi(y)\GroundEta\big(b(y)\big)}
    \prod_{\underset{b(x)= x}{x\in\underline{x}  }}
    \BigSFCorrFun{\domain}{\thectt}
    {\GroundPartner(x)}\,,
\end{align*}
where $(-1)^b$ is the signature of the bijection $b$, and $F_b$ is the number of fixed points of $b$.
\hfill$\diamond$
\end{rmk}

\vspace{-10pt}
\subsection{Tools of discrete complex analysis}
\label{subsec: tools discret C-anal}
Let us start by defining the lattices on which our discrete complex analysis tools take place; see Figure~\ref{fig: lattices}.
For $\delta > 0$, we define
\begin{itemize}
    \item the \term{primary lattice}~$\meshsize\Zprimary\coloneqq\meshsize\Z^2$,
    \item  the \term{dual lattice}~$\meshsize\Zdual\coloneqq(\meshsize\Z+\frac{\meshsize}{2})^2$, and
    \item the \term{medial lattice}~$\meshsize\Zmedial\coloneqq\meshsize\Z\times(\meshsize\Z+\frac{\meshsize}{2})\cup(\meshsize\Z+\frac{\meshsize}{2})\times\meshsize\Z$.
\end{itemize}
Naturally, a medial vertex in $\meshsize\Z\times(\meshsize\Z+\frac{\meshsize}{2})$ we call a \term{vertical edge}, whereas a medial vertex in $(\meshsize\Z+\frac{\meshsize}{2})\times\meshsize\Z$ we call a \term{horizontal edge}.
Finally, we also define the \term{diamond lattice}~$\meshsize\Zdiamond\coloneqq\meshsize\Zprimary\cup\meshsize\Zdual$.

\begin{figure}[t!]
\centering
\begin{overpic}[scale=0.467, tics=10]{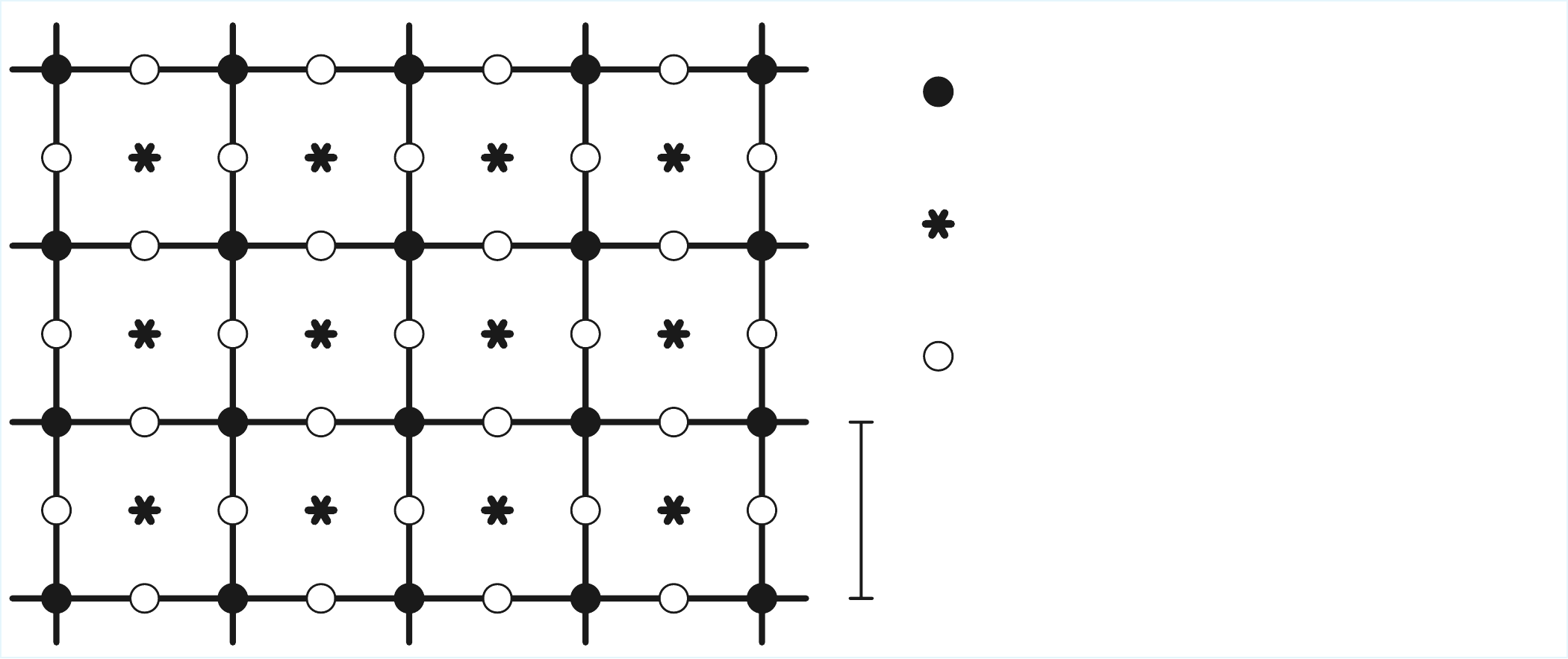}
    \put(64,35){\large $\meshsize\Zprimary$ \hspace{3pt} (\term{primary lattice)}}
    \put(64,26.5){\large $\meshsize\Zdual$ \hspace{3pt} (\term{dual lattice)}}
    \put(64,18){\large $\meshsize\Zmedial$ \hspace{3pt} (\term{medial lattice)}}
    \put(57,8.5){\large $\meshsize$}
\end{overpic}
\centering
\caption{Lattices involved in our tools of discrete complex analysis.} 
\label{fig: lattices}
\end{figure}
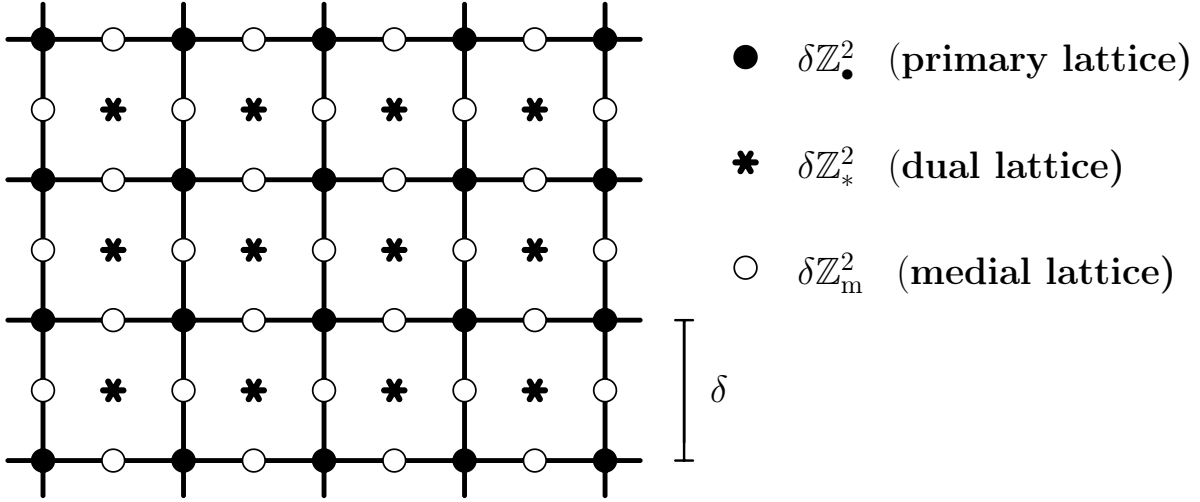

\vspace{-10pt}
\subsubsection{Discrete differential operators}
We start by defining discrete versions of the Wir\-tinger derivatives.
Given a function $f \colon \meshsize\Zdiamond \to \C$ on the diamond lattice, its \term{holomorphic} and \term{antiholomorphic discrete deri\-vatives} are the functions \mbox{$\ddee f , \ddeebar f \colon \meshsize\Zmedial \to \C$} on the medial lattice defined by the formulae
\begin{align*}
\phantom{\Bigg\vert}
\ddee f(\zd) := \; &
    \frac{f(\zd+\frac{\meshsize}{2})-f(\zd-\frac{\meshsize}{2})}{2} - \ii\, \frac{f(\zd+\ii\frac{\meshsize}{2})-f(\zd-\ii\frac{\meshsize}{2})}{2}\,,
    \quad \text{ and }
	\\
\phantom{\Bigg\vert}
\ddeebar f(\zd) := \; &
    \frac{f(\zd+\frac{\meshsize}{2})-f(\zd-\frac{\meshsize}{2})}{2} + \ii\, \frac{f(\zd+\ii\frac{\meshsize}{2})-f(\zd-\ii\frac{\meshsize}{2})}{2}\,.
\end{align*}
Similarly, for a function $f$ on the medial lattice, its holomorphic and antiholomorphic derivatives $\ddee f$ and $\ddeebar f$ are functions on the diamond lattice defined by the same formulae as above.

If a function $f$ satisfies $\ddeebar f(\zd) = 0$, it is said to be \term{discrete holomorphic} at~$\zd$.
Similarly, if it satisfies $\ddee f(\zd) = 0$, then it is said to be \term{discrete antiholomorphic} at $\zd$.

In our analysis, we also need versions of the holomorphic and antiholomorphic derivatives for functions defined only on the primary lattice.
For a function~$f \colon \meshsize\Zprimary \to \C$ on the primary lattice,
we define the functions $\pddee f, \pddeebar f \colon \meshsize\Zmedial \to \C$ by
\begin{align*}
    \phantom{-}
    \pddee f(\zd) \,\coloneqq\,
	f \Big( \zd+\frac{\meshsize}{2} \Big)
    -
    f \Big( \zd-\frac{\meshsize}{2} \Big)
    \mspace{30mu}
    \textnormal{ and }
    \mspace{30mu}
    \pddeebar f(\zd) \,\coloneqq\,
	f \Big( \zd+\frac{\meshsize}{2} \Big)
    -
    f \Big( \zd-\frac{\meshsize}{2} \Big)
\end{align*}
when $\zd\in\Zmedial$ is a horizontal edge,
and by
\begin{align*}
    \pddee f(\zd) \,\coloneqq\,
    -\ii \, \bigg( f \Big( \zd+\ii\frac{\meshsize}{2} \Big)
    - 
    f \Big( \zd-\ii\frac{\meshsize}{2} \Big) \bigg)
    \mspace{20mu}
    \textnormal{ and }
    \mspace{20mu}
    \pddeebar f(\zd) \,\coloneqq\,
    \ii \, \bigg( f \Big( \zd+\ii\frac{\meshsize}{2} \Big)
    -
    f \Big( \zd-\ii\frac{\meshsize}{2} \Big) \bigg)
\end{align*}
when $\zd\in\Zmedial$ is a vertical edge.

Lastly, we define the (combinatorially normalised) \term{discrete Laplacian} operator $\dlaplacian$.
For a function~$f$ on any of the lattices, the function~$\dlaplacian f$ on the same lattice is defined by
\begin{align*}
\dlaplacian f(\zd) := \; & 
    f(\zd + \meshsize) + f(\zd + \ii \meshsize) + f(\zd - \meshsize) + f(\zd - \ii \meshsize) - 4 \, f(\zd)\,.
\end{align*}
Naturally, we say that a function~$f$ is \textbf{discrete harmonic} at $\zd$ if it satisfies $\dlaplacian f ( \zd ) = 0$.

\begin{rmk}\label{rmk: factorisation laplacian}
The discrete Laplacian admits the following factorisations in terms of discrete differential operators:
\begin{align*}
       \ddee\ddeebar \,=\,  \ddeebar\ddee \,=\, \frac{1}{4}\,\dlaplacian \,,
\end{align*}
and, for a function~$f \colon \Zdiamond\to\C$,
\begin{align*}
    \ddee\pddeebar f(\zd) \,=\,  \ddeebar\pddee f(\zd) \,=\,
    \begin{cases} 
    \mspace{10mu}\frac{1}{2}\,\dlaplacian f(\zd) & \mspace{10mu} \textnormal{ if } \zd\in\Zprimary\,, \textnormal{ and } \\
    & \\
    \mspace{35mu} 0 & \mspace{10mu} \textnormal{ if } \zd\in\Zdual\,.
    \end{cases} 
\end{align*}
We will use extensively these factorisations in our computations in later sections.
\hfill$\diamond$
\end{rmk}

\vspace{-10pt}
\subsubsection{Discrete integration}
\label{subsubsec: discrete integration}
The second tool that we use in out treatment is the notion of discrete integration on the square lattice with unit mesh size ---that is, when we set $\meshsize=1$--- that was introduced in~\cite{HKV}.
For that purpose, consider the \term{corner lattice} $\Zcorner \coloneqq (\Z\pm\frac{1}{4})^2$,
whose vertices we call \term{corners}.
A sequence~$(c_0,\ldots,c_n)$ of consecutively nearest corners is called a \term{corner path}.

Given two functions~$f \colon \Zdiamond \to \C$ and $g \colon \Zmedial \to \C$ and a corner path~$\gamma=(c_0,\ldots,c_n)$, we define
\begin{align*}
    \dint_\gamma f(\zu_\diamond) g(\zu_\medial) \dd \zu
    \,\coloneqq\,
    \sum_{j=1}^n \, (c_j - c_{j-1})
    f(\zu_j^\diamond) g(\zu_j^\medial)\,, &
    \mspace{20mu}\textnormal{ and }
    \\
    \dint_\gamma f(\zu_\diamond) g(\zu_\medial)\dd{\overline{\zu}}
    \,\coloneqq\,
    \sum_{j=1}^n \, \overline{(c_j - c_{j-1})}
    f(\zu_j^\diamond) g(\zu_j^\medial) \,, &
\end{align*}
where the diamond vertex~$\zu_j^\diamond \in \Zdiamond$ and the medial vertex~$\zu_j^\medial \in \Zmedial$ are the unique ones that are closest to both $c_j$ and $c_{j-1}$ as shown in Figure~\ref{fig: discrete contour}.
Note that these formulae are also well-defined when one of the integrands ---either $f$ or $g$--- takes values in a complex vector space.

\begin{figure}[t!]
\centering
\begin{overpic}[scale=0.467, tics=10]{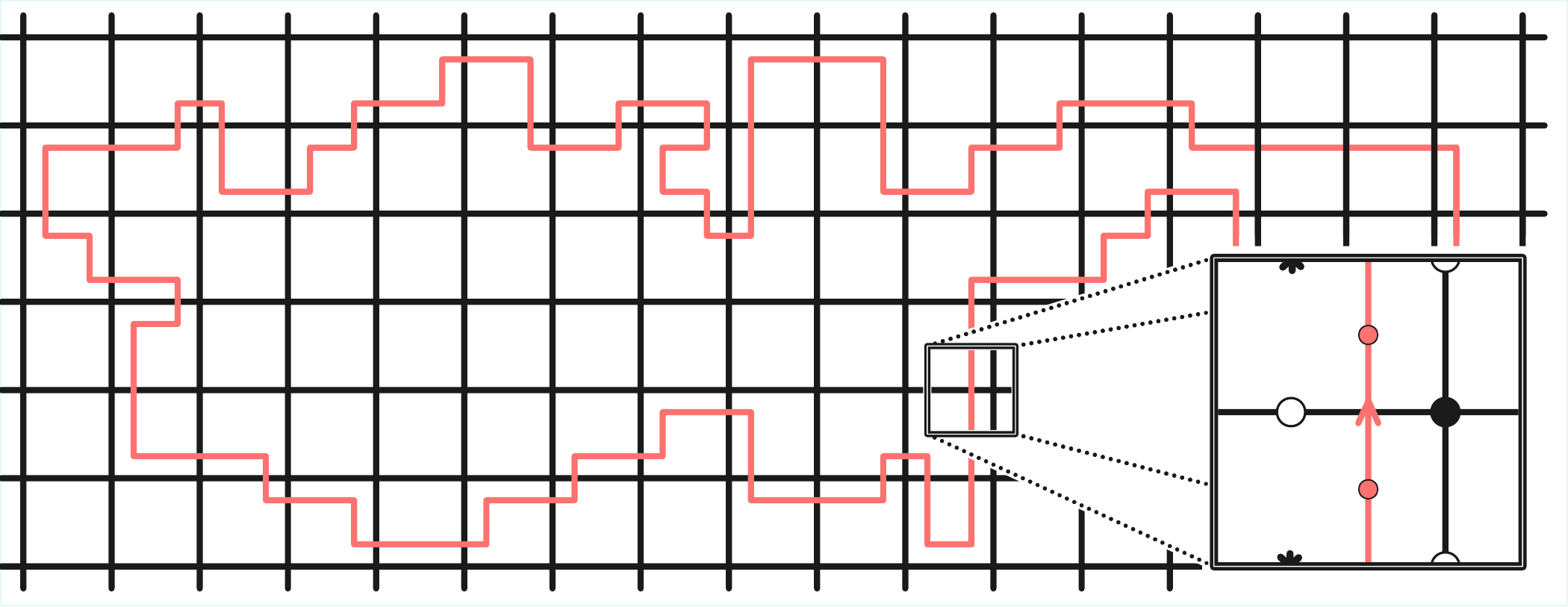}
	\put(78.8,14){$\zu^{\medial}_j$}
	\put(93,14){$\zu^{\diamond}_j$}
	\put(87.7,18.3){$c_j$}
	\put(87.7,5.4){$c_{j-1}$}
    \put(-2,34.5){
		\pgfsetfillopacity{1}
		\colorbox{white}{\centering \Large
		\parbox{0pt}{\pgfsetfillopacity{1}}\color{black} \hspace{0pt}$\Z^2$}} 
\end{overpic}
\centering
\caption{
A corner contour and, zoomed in, 
an integration step $(c_{j-1},c_j)$
\\
with its closest medial vertex $\zu_j^\medial \in \Zmedial$ and diamond vertex $\zu_j^\diamond \in \Zdiamond$.}
\label{fig: discrete contour}
\end{figure}

A \term{corner contour} is a corner path $\gamma=(c_0,\ldots,c_n)$ whose corners are all distinct except $c_0=c_n$.
The set of vertices in $\Zdiamond\cup\Zmedial$ that are enclosed by a corner contour~$\gamma$ is called the \term{interior} of $\gamma$, and it is denoted by~$\interior\gamma$.
We also use the notations~${\interiordiam \gamma \coloneqq \interior \gamma \cap \Zdiamond}$ and
${\interiormed \gamma \coloneqq \interior \gamma \cap \Zmedial}$.
A corner contour is said to be \term{positively oriented} if it encircles its interior in a counterclockwise fashion.
We use the notation~$\dsqint_{\gamma}$ for discrete integration along a positively-oriented corner contour~$\gamma$.

The following property of this notion of discrete integration are the ones that make it suitable for our analysis and will appear throughout the present text in our computations.
For a positively-oriented corner contour~$\gamma$,
and any two functions~$f \colon \Zdiamond \to \C$ and $g \colon \Zmedial \to \C$, we have the \term{discrete Stokes' formulae}
\begin{align*}
& \dsqint_\gamma f(\zu_\diamond)g(\zu_\medial)\dd{\zu}
= \; \phantom{-}\,  \ii \, \sum_{\zubis_\diamond\in\interiordiam\gamma} f(\zubis_\diamond)\,
\ddeebar g(\zubis_\diamond)
	  \,+\, \ii \, \sum_{\zubis_\medial\in\interiormed\gamma} \ddeebar f(\zubis_\medial)\, g(\zubis_\medial) \, , 
\\
& \dsqint_\gamma f(\zu_\diamond)g(\zu_\medial)\dd{\overline \zu}
= \;  -\, \ii \, \sum_{\zubis_\diamond\in\interiordiam\gamma} f(\zubis_\diamond)\,
\ddee g(\zubis_\diamond)
	  \,-\, \ii \, \sum_{\zubis_\medial\in\interiormed\gamma} \ddee f(\zubis_\medial)\, g(\zubis_\medial) \, . &
\end{align*}
Note that, from the discrete Stokes' formula, it follows that one can deform the corner contour of integration without changing the value of the integral as long as the integrands are discrete (anti)holomorphic in the appropriate domain.
\vspace{-10pt}



\vspace{-10pt}
\subsubsection{Discrete complex monomials}
\label{subsubsec: discrete monomials}
The last piece we need in our analysis are the discrete analogues of the complex monomials $z\mapsto z^n$ for all $n\in\Z$ introduced in~\cite{HKV} with a small modification---see~\cite{AdameCarrillo-discrete_symplectic_fermions, ABK-DGFF_local_fields} for the details and discussion of this modification.

\begin{prop}[{\cite[Proposition~2.1]{HKV}}]\label{prop: monomials}
There exists a unique family of $\C$-valued functions $\{\zu\mapsto \zu^{[n]}\}_{n\in\Z}$ on
$\Zdiamond\cup\Zmedial$ that satisfies the following properties:
\begin{enumerate}
\item For all $n\in\Z$, the function $\zu \mapsto \zu^{[n]}$ has the same square-grid symmetries
as the Laurent monomial $z \mapsto z^n$,
i.e., $(\ii \zu)^{[n]}=\ii^n \zu^{[n]}$ and $\overline{\zu}^{[n]} = \overline{\zu^{[n]}}$ for all
$\zu \in \Zdiamond\cup\Zmedial$.
\item For all $\zu\in\Zdiamond\cup\Zmedial$, $\zu^{[0]}=1$ and, for all $n\in\Z$, $\ddee \zu^{[n]} = n\,\zu^{[n-1]}$.
\item For each $\zu\in\Zdiamond\cup\Zmedial$, there exists an $N\in\N$ such that $\zu^{[n]}=0$ for all $n\geq N$.
\item For $n<0$, we have $\zu^{[n]} \to 0$ as $\vert \zu\vert\rightarrow \infty$.
\item The first negative-power monomial has the following explicit
failure of discrete holomorphicity near the origin
\begin{align*}
    \frac{1}{2\pi}\ddeebar \zu^{[-1]}
    =
    \frac{1}{2} \, \delta_{\zu,0}
    +
    \frac{1}{4}
    \sum_{\vert\zubis\vert = \frac{1}{2}} 
    \delta_{\zu,\zubis}
    +
    \frac{1}{8}\sum_{\vert\zubis\vert = \frac{1}{\sqrt{2}}}
    \mspace{-5mu}
    \delta_{\zu,\zubis}\,.
\end{align*}
\item \label{property: poles}
For any $n\geq 0$ and all $\zu\in\Zdiamond\cup\Zmedial$ we have $\ddeebar \zu^{[n]}=0$. 
For any $n<0$, we have $\ddeebar \zu^{[n]}=0$ except at finitely many
points~$\zu \in \Zdiamond\cup\Zmedial$.
\item For any $n,m\in\Z$, we have the discrete residue formula
\begin{align*}
    \frac{1}{2\pi\ii}
    \dsqint_{\gamma} \zu^{[n]}_\diamond \, \zu^{[m]}_\medial \, \dd \zu 
    \; = \; \delta_{n+m+1}\,,
\end{align*}
for any large enough positively-oriented corner contour $\gamma$ that encircles the origin.
\item For any $n \in \Z$, as~${|\zu| \to \infty}$, the discrete monomial has the asymptotics
\begin{align*}
\zu^{[n]} = \zu^n + o(|\zu|^n) \,,
\end{align*}
where in $\zu^n$ we interpret $\zu\in\C$ via the embedding  $\Z^2 \subset \C$.
\end{enumerate}
\end{prop}

\newpage
\section{Local fields of the fDGFF}
\label{sec: local_fields}
In this section, we introduce the discrete framework underlying our results.
More precisely, we construct a space of local fields for the fermionic discrete Gaussian free field ---see also \cite{abdesselam, CCRR-fDGFF_sandpile_UST}--- which we subsequently show to carry a representation of the symplectic fermions algebra in a meaningful way.

We start in Section 3.1 by defining the fermionic discrete Gaussian free field in general discrete domains of $\meshsize\Z^2$.
Then, in Section 3.2, we give the definition of local fields of the model.
Finally, in Section 3.3 we define certain operators on the space of local fields ---the current modes--- that enable the algebraic analysis of this space.

\vspace{-10pt}
\subsection{The fermionic discrete Gaussian free field}
\label{subsec: fDGFF def}
In this section we define the discrete model of interest, the \term{fermionic discrete Gaussian free field} (fDGFF).
For simplicity, we consider our model in discrete domains whose boundary is a polygonal Jordan curve $\gamma$ on $\C$ consisting of vertical and horizontal segments and with corners on $\meshsize\Z^2\subset \C$.
Then, we define a \term{discrete} (\term{Jordan}) \term{domain} as $\ddomain{\meshsize}\coloneqq \overline{\interior {\gamma}}\cap\meshsize\Z^2$ and its boundary as $\bdry\ddomain{\meshsize}\coloneqq \gamma\cap\meshsize\Z^2 \subset \ddomain{\meshsize}$.
We will consider discrete domains~$\ddomain{\meshsize}$ whose interior~$\ddomain{\meshsize}\setminus\bdry\ddomain{\meshsize}$, denoted by $\interior{\ddomain{\meshsize}}$, is path-connected via nearest-neighbour paths.

Given a set $X$, the \term{Grassmann algebra generated by} $X$ is the exterior algebra of the vector space~$\bigoplus_{x\in X} \C \, x$.
We use the notation $x_1x_2\cdots x_n$ for the product of $n$ algebra elements ---typically denoted as $x_1\wedge x_2\wedge\cdots\wedge x_n$ in the literature of exterior algebras---.
Recall that the product in the Grassmann algebra is anticommutative, which in our notation reads as $x_1 x_2 + x_2 x_1 = 0$ for all $x_1,x_2\in X$.

Let $\ddomain{\meshsize}\subset\meshsize\Z^2$ be a discrete domain.
Let $\Grass_{\ddomain{\meshsize}}[\fDGFFxi,\fDGFFtheta]$ denote Grassmann algebra generated by the set $\{\fDGFFxi(\zd) , \fDGFFtheta(\zd) \colon \zd\in\interior{\ddomain{\meshsize}}\}$; this space constitutes the observables of the fDGFF on $\ddomain{\meshsize}$.
Then, the (normalized) correlation functions of the fDGFF on $\ddomain{\meshsize}$ (with Dirichlet boundary conditions) are linear maps of the type
\begin{align*}
    \bigfDGFFCorrFun{\ddomain{\meshsize}}{\cdot}
    \;\colon\;
    \text{Gr}_{\ddomain{\meshsize}}[\fDGFFxi,\fDGFFtheta]
    \;\longrightarrow\;
    \C\,.
\end{align*}
In Remark~\ref{rmk: 2n pt function fDGFF} below, we give explicit formulae for all the corelation functions of the fDGFF on any discrete domain~$\ddomain{\meshsize}$.
However, the correlation functions are typically defined via the Berezin notion of integration.
To define it, we take a total order on the set $\ddomain{\meshsize}$ and consider the vector-space basis
\begin{align*}
    \bigg\{
    \prod_{\zd\in S_1}^{\rightarrow}
    \fDGFFxi(\zd)
    \prod_{\zdbis\in S_2}^{\rightarrow}
    \fDGFFtheta(\zdbis)
    \bigg\}_{S_1,S_2\subset \interior{\ddomain{\meshsize}}}
\end{align*}
where the arrow ($\rightarrow$) indicates that the order-dependent product is taken increasingly from \pagebreak left to right in the total order of $\ddomain{\meshsize}$.
Then, the \term{Berezin integral} over the Grassmann algebra $\Grass_{\ddomain{\meshsize}}[\fDGFFxi,\fDGFFtheta]$ is the linear map
\begin{align*}
    \iint \cdot \; \cd\fDGFFxi\cd\fDGFFtheta
    \;\colon\; 
    \Grass_{\ddomain{\meshsize}}[\fDGFFxi,\fDGFFtheta]
    \longrightarrow
    \C
\end{align*}
that satisfies
\begin{align*}
    \iint
    \Bigg(
    \prod_{\zd\in\interior{\ddomain{\meshsize}}} \!\!
    \fDGFFxi(\zd)\fDGFFtheta(\zd)
    \Bigg)
    \cd\fDGFFxi\cd\fDGFFtheta
    \, = \,1\,,
\end{align*}
and takes the value $0$ on the all basis elements except the one indexed by $S_1 = S_2 = \interior{\ddomain{\meshsize}}$ ---since it is proportional to $\prod_{\zd\in\ddomain{\meshsize}}\fDGFFxi(\zd)\fDGFFtheta(\zd)$---.

Then, to define the correlation functions of the fDGFF on $\ddomain{\meshsize}$, we consider the quadratic action
\begin{align*}
    \ActionfDGFF_{\ddomain{\meshsize}}
    \big[\fDGFFxi,\fDGFFtheta\big]
    \;\coloneqq\;
    \frac{1}{4\pi} \!\!
    \sum_{\zd\in\interior{\ddomain{\meshsize}}} \!\!
    \fDGFFxi(\zd)
    \Delta_{\ddomain{\meshsize}}^{\texttt{D}}
    \fDGFFtheta (\zd)
\end{align*}
where $\fDGFFtheta$ and $\fDGFFxi$ are conceived as maps $\interior{\ddomain{\meshsize}}\rightarrow\Grass_{\ddomain{\meshsize}}[\fDGFFxi,\fDGFFtheta]$,
and $\Delta_{\ddomain{\meshsize}}^{\texttt{D}}$ denotes the Laplacian operator with Dirichlet boundary conditions on $\ddomain{\meshsize}$, that is,
\begin{align*}
    \Delta_{\ddomain{\meshsize}}^{\texttt{D}}
    \fDGFFtheta (\zd)
    \,\coloneqq \!
    \sum_{\underset{\vert\zd-\zdbis\vert=\meshsize}{\zdbis\in\interior{\ddomain{\meshsize}}}} \!\!\!
    \fDGFFtheta(\zdbis)
    \, - \, 4\,\fDGFFtheta(\zd)\,.
\end{align*}
Then, for any element $F\in\Grass_{\ddomain{\meshsize}}[\fDGFFxi,\fDGFFtheta]$ we define its \term{fDGFF correlation function} as
\begin{align*}
    \bigfDGFFCorrFun{\ddomain{\meshsize}}{F\;\!}
    \;\coloneqq\;
    \frac{1}{\mathcal{Z}_{\ddomain{\meshsize}}}
    \int\!\!\!\int
    F\;
    e^{-\ActionfDGFF_{\ddomain{\meshsize}}[\fDGFFxi,\fDGFFtheta]}
    \cd\fDGFFxi\cd\fDGFFtheta\,,
\end{align*}
where $\mathcal{Z}_{\ddomain{\meshsize}} \coloneqq \iint e^{-\ActionfDGFF_{\ddomain{\meshsize}}[\fDGFFxi,\fDGFFtheta]} \cd\fDGFFxi\cd\fDGFFtheta$.

\noindent
For proofs of the statements in the following remark, we refer the reader to \cite{abdesselam}.

\begin{rmk}
\label{rmk: 2n pt function fDGFF}
The correlation function of the monomial~$\fDGFFxi(\zd_1)\cdots\fDGFFxi(\zd_n)\fDGFFtheta(\zdbis_1)\cdots\fDGFFtheta(\zdbis_m)$ vani\-shes unless $n=m$.
For the two-point function, we have
\begin{align*}
    \BigfDGFFCorrFun{\ddomain{\meshsize}}{\fDGFFxi(\zd)\fDGFFtheta(\zdbis)}
    \;=\;
    4\pi\,
    \Green_{\ddomain{\meshsize}}(\zd,\zdbis)\,,
\end{align*}
where $\Green_{\ddomain{\meshsize}}$ is ---up to a negative sign--- the inverse of the Dirichlet Laplacian, that is, it satisfies
\begin{align*}
    \big(\Delta_{\ddomain{\meshsize}}^{\texttt{D}}\big)_{(1)}
    \Green_{\ddomain{\meshsize}} (\zd,\zdbis)
    \,=\,
    -\delta_{\zd,\zdbis}
\end{align*}
for $\zd,\zdbis\in\interior{\ddomain{\meshsize}}$, where the subindex $(\,\cdot\,)_{(1)}$ indicates the Laplacian operator is taken with respect to the first variable of $\Green_{\ddomain{\meshsize}}$.
Finally, the $2n$-point functions satisfy the \term{fermionic Wick's formula}, i.e.~we have
\begin{align*}
    \BigfDGFFCorrFun{\ddomain{\meshsize}}{\fDGFFxi(\zd_1)\fDGFFtheta(\zdbis_1)\cdots\fDGFFxi(\zd_n)\fDGFFtheta(\zdbis_n)}
    \;=\;
    \sum_{\sigma\in\mathfrak{S}_n}
    (-1)^\sigma
    \prod_{i=1}^n
    \BigfDGFFCorrFun{\ddomain{\meshsize}}{\fDGFFxi(\zd_i)\fDGFFtheta(\zdbis_{\sigma(i)})}
\end{align*}
where $\mathfrak{S}_n$ denotes the symmetric group of order~$n$ and $(-1)^\sigma$ is the signature of the permutation~$\sigma$.
\hfill$\diamond$
\end{rmk}

We extend the correlation functions of the fDGFF to the boundary $\bdry\ddomain{\meshsize}$ by declaring $\bigfDGFFCorrFun{\ddomain{\meshsize}}{\fDGFFxi(\zd_1)\fDGFFtheta(\zdbis_1)\cdots\fDGFFxi(\zd_n)\fDGFFtheta(\zdbis_n)} = 0$ whenever at least one of the insertion points $\zd_1,\ldots,\zdbis_n$ sits on $\bdry\ddomain{\meshsize}$.
Note this does not spoil the discrete harmonicity of the correlation functions near the boundary since $\Green_{\ddomain{\meshsize}}$ is the Dirichlet Green's function.

\noindent
The following remark will be useful when we discuss scaling limits in Section~\ref{subsec: scaling limit}

\begin{rmk}
\label{rmk: discrete Wick with log corrections}
In a similar fashion as in Remark~\ref{rmk: general wick formula}, for three non-intersecting collections
$\underline{\zd} \coloneqq\{\zd_1 ,\ldots,\zd_n \}$,
$\underline{\zdbis} \coloneqq\{\zdbis_1 ,\ldots,\zdbis_n \}$,
and $\underline{\xd} \coloneqq\{\xd_1 ,\ldots,\xd_k \}$
of points in the discrete domain~$\ddomain{\meshsize}$,
the set of bijections from the set~${\underline{\zd} \cup\underline{\xd} }$ onto the set~${\underline{\zdbis} \cup\underline{\xd} }$ is denoted by ${\Bij(\underline{\zd} ,\underline{\zdbis} ;\underline{\xd} )}$.
For any fixed constant $c\in\C$, a straightforward manipulation of the fermionic Wick formula ---Remark~\ref{rmk: 2n pt function fDGFF}--- leads to
\begin{align*}
    &
    \BigfDGFFCorrFun{\ddomain{\meshsize}}
    {\fDGFFxi(\zd_1 ) \, \fDGFFtheta(\zdbis_1 )
    \cdots
    \fDGFFxi(\zd_n ) \, \fDGFFtheta(\zdbis_n ) \,
    \big( \, \fDGFFxi(\xd_1 ) \, \fDGFFtheta(\xd_1 ) + c \, \big)
    \cdots
    \big( \, \fDGFFxi(\xd_k ) \, \fDGFFtheta(\xd_k ) + c \, \big)}
    \phantom{\Bigg\vert}
    \\
    & 
    \mspace{40mu}
    \phantom{\Bigg\vert}
    =
    (-1)^k \!\!
    \sum_{b\in\Bij(  \underline{\zd} ,  \underline{\zdbis} ;  \underline{\xd} )} \!\!\!
    (-1)^b(-1)^{F_b}
    \prod_{\underset{b(y)\neq y}{y\in  \underline{\zd} \cup  \underline{\xd} }}
    \BigfDGFFCorrFun{\ddomain{\meshsize}}{
    \fDGFFxi(\xdbis ) \, \fDGFFtheta\big( b(\xdbis ) \big)}
    \prod_{\underset{b(\xd )= \xd }{\xd \in\underline{\xd} }}
    \bigg(\BigfDGFFCorrFun{\ddomain{\meshsize}}{
    \fDGFFxi(\xd ) \, \fDGFFtheta(\xd )}
    +\, c \,
    \bigg),
\end{align*}
where $(-1)^b$ is the signature of the bijection $b$, and $F_b$ is the number of fixed points of $b$.
\hfill$\diamond$
\end{rmk}

\vspace{-10pt}
\subsection{Local fields of the fDGFF}
\label{subsec: local fields}
We now define the local observables in the discrete setting that we aim to analyse and put in one-to-one correspondence with the local fields of a CFT.

\vspace{-10pt}
\subsubsection{Field polynomials}
The space of \term{field polynomials}\footnote{Note that, actually, they are \textit{not} polynomials since they are built with anticommuting variables.} of the fDGFF is
\begin{align*}
    \FieldPoly
    \;\coloneqq\;
    \Grass_{\Z^2}[\Fieldxi,\Fieldtheta]\,,
\end{align*}
that is, a field polynomial is an element of the Grassmann algebra over the set of generators $\{\Fieldxi(\zu),\Fieldtheta(\zu)\,\colon \zu\in\Z^2\}$.
When necessary, we will stress the product in $\FieldPoly$ with a dot ($\,\cdot\,$).
The \term{support} of a field polynomial $P\in\FieldPoly$, denoted by $\supp P$ is the smallest subset $S\subset \Z^2$ such that we have
\begin{align*}
    P
    \,\in\,
    \Grass_S[\Fieldxi,\Fieldtheta]
    \,\subset\,
    \Grass_{\Z^2}[\Fieldxi,\Fieldtheta]
    \,=\,
    \FieldPoly\,.
\end{align*}

The way to obtain an fDGFF observable is via the \term{evaluation map}:
given a discrete domain~$\ddomain{\meshsize}$ and point~$\zd\in\ddomain{\meshsize}$ therein, we define the map
\begin{align*}
    \ev^{\ddomain{\meshsize}}_{\zd}
    \;:\;
    \FieldPoly
    \;\longrightarrow\;
    \Grass_{\ddomain{\meshsize}}[\fDGFFxi,\fDGFFtheta]\,,
\end{align*}
as the unique associative algebra homomorphism that on the generators acts as
\begin{align*}
    \text{ev}^{\ddomain{\meshsize}}_{\bf{z}}\big(
    \Fieldxi(\zu)\big)
    \,=\,
    \fDGFFxi(\zd + \meshsize\zu)
    \mspace{40mu}
    \text{ and }
    \mspace{40mu}
    \text{ev}^{\ddomain{\meshsize}}_{\bf{z}}\big(
    \Fieldtheta(\zu)\big)
    \,=\,
    \fDGFFtheta(\zd + \meshsize\zu)
\end{align*}
where $\fDGFFxi(\zd + \meshsize\zu)$ and $\fDGFFtheta(\zd + \meshsize\zu)$ are interpreted as $0$ if $\zd + \meshsize\zu$ sits outside of ${\ddomain{\meshsize}}$.

\begin{ex}
Given a medial point $\zu_\medial\in\Z^2_\medial$, the field polynomials
\begin{align*}
    \holcurrxi(\zu_\medial)
    \,\coloneqq\,
    \pddee \Fieldxi(\zu_\medial)\,,
    &
    \mspace{40mu}
    \phantom{\Big\vert}
    \holcurrtheta(\zu_\medial)
    \,\coloneqq\,
    \pddee \Fieldtheta(\zu_\medial)\,,
    \mspace{50mu}
    \text{ and }
    \\
    \antiholcurrxi(\zu_\medial)
    \,\coloneqq\,
    \pddeebar \Fieldxi(\zu_\medial)\,,
    &
    \mspace{40mu}
    \phantom{\Big\vert}
    \antiholcurrtheta(\zu_\medial)
    \,\coloneqq\,
    \pddeebar \Fieldtheta(\zu_\medial)
\end{align*}
are called, respectively, the \term{holomorphic} and \term{antiholomorphic currrents}.
\hfill$\diamond$
\end{ex}

\vspace{-10pt}
\subsubsection{Null field polynomials}
A field polynomial~$P\in\FieldPoly$ is said to be \term{null}
if, for any discrete domain~$\ddomain{\meshsize}$ and any insertion point~$\zd\in\ddomain{\meshsize}$, we have
\begin{align*}
    \BigfDGFFCorrFun{\ddomain{\meshsize}}{\,\ev^{\ddomain{\meshsize}}_{\zd}(P)\,
    \fDGFFxi(\zd_1)\cdots\fDGFFxi(\zd_n)\,
    \fDGFFtheta(\zdbis_1)\cdots\fDGFFtheta(\zdbis_m)
    }
    \,=\,0
\end{align*}
whenever the distances~$\vert\,\zd-\zd_i\vert$, $\vert\,\zd-\zdbis_i\vert$, and $\cd(\zd,\bdry \ddomain{\meshsize})$ are large enough.
We let $\NullFieldPoly$ denote the subspace of null field polynomials.
Note that, while $\NullFieldPoly$ is a vector subspace, it is \emph{not} an ideal; furthermore, note the product of two null field polynomials is not necessarily null.

In the following, we provide some examples of field polynomials which are null and not null.
Let us start with a field polynomial that is \textit{not} null.

\begin{ex}\label{ex: 1 is not 0}
The field polynomial $1\in\FieldPoly$ is not null since we have, for example,
\begin{align*}
    \BigfDGFFCorrFun{\ddomain{\meshsize}}{\,\ev^{\ddomain{\meshsize}}_{\zd}(1)\,
    \fDGFFxi(\zd_1)\fDGFFtheta(\zdbis_1)}
    \,=\,
    \BigfDGFFCorrFun{\ddomain{\meshsize}}{\,\fDGFFxi(\zd_1)\fDGFFtheta(\zdbis_1)}
    \,=\,
    4\pi\,\Green_{\ddomain{\meshsize}}(\zd_1,\zdbis_1)
    \,\neq\,0
\end{align*}
for any discrete domain $\ddomain{\meshsize}$ and any points $\zd\in\ddomain{\meshsize}$ and $\zd_1,\zdbis_1\in\interior{\ddomain{\meshsize}}$.
\hfill$\diamond$
\end{ex}

Let us now see examples of field polynomials that \emph{are} null.

\begin{ex}\label{ex: Laplacian is null}
For any $\zu\in\Z^2$, the field polynomials
\begin{align*}
    \phantom{\bigg\vert}
    \dlaplacian\Fieldxi(\zu)
    \,\coloneqq\,
    \sum_{\underset{\vert\zubis-\zu\vert=1}{\zubis\in\Z^2}}
    \!\!\Fieldxi(\zubis)
    -4\Fieldxi(\zu)
    \mspace{50mu}
    \textnormal{ and }
    \mspace{50mu}
    \dlaplacian\Fieldtheta(\zu)
    \,\coloneqq\,
    \sum_{\underset{\vert\zubis-\zu\vert=1}{\zubis\in\Z^2}}
    \!\!\Fieldtheta(\zubis)
    -4\Fieldtheta(\zu)
    \phantom{\bigg\vert}
\end{align*}
are null.
This can be straightforwardly verified from Remark~\ref{rmk: 2n pt function fDGFF}.
Moreover, using Wick's formula, one can see that, if $P\in\FieldPoly$ is a field polynomial satisfying $\zu \notin \supp P$, the field polynomials $\dlaplacian\Fieldxi(\zu)\cdot P$ and $\dlaplacian\Fieldtheta(\zu)\cdot P$ are null, too.
\hfill$\diamond$
\end{ex}

\begin{ex}
By virtue of the factorisation of the Laplacian ---{Remark~\ref{rmk: factorisation laplacian}}--- and Example~\ref{ex: Laplacian is null}, we have that the field polynomials
\begin{align*}
	\ddeebar\holcurrxi(\zu)\,,
	\mspace{50mu}
	\ddeebar\holcurrtheta(\zu)\,,
	\mspace{50mu}
	\ddee\antiholcurrxi(\zu)\,,
	\mspace{40mu}
	\textnormal{and}
	\mspace{40mu}
	\ddee\antiholcurrtheta(\zu)
\end{align*}
are null for any $\zu\in\Z^2$.
\hfill$\diamond$
\end{ex}

\begin{ex}\label{ex: cubic null field}
The field polynomial~$\dlaplacian\Fieldxi(\zu)\Fieldtheta(0)\Fieldxi(0) + 4\pi \delta_{0,\zu}\Fieldxi(0)$ is null.
To see this, we evaluate it at some $\zd$ in some discrete domain~$\ddomain{\meshsize}$ and test it against distant insertions.
Using Wick's formula and the multilinearity of correlation functions, we have
\begin{align*}
    \BigfDGFFCorrFun{\ddomain{\meshsize}}
    {
    \Big( \dlaplacian &\,\fDGFFxi(\zd + \meshsize\zu) \fDGFFtheta(\zd) + 4\pi \delta_{0,\zu} \Big) \, \fDGFFxi(\zd)
    \cdot
    \fDGFFxi(\zd_1)\cdots\fDGFFxi(\zd_n)\,
    \fDGFFtheta(\zdbis_1)\cdots\fDGFFtheta(\zd_m)
    }
    \phantom{\bigg\vert}
    \\
    &
    \,=
    \bigg( \big(\dlaplacian\big)_{\zu}
    \BigfDGFFCorrFun{\ddomain{\meshsize}}
    {
    \fDGFFxi(\zd+\meshsize\zu) \fDGFFtheta(\zd) 
    }
    +
    4\pi \delta_{0,\zu} \,\bigg)\,
    \BigfDGFFCorrFun{\ddomain{\meshsize}}
    {
    \fDGFFxi(\zd)
    \,
    \fDGFFxi(\zd_1)\cdots\fDGFFxi(\zd_n)\,
    \fDGFFtheta(\zdbis_1)\cdots\fDGFFtheta(\zdbis_m)
    }
    \\
    &
    \mspace{20mu}
    -\,
    \big(\dlaplacian\big)_{\zu}
    \BigfDGFFCorrFun{\ddomain{\meshsize}}
    {
    \fDGFFxi(\zd+\meshsize\zu) \fDGFFxi(\zd) 
    }
    \,
    \BigfDGFFCorrFun{\ddomain{\meshsize}}
    {
    \fDGFFtheta(\zd)
    \,
    \fDGFFxi(\zd_1)\cdots\fDGFFxi(\zd_n)\,
    \fDGFFtheta(\zdbis_1)\cdots\fDGFFtheta(\zdbis_m)
    }
    \phantom{\Bigg\vert}
    \\
    &
    \mspace{20mu}
    -\,
    \sum_i (-1)^{i}
    \big(\dlaplacian\big)_{\zu}
    \BigfDGFFCorrFun{\ddomain{\meshsize}}
    {
    \fDGFFxi(\zd+\meshsize\zu) \fDGFFxi(\zd_i) 
    }
    \,
    \BigfDGFFCorrFun{\ddomain{\meshsize}}
    {
    \fDGFFxi(\zd) \, \fDGFFtheta(\zd)
    \cdots \widehat{\fDGFFxi(\zdbis_i)} \cdots
    }
    \phantom{\Bigg\vert}
    \\
    &
    \mspace{20mu}
    -\,
    \sum_j (-1)^{n+j}
    \big(\dlaplacian\big)_{\zu}
    \BigfDGFFCorrFun{\ddomain{\meshsize}}
    {
    \fDGFFxi(\zd+\meshsize\zu) \fDGFFtheta(\zdbis_j) 
    }
    \,
    \BigfDGFFCorrFun{\ddomain{\meshsize}}
    {
    \fDGFFxi(\zd) \, \fDGFFtheta(\zd)
    \cdots \widehat{\fDGFFtheta(\zdbis_j)} \cdots
    }
    \phantom{\Bigg\vert}\,,
\end{align*}
where the hat ($\,\widehat{\phantom{m}}\;\!$) indicates that a term does not appear in the product.
Note that in the right-hand side of the above equation, the first factor in each term vanishes, proving the nullity of the given field polynomial.
\hfill$\diamond$
\end{ex}

\vspace{-10pt}
\subsubsection{Local fields}
The space of local observables that is actually meaningful is the quotient
\begin{align*}
    \Fields
    \,\coloneqq\,
    \FieldPoly/\NullFieldPoly\,.
\end{align*}
The elements of $\Fields$ we refer to as \term{local fields} of the fDGFF.
We use the notation $P+\Null$ for the equivalence class of $P\in\FieldPoly$ in the quotient $\Fields$.

\begin{ex}
\label{ex: rep of ground states}
In Theorem~\ref{thm: isomorphism} below, we prove that the space of local fields~$\Fields$ is isomorphic to the logarithmic Fock space of the symplectic fermions.
Through that isomorphism the local fields
\begin{align*}
    \mspace{76mu}
    \fDGFFGround \,\coloneqq\, 1 + \Null
    \mspace{50mu}
    \textnormal{ and }
    \mspace{50mu}
    \fDGFFGroundPartner \, \coloneqq - \Fieldxi(0)\Fieldtheta(0) + \Null
\end{align*}
correspond to the identity field and its logarithmic partner, respectively,
and the local fields
\begin{align*}
    \phantom{\Big\vert}
    \fDGFFGroundXi \,\coloneqq\, \Fieldxi(0) + \Null
    \mspace{50mu}
    \textnormal{ and }
    \mspace{50mu}
    \fDGFFGroundTheta \, \coloneqq  \Fieldtheta(0) + \Null
\end{align*}
correspond to the fermionic ground states.
\hfill$\diamond$
\end{ex}

\vspace{-20pt}
\subsection{Current modes}
\label{subsec: current modes}
\vspace{-5pt}
Our goal is to understand the structure of the space~$\Fields$ in relation to the scaling limit of (the evaluation of representatives of) its elements within correlation functions.
We take CFT as a source of inspiration for that.
In the symplectic fermions CFT, the currents ---that is the fields~$\HolCurrentChi$, $\HolCurrentEta$, $\AntiHolCurrentChi$, and $\AntiHolCurrentEta$--- are fields that, in a sense, carry the symmetry algebra as their Fourier modes.
More precisely, inside correlation functions and for $\vert z - w \vert $ small enough, we have the convergent series
\begin{align*}
    \HolCurrentChi(z)
    \field(w)
    \, = \,
    \sum_{k\in\Z} \,
    \frac{ \big[ \HolCurrModexi{k} \field \big] (w) }
    {(z-w)^{k+1}}\,,
\end{align*}
and similar formulae for the other currents---see Property~\textbf{(CUR)} in Theorem~\ref{thm: characterization of CF}.
Series expansions of this type are called \term{operator pro\-duct expansions} (OPEs) and their existence is a fundamental pillar of two-dimensional conformal field theory.

In the case of the symplectic fermions, due to the holomorphicity of $\HolCurrentChi$ within correlation functions, that means that, in the continuum theory, we can obtain the correlation functions of the field~$[\HolCurrModexi{k} \field]$ as by contour integrating $\HolCurrentChi$ around $\field$ with the appropriate coefficient.
Although such converging expansion are not available in the discrete, we can define operators that mimic the extraction of Laurent coefficients with the tools of discrete complex analysis that we described in Section~\ref{subsec: tools discret C-anal}.

\begin{prop}
\label{prop: current modes}
Let $P+\Null\in\Fields$ be a local field.
The linear operators
\begin{align}
	\label{eq: chi mode}
	\HolCurrModexi{k}\big( P + \Null\big) &
	\,\coloneqq\,
	\dsqint_{\gamma}
	\frac{\dd\zu}{2\pi \ii}\,
	\zu^{[k]}_\diamond\,
	\holcurrxi(\zu_\medial) \cdot P 
	+ \Null\,,
	\\
	\label{eq: eta mode}
	\HolCurrModetheta{k}\big( P + \Null\big) &
	\,\coloneqq\,
	\dsqint_{\gamma}
	\frac{\dd\zu}{2\pi \ii}\,
	\zu^{[k]}_\diamond\,
	\holcurrtheta(\zu_\medial) \cdot P 
	+ \Null\,,
	\\
	\label{eq: chi bar mode}
	\AntiHolCurrModexi{k}\big( P + \Null\big) &
	\,\coloneqq\,
	-
	\dsqint_{\gamma}
	\frac{\dd\overline{\zu}}{2\pi \ii}\,
	\overline{\zu}^{[k]}_\diamond\,
	\antiholcurrxi(\zu_\medial) \cdot P 
	+ \Null\,,
	\mspace{20mu} \text{ and }
	\\
	\label{eq: eta bar mode}
	\AntiHolCurrModetheta{k}\big( P + \Null\big) &
	\,\coloneqq\,
	-
	\dsqint_{\gamma}
	\frac{\dd\overline{\zu}}{2\pi \ii}\,
	\overline{\zu}^{[k]}_\diamond\,
	\antiholcurrtheta(\zu_\medial) \cdot P 
	+ \Null\, ,
\end{align}
where $\gamma$ is a large enough discrete corner contour around the origin, are well-defined.
\end{prop}
\pagebreak
\begin{proof}
The operators are well-defined if they do not depend on the choice of corner contour~$\gamma$ and the choice of field-polynomial representative~$P$.
A straightforward application of the discrete Stokes' formula ---see Section~\ref{subsubsec: discrete integration}--- combined with the discrete holomorphicity of the discrete monomials ---Proposition~\ref{prop: monomials}--- and the null fields in Example~\ref{ex: Laplacian is null} is enough to see that the operators are well-defined.
To see that the operator does not depend on the representative $P$, one just needs to remember that the product of two field polynomials with supports that are at a large enough distance is null if one of the two field polynomials is null.
For a more detailed proof see Lemmata~4.2 and 4.3 in \cite{AdameCarrillo-discrete_symplectic_fermions}.
\end{proof}

We refer to the operators~\eqref{eq: chi mode}-\eqref{eq: eta mode} and \eqref{eq: chi bar mode}-\eqref{eq: eta bar mode} as the \term{holomorphic and antiholomorphic} \term{current modes}, respectively.
They constitute a representation of two anticommuting copies of the symplectic fermions algebra, as we state in the following proposition.

Let $\{\cdot,\cdot\}$ denote the usual anticommutator of operators, that is $\{A,B\}\coloneqq AB+BA$.

\begin{prop}\label{prop: anticommutation relations}
For all $k,\ell \in \Z$ the current modes satisfy
\begin{align*}
    \{\HolCurrModetheta{k}, \HolCurrModexi{\ell}\}
    \,=\,
    k \, \Kronecker_{k+\ell} \, \id_{\Fields}
    \,=\,
    \{\AntiHolCurrModetheta{k}, \AntiHolCurrModexi{\ell}\}\,,
\end{align*}
\begin{align*}
    \{\HolCurrModetheta{k}, \HolCurrModetheta{\ell}\}
    \,=\,
    \{\HolCurrModexi{k}, \HolCurrModexi{\ell}\}
    \,=\,
    \{\AntiHolCurrModetheta{k}, \AntiHolCurrModetheta{\ell} \}
    \,=\,
    \{\AntiHolCurrModexi{k}, \AntiHolCurrModexi{\ell}\}
    \,=\, 0\,,
\end{align*}
and
\begin{align*}
    \{\HolCurrModetheta{k}, \AntiHolCurrModetheta{\ell}\}
    \,=\,
    \{\HolCurrModexi{k}, \AntiHolCurrModexi{\ell}\}
    \,=\,
    \{\HolCurrModetheta{k}, \AntiHolCurrModexi{\ell}\}
    \,=\,
    \{\AntiHolCurrModetheta{k}, \HolCurrModexi{\ell}\}
    \,=\, 0\,.
\end{align*}
\end{prop}
\begin{proof}
The proof that each sector ---holomorphic and antiholomorphic--- constitutes a representation of the symplectic fermions algebra follows the same lines as the proof of Proposition~4.4 in \cite{AdameCarrillo-discrete_symplectic_fermions}.
For the anticommutativity of the sectors with one another, we write:
\begin{align*}
\big[\{\HolCurrModexi{k}, \AntiHolCurrModetheta{\ell}\}\big]
\big( P + \Null\, \big)
=\; &\,
-\frac{1}{(2\pi \ii)^2}
\bigg[
\dsqint_{\gamma_+}\,
\dd\zu\,
\zu^{[k]}_\diamond\,
\holcurrxi(\zu_\medial)
\dsqint_{\gamma}\,
\dd\overline{\zubis}\,
\overline{\zubis}^{[\ell]}_\diamond\,
\antiholcurrtheta(\zubis_\medial)
\\
&
\mspace{110mu}
+\;
\dsqint_{\gamma}\,
\dd\overline{\zubis}\,
\overline{\zubis}^{[\ell]}_\diamond\,
\antiholcurrtheta(\zubis_\medial)
\dsqint_{\gamma_-}\,
\dd\zu\,
\zu^{[k]}_\diamond\,
\holcurrxi(\zu_\medial)
\bigg]
\cdot P \,+\, \Null
\\
= \;&\,
\frac{1}{(2\pi \ii)^2}
\dsqint_{\gamma}\,
\dd\overline{\zubis}\,
\overline{\zubis}^{[\ell]}_\diamond\,
\bigg[
\dsqint_{\gamma_+}\,
-
\dsqint_{\gamma_-}
\bigg]
\dd\zu\,
\zu^{[k]}_\diamond\,
\holcurrxi(\zu_\medial)
\antiholcurrtheta(\zubis_\medial)
\cdot P \,+\, \Null
\phantom{\Bigg\vert}
\\
= \;&\,
\frac{\ii}{(2\pi \ii)^2}
\dsqint_{\gamma}\,
\dd\overline{\zubis}\,
\overline{\zubis}^{[\ell]}_\diamond\,
\sum_{\underset{\zu_\diamond\notin\interiordiam\gamma_-}{\zu_\diamond\in\interiordiam\gamma_+}}
\zu^{[k]}_\diamond\,
\ddeebar\holcurrxi(\zu_\diamond)
\antiholcurrtheta(\zubis_\medial)
\cdot P \,+\, \Null
\phantom{\Bigg\vert}
\\
= \;&\,
\frac{\ii}{2(2\pi \ii)^2}
\dsqint_{\gamma}\,
\dd\overline{\zubis}\,
\overline{\zubis}^{[\ell]}_\diamond\,
\sum_{\underset{\zu_\primary\notin\interiorprim\gamma_-}{\zu_\primary\in\interiorprim\gamma_+}}
\zu^{[k]}_\primary\,
\dlaplacian\Fieldxi(\zu_\primary)
\pddeebar\Fieldtheta(\zubis_\medial)
\cdot P \,+\, \Null
\end{align*}
\begin{align*}
\mspace{100mu}
= \;&\,
-
\frac{4\pi\ii}{2(2\pi \ii)^2}
\dsqint_{\gamma}\,
\dd\overline{\zubis}\,
\overline{\zubis}^{[\ell]}_\diamond\,
\sum_{\underset{\zu_\primary\notin\interiorprim\gamma_-}{\zu_\primary\in\interiorprim\gamma_+}}
\zu^{[k]}_\primary\,
\big(\pddeebar\big)_{\zubis_\medial}\big(\delta_{\zu_\primary,\zubis_\medial}\big)
\cdot P \,+\, \Null
\\
= \;&\,
-
\frac{4\pi\ii}{2(2\pi \ii)^2}
\bigg[
\dsqint_{\gamma}\,
\dd\overline{\zubis}\,
\overline{\zubis}^{[\ell]}_\diamond\,
\pddeebar\zubis^{[k]}_\medial
\bigg]
\, P \,+\, \Null\,.
\end{align*}
The discrete integral in square brackets vanishes for all values of $k$ and $\ell$; see the proof of Proposition~4.2 in \cite{ABK-DGFF_local_fields}.
\end{proof}

In other words, the current modes render the space of local fields of the fDGFF a representation of the symmetry algebra of the symplectic fermions CFT.

\begin{coro}
The map $\FullSymFer\to\End(\Fields)$ given by (the linear extension of)
\begin{align*}
    &
    \phantom{\Big\vert}
    \AlgEta{k} \,\longmapsto\, \HolCurrModetheta{k}\,,
    \mspace{90mu}
    \AlgChi{k} \,\longmapsto\, \HolCurrModexi{k}\,,
    \mspace{90mu}
    \AlgK \,\longmapsto\, \id_\Fields\,,
    \\
    &
    \phantom{\Big\vert}
    \AlgEtaBar{k} \,\longmapsto\, \AntiHolCurrModetheta{k}\,,
    \mspace{90mu}
    \AlgChiBar{k} \,\longmapsto\, \AntiHolCurrModexi{k}\,,
    \mspace{20mu}
    \text{ and }
    \mspace{30mu}
    \AlgKBar \,\longmapsto\, \id_\Fields
\end{align*}
is a Lie superalgebra homomorphism.
\end{coro}

We now present some remarks and examples that will be useful to prove the isomorphism between $\Fields$ and $\FullFock$.

\begin{rmk}
The current modes of index $0$ satisfy
\begin{align*}
    \HolCurrModexi{0}
    \,=\,
    \AntiHolCurrModexi{0}
    \mspace{30mu}
    \text{ and }
    \mspace{30mu}
    \HolCurrModetheta{0}
    \,=\,
    \AntiHolCurrModetheta{0}
\end{align*}
as linear operators $\Fields\to\Fields$.
This can be checked by a straightforward application of the discrete Stokes' formula ---Section~\ref{subsubsec: discrete integration}---.
\hfill$\diamond$
\end{rmk}

\begin{rmk}
\label{rmk: action currents on ground states}
The ground states satisfy the relations
\begin{align*}
    \HolCurrModexi{0}\fDGFFGroundPartner
    \, = \,
    -\fDGFFGroundXi\,,
    \mspace{150mu} 
    \HolCurrModetheta{0}\fDGFFGroundPartner
    \, = \,&\,
    -\fDGFFGroundTheta\,,
    \phantom{\Big\vert}
    \\
    \HolCurrModexi{0}\fDGFFGroundTheta
    \, = \,
    \fDGFFGround\,,
    \mspace{65mu}
    \text{ and }
    \mspace{66mu}
    \HolCurrModetheta{0}\fDGFFGroundXi
    \, = \,&\,
    -\fDGFFGround\,.
    \phantom{\Big\vert}
\end{align*}
Let us explicitly write down the computation for the first of these relations.
We have, for a sufficiently large corner contour~$\gamma$ encircling the origin,
\begin{align*}
    \HolCurrModexi{0}\fDGFFGroundPartner
    \, = \, & \,
    \dsqint_{\gamma}
    \frac{\dd\zu}{2\pi \ii}\,
    \zu^{[0]}_\diamond\,
    \holcurrxi(\zu_\medial) \cdot 
    \Big(-\Fieldxi(0)\Fieldtheta(0)\Big)
    + \Null
    \\
    \, = \, & \,
    -\frac{1}{2\pi}
    \sum_{\zu_\diamond\in\interiordiam\gamma}
    \ddee\holcurrxi(\zu_\diamond)\Fieldxi(0)\Fieldtheta(0)
    + \Null
\end{align*}
\begin{align*}
    \, = \, & \,
    -\frac{1}{4\pi}
    \sum_{\zu_\primary\in\interiorprim\gamma}
    \dlaplacian\Fieldxi(\zu_\primary)\Fieldxi(0)\Fieldtheta(0)
    + \Null
    \\
    \, = \, & \,
    -\frac{1}{4\pi}
    \sum_{\zu_\primary\in\interiorprim\gamma}
    4\pi\delta_{0,\zu_\primary}\Fieldxi(0)
    + \Null
    \\
    \, = \, & \,
    -\Fieldxi(0)
    + \Null\phantom{\Big\vert}
    \\
    \, = \, & \,
    -\fDGFFGroundXi\,,\phantom{\Big\vert}
\end{align*}
where we have used the discrete Stokes' formula ---Section~\ref{subsubsec: discrete integration}---, the factorisation of the discrete Laplacian operator ---Remark~\ref{rmk: factorisation laplacian}---, and the null field from Example~\ref{ex: cubic null field}.
\hfill$\diamond$
\end{rmk}

\begin{ex}
With a simple application of the discrete Stokes' formula and the discrete monomials value $0^{[k]}=0$ for $k>0$ ---Section~\ref{subsubsec: discrete monomials}---, one can check that all four ground states ($\fDGFFGround$, $\fDGFFGroundXi$, $\fDGFFGroundTheta$, and $\fDGFFGroundPartner$) are annihilated by the current modes with positive index.
That is, we have
\begin{align*}
    \HolCurrModexi{k}\fDGFFGroundPartner
    \,=\,
    \AntiHolCurrModexi{k}\fDGFFGroundPartner
    \, = \,
    \HolCurrModetheta{k}\fDGFFGroundPartner
    \,=\,
    \AntiHolCurrModetheta{k}\fDGFFGroundPartner
    \, = \, 0\,,
\end{align*}
and similarly for $\fDGFFGround$, $\fDGFFGroundXi$, and $\fDGFFGroundTheta$\,.
\hfill $\diamond$
\end{ex}

\vspace{-10pt}
\subsubsection{The Sugawara construction}
\label{subsubsec: discrete Sugawara}
In Section~\ref{subsec: log fock space}, we discussed how the logarithmic Fock space of the sympletic fermions becomes a representation of the Virasoro algebra via the so-called Sugawara construction.
In the case of the space~$\Fields$, the following property enables this construction to work.

\begin{lemma}
For any local field~$P+\Null$, there exists $N\in\Zpos$ such that for all $n\geq N$ we have
\begin{align*}
    \HolCurrModexi{n} \big(\, P + \Null \,\big)
    \, = \,
    \HolCurrModetheta{n} \big(\, P + \Null \,\big)
    \, = \,
    \AntiHolCurrModexi{n} \big(\, P + \Null \,\big)
    \, = \,
    \AntiHolCurrModetheta{n} \big(\, P + \Null \,\big)
    \, = \,
    0\,.
\end{align*}
\end{lemma}
\begin{proof}
Fix a corner contour~$\gamma$ for the action ---\eqref{eq: chi mode}--\eqref{eq: eta bar mode}--- of any current mode on $P+\Null$.
Since the discrete neighbourhood of the origin where the discrete monomials $\zu\mapsto\zu^{[k]}$ vanishes grows with $k$, we just need to pick $N$ big enough such that $\zu^{[N]}$ vanishes (in the interior of and) around the contour~$\gamma$.
Then, by formulae~\eqref{eq: chi mode}--\eqref{eq: eta bar mode}, it is clear that the action of $\HolCurrModexi{n}$, $\HolCurrModetheta{n}$, $\AntiHolCurrModexi{n}$, and $\AntiHolCurrModetheta{n}$ on $P+\Null$ vanishes for $n\geq N$.
\end{proof}

Formally, the arguments in the proof of Lemma~\ref{lemma: sugawara} prove the following statement too.

\begin{lemma}
\label{lemma: discrete sugawara}
The linear operators~\mbox{$\Fields\rightarrow\Fields$} given by the formal sums
\begin{align*}
\dL{n}
\coloneqq
\sum_{k \,\geq\, n/2}
\HolCurrModexi{n-k}\HolCurrModetheta{k}
-
\sum_{k \,<\, n/2}
\HolCurrModetheta{k}\HolCurrModexi{n-k}
\mspace{30mu}
\text{ and }
\mspace{30mu}
\dLBar{n}
\coloneqq
\sum_{k \,\geq\, n/2}
\AntiHolCurrModexi{n-k}\AntiHolCurrModetheta{k}
-
\sum_{k \,<\, n/2}
\AntiHolCurrModetheta{k}\AntiHolCurrModexi{n-k}\,,
\end{align*}
for $n\in\Z$, are well-defined and constitute two commuting representations of the Virasoro algebra with central charge $c=-2$.
\end{lemma}

\newpage
\vspace{-5pt}
\section{Results}
\label{sec:results}
We now proceed to state and discuss the main results of this paper---their proofs are postponed for Section~\ref{sec: proofs}.

In Section~\ref{subsec: isomorphism}, we present our first main result ---Theorem~\ref{thm: isomorphism}---, which concerns the local fields of the fDGFF.
In particular, the space~$\Fields$ of fDGFF local fields is in one-to-one correspondence with the logarithmic Fock space of the symplectic fermions.
This correspondence becomes meaningful only thanks to our second main result ---Theorem~\ref{thm: scaling limit}--- regarding the scaling limit of correlation functions, which is presented in Section~\ref{subsec: scaling limit}.
Informally, the scaling limit of the ---appropriately renormalised---  correlation functions of fDGFF local fields are the correlation functions of the corresponding ---via the correspondence in Theorem~\ref{thm: isomorphism}--- symplectic fermions local fields.
We end the section by discussing some applications of these results to two models of statistical mechanics---the uniform spanning tree and the Abelian sandpile model.
Certain local observables in those models can be computed with the correlation functions of the fermionic DGFF---we discuss how to interpret the scaling limit of such observables as local fields in the symplectic fermions CFT in Section~\ref{subsec: apps to stat mech}.

\vspace{-10pt}
\subsection{Local fields of the fDGFF}
\label{subsec: isomorphism}
The one-to-one correspondence between the spaces of local observables in the discrete and the continuum is phrased algebraically---we prove those two spaces to be isomorphic as representations of the symplectic fermions symmetry algebra.
This approach is not only a tool but a feature that provides information about the discrete local fields at hand.
In particular, the appropriate way of renormalising the correlation functions in the discrete so as to recover non-trivial scaling limits is dictated by the algebraic structure of the space of fields.
This will \linebreak[1] become clear in Theorem~\ref{thm: scaling limit} below.

\begin{thm}\label{thm: isomorphism}
As representations of the symplectic fermions algebra, we have the isomorphism
\begin{align*}
    \Fields
    \,\cong\,
    \Fock\,.
\end{align*}
\end{thm}

\noindent
The proof is presented in Sections~\ref{subsec: lin field poly}, \ref{subsec: higher fields} and \ref{subsec: proof of isomorphism}.

One immediate consequence of this theorem is that,
in a similar fashion as the logarithmic Fock space ---Section~\ref{subsec: log fock space}---,
the space $\Fields$ admits the basis
\begin{align}\label{eq: basis}
\HolCurrModexi{-k_r}\cdots\HolCurrModexi{-k_1}
\HolCurrModetheta{-\ell_s}\cdots\HolCurrModetheta{-\ell_1}
\AntiHolCurrModexi{-\bar{k}_{\bar{r}}}\cdots\AntiHolCurrModexi{-\bar{k}_1}
\AntiHolCurrModetheta{-\bar{\ell}_{\bar{s}}}\cdots\AntiHolCurrModetheta{-\bar{\ell}_1}
\fDGFFGroundPartner
\end{align}
\vspace{-25pt}
\begin{align}
\nonumber
\text{ with } \qquad
& \phantom{\Big\vert}r,s,\bar{r},\bar{s}\in \Znneg \\
\nonumber
\text{ and the orderings } \qquad
& \phantom{\Big\vert}{0 \leq k_1 < k_2 < \cdots < k_r}\,,\\
\nonumber
& \phantom{\Big\vert}{0 \leq \ell_1 < \ell_2 < \cdots < \ell_s}\,,\\
\nonumber
& \phantom{\Big\vert}{0 < \bar{k}_1 < \bar{k}_2 < \cdots < \bar{k}_{\bar{r}}}
\,, \text{ and }\\
\nonumber
& \phantom{\Big\vert}{0 < \bar{\ell}_1 < \bar{\ell}_2 < \cdots < \bar{\ell}_{\bar{s}}}\,.
\end{align}

Before moving on, we comment on the Virasoro structure of $\Fields$, which plays a key role in the scaling limit of correlation functions.

The operators~$\dL{0}$ and $\dLBar{0}$ ---defined in Lemma~\ref{lemma: sugawara}--- provide a natural \mbox{$\Znneg$-bigrading} of the space~$\Fields$ via their generalised eigenspaces.
The generalised eigenvalues of $\dL{0}$ and $\dLBar{0}$ are called, respectively, \term{holomorphic} and \term{antiholomorphic generalised conformal dimensions}.
The bigrading can be expressed simply in terms of the basis~\eqref{eq: basis}:
a straightforward computation leads to
\begin{align*}
\Big[\,
\dL{0} -
\Delta_{\bk,\bl}\,
\Big]^2
\big(
\HolCurrModexi{-k_r}\cdots\HolCurrModexi{-k_1}
\HolCurrModetheta{-\ell_s}\cdots\HolCurrModetheta{-\ell_1}
\AntiHolCurrModexi{-\bar{k}_{\bar{r}}}\cdots\AntiHolCurrModexi{-\bar{k}_1}
\AntiHolCurrModetheta{-\bar{\ell}_{\bar{s}}}\cdots\AntiHolCurrModetheta{-\bar{\ell}_1}
\fDGFFGroundPartner
\big)
\,=\, 0\,,
\end{align*}
and
\begin{align*}
\Big[\,
\dLBar{0} -
\bar{\Delta}_{\bar{\bk},\bar{\bl}}\,
\Big]^2
\big(
\HolCurrModexi{-k_r}\cdots\HolCurrModexi{-k_1}
\HolCurrModetheta{-\ell_s}\cdots\HolCurrModetheta{-\ell_1}
\AntiHolCurrModexi{-\bar{k}_{\bar{r}}}\cdots\AntiHolCurrModexi{-\bar{k}_1}
\AntiHolCurrModetheta{-\bar{\ell}_{\bar{s}}}\cdots\AntiHolCurrModetheta{-\bar{\ell}_1}
\fDGFFGroundPartner
\big)
\,=\, 0\,,
\end{align*}
where $\Delta_{\bf{k},\bf{l}}\coloneqq\sum_{i=1}^rk_i+\sum_{j=1}^{s}\ell_j$
and
$\bar{\Delta}_{\bar{\bf{k}},\bar{\bf{l}}}
\coloneqq
\sum_{i=1}^{\bar{r}}\bar{k}_i+\sum_{j=1}^{\bar{s}}\bar{\ell}_j$ are the generalised conformal dimensions of the basis state.
In particular, this implies that the Virasoro modes~$\dL{0}$ and $\dLBar{0}$ have Jordan blocks of rank at most~$2$.
The generalised $(\dL{0}+\dLBar{0})$-eigenvalue of a field is referred to as its \term{generalised scaling dimension}.

\vspace{-10pt}
\subsection{Scaling limit of correlation functions}
\label{subsec: scaling limit}
For the rest of this section, we let $\domain\subset\C$ be a (open, simply-connected, proper, non-empty) domain of the complex plane,
and we take distinct points $z_1,\ldots,z_n,w_1,\ldots,w_m\in\domain$ thereon.
We also let $\ddomain{\meshsize}\subset\meshsize\Z^2$ be a family of discrete domains such that $\ddomain{\meshsize}\to\domain$ as $\meshsize\downarrow0$ in the Carath\'eodory sense\footnote{More precisely, what converges in the Carath\'eodory sense are the Jordan domains bounded by $\bdry\ddomain{\meshsize}$.}.
Finally, we let  $\zd_1^\meshsize,\ldots,\zd_n^\meshsize,\zdbis_1^\meshsize,\ldots,\zdbis_m^\meshsize\in\ddomain{\meshsize}$ be vertices in $\ddomain{\meshsize}$ that are closest to points $z_1,\ldots,z_n,w_1,\ldots,w_m\in\domain$ in the complex plane.

We also define $\mathsf{C} \coloneqq \gamma + \frac{3}{2} \log 2$, where $\gamma$ is the Euler--Mascheroni constant.

\pagebreak[4]

\begin{thm}\label{thm: scaling limit}
Fix a positive constant $\lambda\in\R_{>0}$.
For $i=1,\ldots,n$, let $\field_i\in\Fields$ be a local field with generalised scaling dimension $\Delta_i+\bar\Delta_i$.
Let $P_i^\meshsize\in\FieldPoly$ be a field-polynomial representative of the local field
\begin{align*}
    \field_i
    \,-\, \log \big( \lambda \, \meshsize \big)
    \Big[ (\dL{0}+\dLBar{0})-(\Delta_i+\bar\Delta_i)\Big]
    \field_i\,.
\end{align*}
We have
\begin{align*}
    \frac{\BigfDGFFCorrFun{\ddomain{\meshsize}}{P_1^\meshsize\big(\zd_1^\meshsize\big)
    \cdots P_n^\meshsize\big(\zd_n^\meshsize\big)}}
    {\mathlarger{\mathlarger{\meshsize}}^{\,\sum_i(\Delta_i+\bar\Delta_i)}}
    \mspace{20mu}
    \xrightarrow{\mspace{50mu}\meshsize\,\downarrow\,0\mspace{50mu}}
    \mspace{20mu}
    \BigSFCorrFun{\domain}{\thectt(\lambda)}{\field_1(z_1) \cdots \field_n(z_n)}
\end{align*}
uniformly for $(z_1,\ldots,z_n)$  on compacts of $\Conf{n}{\domain}$, where we used $P_i^\meshsize\big(\zd_i^\meshsize\big)\coloneqq\ev_{\zd_i^\meshsize}^{\ddomain{\meshsize}}\big(P_i^\meshsize\big)$ and $\thectt(\lambda)\coloneqq 2 \big( \log \lambda + \mathsf{C}\, \big)$.
\end{thm}

\begin{rmk}
\label{rmk: arbitrary scale}
In each discrete approximation $\ddomain{\meshsize}$ of the continuum domain~$\domain$, the linear scale of the system is proportional to $\meshsize$; the proportionality constant $\lambda$ is arbitrary.
There is no good reason to pick as a linear scale the edge length $\meshsize$ instead of the diagonal length of the square-grid plaquettes $\sqrt{2}\,\meshsize$ or half the length of an edge $\meshsize/2$.
\hfill$\diamond$
\end{rmk}

\begin{rmk}
This arbitrariness is \emph{not} relevant in non-logarithmic theories---a change in the linear scale can be absorbed as a multiplicative factor in the local observables.
This is also the case for local fields in our theory with well-defined scaling dimensions.
However, our theory contains more exotic fields that feature logarithmic divergencies in the scaling limit.
These are the generalised eigenvectors of $\dL{0}+\dLBar{0}$ that are not eigenvectors.
\hfill$\diamond$
\end{rmk}

\vspace{-10pt}
In Section~\ref{subsec: corr fun symplectic fermions}, we gave an algebraic argument ---the existence of x-trivial Virasoro self-isomorphisms of the space of the space of fields--- to explain why symplectic-fermion correlation functions are defined only up to a choice of an arbitrary constant.
In this picture of the correlation functions of a logarithmic CFT ---as the scaling limit of discrete correlation functions---, this fact becomes less abstract.
This ambiguity arises purely from the choice of an arbitrary linear scale; that is the choice of the arbitrary constant $\lambda$ discussed in Remark~\ref{rmk: arbitrary scale}.
Let us see this in a specific example.


\begin{ex}
The ground state~$\fDGFFGroundPartner$ does not have a well-defined scaling dimension ---recall that we have $(\dL{0}+\dLBar{0})\fDGFFGroundPartner = 2\cdot\fDGFFGround$--- which implies that it requires a logarithmic correction to get a non-trivial scaling limit.
More precisely, we have the convergence
\begin{align*}
    \BigfDGFFCorrFun{\ddomain{\meshsize}}{
    \ev_{\zd^\meshsize}^{\ddomain{\meshsize}}
    \Big(  \fDGFFGroundPartner  - 2 \log( \lambda \, \meshsize ) \fDGFFGround
    \Big) 
    }
    \mspace{20mu}
    \xrightarrow{
    \mspace{50mu}\meshsize\,\downarrow\,0\mspace{50mu}
    }
    \mspace{20mu}
    \BigSFCorrFun{\domain}{\thectt(\lambda)}{\GroundPartner(z)}
    = \, - \, 4\pi\,\HarmOfGreen_{\domain}(z,z) - \thectt(\lambda)
\end{align*}
uniformly for $z$ on compacts of $\domain$.
The arbitrary CFT constant $\alpha$ on the right-hand side arises from the arbitrary linear scale $\lambda$ chosen in the discrete.
\hfill$\diamond$ 
\end{ex}

We will see another example of this logarithmic behaviour in Section~\ref{subsubsec: ASM} below when we discuss the dissipation field of the Abelian sandpile model.
For general fields, we have the following formula that is reminiscent of the self-isomorphisms of the space $\FullFock$ discussed in \cite{AdaCar-symplectic_fermions}.

\begin{ex}
\label{ex: good basis scaling fields}
For a field
\begin{align*}
    \field \coloneqq
    \HolCurrModexi{-k_r}\cdots\HolCurrModexi{-k_1}
    \HolCurrModetheta{-\ell_s}\cdots\HolCurrModetheta{-\ell_1}
    \AntiHolCurrModexi{-\bar{k}_{\bar{r}}}\cdots\AntiHolCurrModexi{-\bar{k}_1}
    \AntiHolCurrModetheta{-\bar{\ell}_{\bar{s}}}\cdots\AntiHolCurrModetheta{-\bar{\ell}_1}
    \fDGFFGroundPartner
\end{align*}
in the basis~$\eqref{eq: basis}$, the corresponding field that behaves well under rescaling ---in the sense of Theorem \ref{thm: scaling limit}--- is
\begin{align*}
    \field
    \,-\, \log \big( \lambda \, \meshsize \big)
    \Big[ (\dL{0}
    &
    + \dLBar{0})-(\Delta+\bar\Delta)\Big]
    \field
    \,=
    \\
    &
    \HolCurrModexi{-k_r}\cdots\HolCurrModexi{-k_1}
    \HolCurrModetheta{-\ell_s}\cdots\HolCurrModetheta{-\ell_1}
    \AntiHolCurrModexi{-\bar{k}_{\bar{r}}}\cdots\AntiHolCurrModexi{-\bar{k}_1}
    \AntiHolCurrModetheta{-\bar{\ell}_{\bar{s}}}\cdots\AntiHolCurrModetheta{-\bar{\ell}_1}
    \Big(
    \fDGFFGroundPartner
    - 2\,\log\big(\lambda\,\meshsize\big) \fDGFFGround
    \Big)
\end{align*}
where $\Delta+\bar\Delta = \sum_{a=1}^r k_a + \sum_{b=1}^s \ell_b
+ \sum_{\bar a=1}^{\bar r} \bar k_{\bar a} + \sum_{\bar b=1}^{\bar s} \bar\ell_{\bar b}$.
\hfill$\diamond$
\end{ex}

\begin{figure}[b!]
\centering
\begin{overpic}[scale=0.467, tics=10]{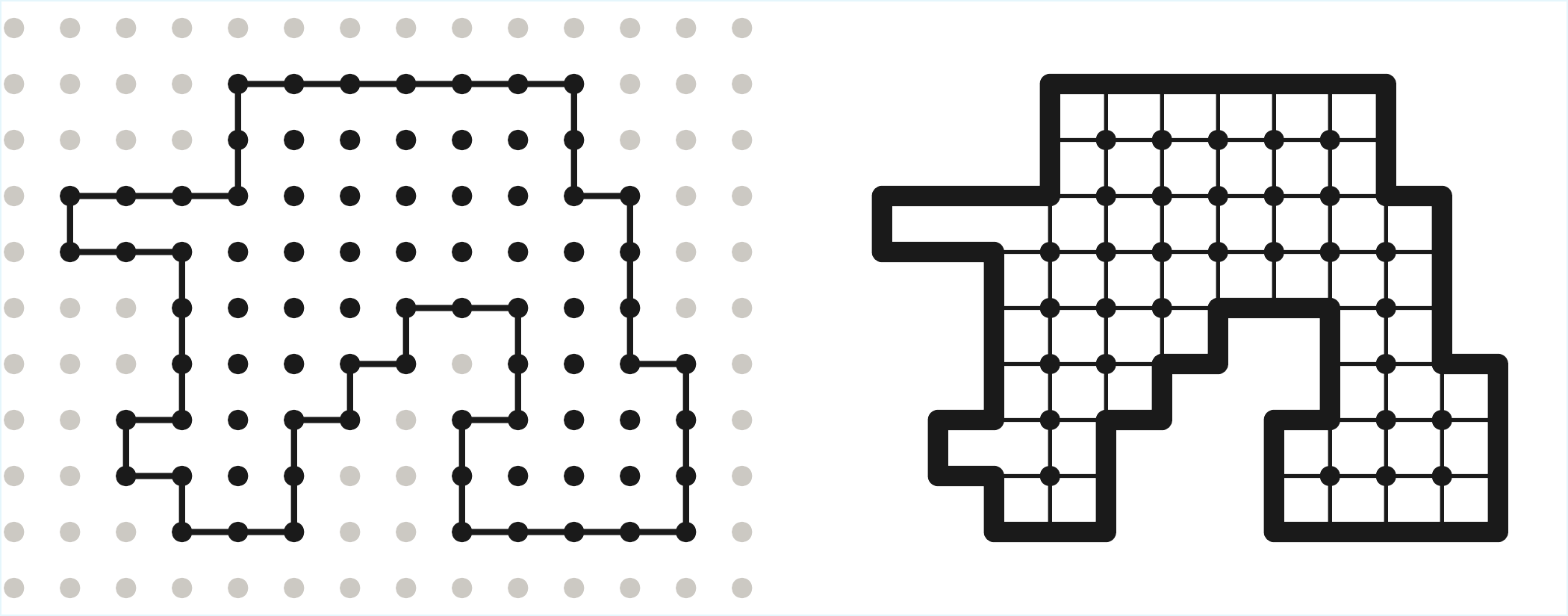}
\end{overpic}
\centering
\caption{On the left, a discrete (Jordan) domain $\ddomain{\meshsize}$ and its boundary~$\bdry\ddomain{\meshsize}$ (black) \\ in the infinite square grid~$\Z^2$ (gray).
On the right, its nearest-neighbour graph with \\ wired boundary $(V_{\ddomain{\meshsize}}, E_{\ddomain{\meshsize}})$;
the thick line along the boundary vertices \\ indicates they all correspond to the same vertex in $V_{\ddomain{\meshsize}}$.}
\label{fig: wired graph}
\end{figure}

\vspace{-10pt}
\subsection{Applications to probabilistic models}
\label{subsec: apps to stat mech}
The results discussed above ---Theorems~\ref{thm: isomorphism} and \ref{thm: scaling limit}--- concern purely the local fields of the fDGFF.
Yet, the fDGFF is known to be suitable to compute the correlation functions of certain observables of the uniform spanning tree and the Abelian sandpile model among other models.
We hence exploit this connection to understand (the scaling limit of) local observables in these models as local fields in the symplectic fermions CFT

\noindent
To define our probabilistic models below we will need the following definitions.

For a discrete domain~$\ddomain{\meshsize}$, consider its nearest-neighbour \term{graph with wired boun\-dary}.
That is, the (multi)graph that one obtains by connecting with an edge all nearest neighbours in the domain~$\ddomain{\meshsize}$ and wiring all vertices in $\bdry\ddomain{\meshsize}$ into one single vertex while keeping track of possible double and triple edges near the boundary---see Figure~{\ref{fig: wired graph}}.
We let $(V_{\ddomain{\meshsize}}, E_{\ddomain{\meshsize}})$ denote this graph.

Let $e\in E_{\ddomain{\meshsize}}$ be an edge.
If it is a horizontal edge, we let $\zd_{\mathsf E ; e}\in V_{\ddomain{\meshsize}}$ and $\zd_{\mathsf W ; e} \in V_{\ddomain{\meshsize}}$ denote the vertices on the right and left of $e$, respectively.
Similarly, if $e$ is a vertical edge, we let $\zd_{\mathsf N ; e}\in V_{\ddomain{\meshsize}}$ and $\zd_{\mathsf S ; e}\in V_{\ddomain{\meshsize}}$ denote the vertices at the top and at the bottom of $e$, respectively.
Then, the \term{discrete gradient} of a function $f$ on $\ddomain{\meshsize}$ at the edge $e$ is given by
\begin{align}
\label{eq: discrete gradient}
    \dgrad f (e)
    \,\coloneqq\,
    \begin{cases}
        \;
        f( \zd_{\mathsf E ; e} ) - f( \zd_{\mathsf W ; e} )
        \phantom{\big\vert}
        &
        \mspace{20mu}
        \text{ if $e$ horizontal}\,,
        \\
        \phantom{.}
        \\
        \;
        f( \zd_{\mathsf N ; e} ) - f( \zd_{\mathsf S ; e} )
        \phantom{\big\vert}
        &
        \mspace{20mu}
        \text{ if $e$ vertical}\,.
    \end{cases}
\end{align}
For functions of multiple variables, we use the notation $\big(\dgrad\big)_{(i)}$ to indicate the gradient is taken with respect to the $i$-th variable.
For functions on $\Z^2$, we identify edges of the infinite square-grid with their midpoint seen as embedded in $\C$, and we use the notation
$\dgrad f \big(k + \frac{1}{2} + \ii\ell \big) \coloneqq f(k+1+\ii\ell) - f(k+\ii\ell)$
and
$\dgrad f \big(k + \ii\ell +\frac{\ii}{2} \big) \coloneqq f(k+\ii\ell+\ii) - f(k+\ii\ell)$
for $k,\ell\in\Z$.

\vspace{-10pt}
\subsubsection{The uniform spanning tree}
\label{subsubsec: UST}
A spanning tree of a graph is a subset of edges that contains no loops and covers every vertex.
Given a discrete domain~$\ddomain{\meshsize}$, the \textbf{uniform spanning tree} (UST) on $\ddomain{\meshsize}$ with wired boundary conditions is a spanning tree $\boldsymbol{\tau}$ of $(V_{\ddomain{\meshsize}}, E_{\ddomain{\meshsize}})$ taken uniformly at random among all possible spanning trees.
We let $\P^\text{UST}_{\ddomain{\meshsize}}$ and $\E^\text{UST}_{\ddomain{\meshsize}}$ denote, respectively, the UST probability measure and its expected values.
Then, the classical Kirchhoff's matrix-tree theorem implies
\begin{align*}
    \P^\text{UST}_{\ddomain{\meshsize}}
    \big[\, \boldsymbol{\tau} = T \,\big]
    \, = \,
    \frac{1}{\det( -\Delta_{\ddomain{\meshsize}}^{\texttt{D}})}
\end{align*}
for any wired spanning tree $T$ of $\ddomain{\meshsize}$, where the matrix $\dlaplacian^{\texttt{D}}$ with rows and columns indexed by the (interior) vertices of $\ddomain{\meshsize}$ is the Dirichlet Laplacian matrix
\begin{align*}
    (\Delta_{\ddomain{\meshsize}}^{\texttt{D}})_{\zd,\zdbis}
    \,=\,
    \begin{cases}
        \mspace{15mu} -4 & \mspace{30mu} \textnormal{if } \zd = \zdbis\,,
        \\
        \mspace{15mu} \; \; 1 &  \mspace{30mu} \textnormal{if } \vert \, \zd - \zdbis \,\vert = \meshsize\,, \textnormal{and}
        \\
        \mspace{15mu} \; \; 0 & \mspace{30mu} \textnormal{otherwise}\,,
    \end{cases}
\end{align*}
as we defined in Section~\ref{subsec: fDGFF def}.
By the celebrated transfer-current theorem of Burton and Pemantle  ---\cite{burtonpem}--- we have that the edges of the spanning tree form a determinantal process.
In precise terms, for any distinct edges $e_1, \ldots , e_n$, we have
\begin{align*}
    \P^\text{UST}_{\ddomain{\meshsize}}
    \big[\, e_1, \ldots , e_n  \in \boldsymbol{\tau} \big]
    \, = \,
    \det \bigg(
    \big(\dgrad\big)_{(1)} \big(\dgrad\big)_{(2)}
    \Green_{\ddomain{\meshsize}}( e_i, e_j )
    \bigg)_{i,j \in \{1 , \ldots , n\}}\,,
\end{align*}
where $\Green_{\ddomain{\meshsize}}$ is the Dirichlet Green's function as defined in Remark~\ref{rmk: 2n pt function fDGFF}.
Then, the following connection between the fDGFF and the UST follows from the transfer-current theorem and Remark~\ref{rmk: 2n pt function fDGFF}.

\begin{figure}[t!]
\centering
\begin{overpic}[scale=0.467, tics=10]{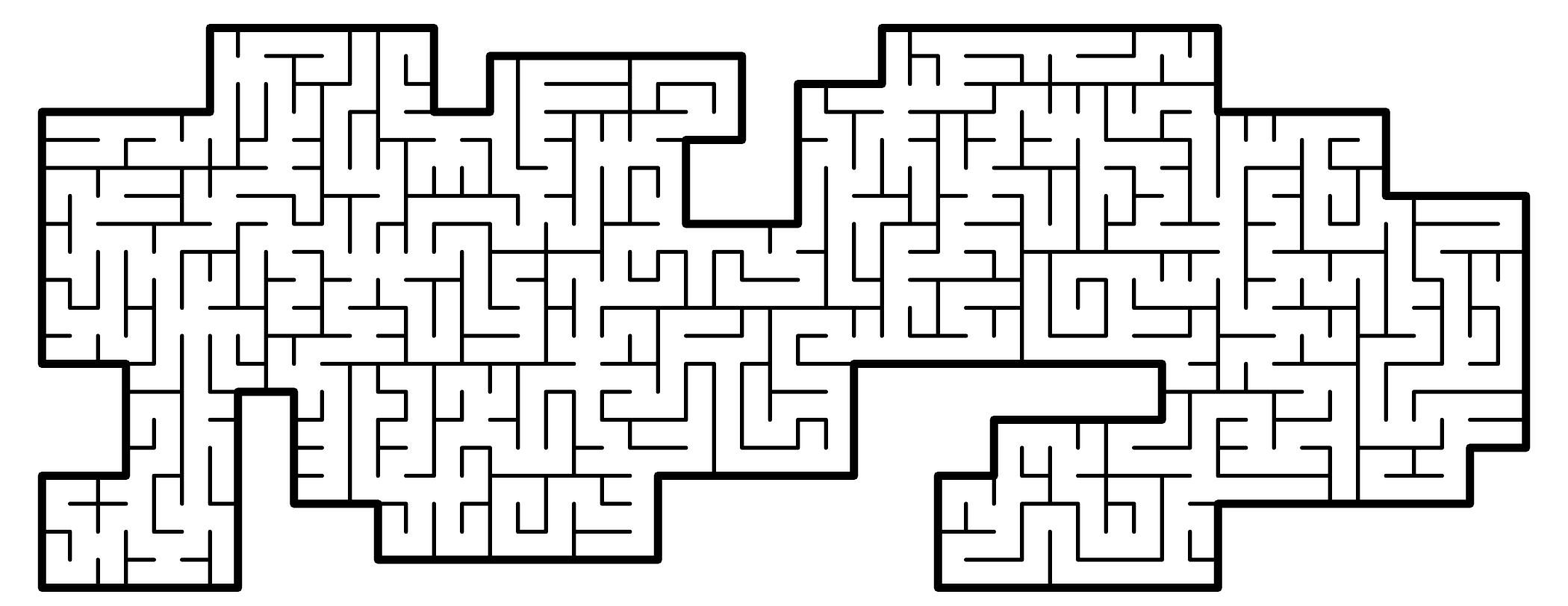}
\end{overpic}
\centering
\caption{ An example of a spanning tree with wired boundary conditions.}
\label{fig: ust}
\end{figure}

\begin{prop}
\label{prop: open-edge fDGFF}
Let $e_1,\ldots,e_n$ be distinct edges in $E_{\ddomain{\meshsize}}$.
We have
\begin{align*}
    \P^\textnormal{UST}_{\ddomain{\meshsize}}
    \big[\, e_1, \ldots , e_n  \in \boldsymbol{\tau} \big]
    \, = \,
    \frac{1}{(4\pi)^n}
    \BigfDGFFCorrFun{\ddomain{\meshsize}}
    {\dgrad\fDGFFxi(e_1)\dgrad\fDGFFtheta(e_1)
    \cdots
    \dgrad\fDGFFxi(e_n)\dgrad\fDGFFtheta(e_n)}
    \,.
\end{align*}
\end{prop}

This connection allows one to understand UST local observables as CFT fields in the scaling limit via our isomorphism ---Theorem~\ref{thm: isomorphism}---.
For concreteness we give two examples.

\vspace{-10pt}
\paragraph{The degree field.}
For a vertex~$\zd\in\ddomain{\meshsize}$, we define the degree of $\zd$ in the UST as the random variable~$\degreeField(\zd)$ whose value on each UST realisation is the number of open edges adjacent to the vertex~$\zd$---we refer to $\degreeField$ as the \term{degree field}---see also \cite{CCRR-fDGFF_sandpile_UST, AlanCor}.
From \cite[Corollary 3.4]{CCRR-fDGFF_sandpile_UST} and Proposition~\ref{prop: open-edge fDGFF}, it straightforwardly follows that we have
\begin{align*}
    \E^{\textnormal{UST}}_{\ddomain{\meshsize}}
    \big[\, \degreeField(\zd_1) \cdots \degreeField(\zd_n)
    \,\big]
    \;=\;
    \frac{1}{(4\pi)^n}\,
    \BigfDGFFCorrFun{\ddomain{\meshsize}}
    {\sum_{e_1\sim\,\zd_1} \!
    \dgrad\fDGFFxi(e_1)\dgrad\fDGFFtheta(e_1)
    \cdots
    \sum_{e_n\sim\,\zd_n} \!
    \dgrad\fDGFFxi(e_n)\dgrad\fDGFFtheta(e_n)}\,, 
\end{align*}
where the notation $e \sim \zd$ indicated that the edge $e$ contains the vertex $\zd$. 
However, these correlation functions have a trivial scaling limit; they converge to $2^n$, where $2$ is the thermodynamic limit of the average of $\degreeField$.
A more interesting scaling limit is obtained if one considers the correlation functions of the difference of $\degreeField$ from its average value.
Hence, we consider the fDGFF local field 
\begin{align*}
    \dDegree
    \;\coloneqq\;
    \frac{1}{4\pi}
    \mspace{-20mu}
    \sum_{\mspace{30mu} x \in \{\pm 1 , \pm \ii\}}
    \mspace{-20mu}
    \dgrad\Fieldxi \big( \tfrac{x}{2} \big)
    \dgrad\Fieldtheta \big( \tfrac{x}{2} \big)
    \,-\,
    2
    \,+\,
    \Null
\end{align*}
where we already subtracted the full-plane average of $\degreeField$---that is, $2$.

In order to take the scaling limit of correlation functions of the local field $\dDegree$, we need to know its expression in the basis~\eqref{eq: basis}---in Section~\ref{subsec: ex of computation}, we present an example of such a computation in detail.
For the degree field~$\dDegree$, we have
\begin{align*}
    \dDegree
    \;=\;
    \frac{1}{\pi}
    \big(\,  \AntiHolCurrModexi{-1}\HolCurrModetheta{-1}
    +
    \HolCurrModexi{-1}\AntiHolCurrModetheta{-1}
    \big) \fDGFFGround
    \mspace{15mu} 
    +
    \parbox{10em}{ \centering terms with higher scaling dimensions. }
\end{align*}

In the scaling limit, the only relevant terms are those with the lowest scaling dimensions---higher scaling-dimension terms do not contribute to the scaling limit and can depend on the particular discretisation of the field theory. 
These terms are mapped to the symplectic-fermions field~$\contDegree\coloneqq\frac{1}{\pi}( \AlgEta{-1}\AlgChiBar{-1} + \AlgChi{-1}\AlgEtaBar{-1}) \Ground$ via the isomorphism in Theorem~\ref{thm: isomorphism}.
Hence, by Theorem~\ref{thm: scaling limit}, we have
\begin{align*}
    \frac{\E^{\textnormal{UST}}_{\ddomain{\meshsize}}
    \big[
    \big(\, \degreeField(\zd^\meshsize_1) - 2 \,\big)
    \cdots
    \big(\, \degreeField(\zd^\meshsize_n) - 2 \,\big)
    \big]}
    {\meshsize^{2\,n}}
    \mspace{10mu}
    \xrightarrow{\mspace{20mu} \meshsize\downarrow0 \mspace{20mu}}
    \mspace{10mu}
    \BigSFCorrFun{\domain}{\thectt}{
    \contDegree(z_1) \cdots \contDegree(z_n)
    }
\end{align*}
uniformly for $(z_1,\ldots,z_n)$ on compacts of $\Conf{n}{\domain}$.

\begin{ex}
Let $\zd^\meshsize$ be the closest vertex to $z\in\domain$ in the discrete domain~$\ddomain{\meshsize}$.
We have
\begin{align*}
    \frac{1}{\meshsize^2}
    \Big(\,
    \E^{\textnormal{UST}}_{\ddomain{\meshsize}}
    \big[ \, \degreeField(\zd^\meshsize) \, \big]
    \,-\, 2
    \, \Big)
    \mspace{10mu}
    \xrightarrow{\mspace{20mu} \meshsize\downarrow0 \mspace{20mu}}
    \mspace{10mu}
    \BigSFCorrFun{\domain}{\thectt}{
    \contDegree(z)
    }
    \,=\,
    8 \, \partial_1\partialBar_2\HarmOfGreen_{\domain}(z,z)
\end{align*}
uniformly for $z$ on compacts of $\domain$.
Note that the limit is actually independent of $\thectt$.
\hfill$\diamond$
\end{ex}

\vspace{-10pt}
\paragraph{The (horizontal) open-edge field.}
For a vertex $\zd\in\ddomain{\meshsize}\setminus\bdry\ddomain{\meshsize}$,
let $\{ \,\,\OEright \text{ open at } \zd \,\}$
be the UST event that the edge adjacent to $\zd$ on its right is open and let $\indicator_{\,\OEright}(\zd)$ denote its indicator function.
Similarly, let $\{ \,\,\OEleft \text{ open at } \zd \,\}$
be the UST event that the edge adjacent to the vertex $\zd$ on its left is open and let $\indicator_{\,\OEleftsubs}(\zd)$ denote its indicator function.
The UST local observables $\zd\mapsto\indicator_{\,\OEright}(\zd)$
and $\zd\mapsto\indicator_{\,\OEleftsubs}(\zd)$
can be respectively associated with the fDGFF field polynomials
\begin{align*}
    P_{\OEright}
    \;\coloneqq\;
    \frac{1}{4\pi}\,
    \dgrad\Fieldxi \big( \tfrac{1}{2} \big)
    \dgrad\Fieldtheta \big( \tfrac{1}{2} \big)
    \mspace{30mu}
    \textnormal{ and }
    \mspace{30mu}
    P_{\OEleftsubs}
    \;\coloneqq\;
    \frac{1}{4\pi}\,
    \dgrad\Fieldxi \big( \tfrac{-1\;}{2} \big)
    \dgrad\Fieldtheta \big( \tfrac{-1\;}{2} \big)\,.
\end{align*}
in the sense that we have the equality ---Proposition~\ref{prop: open-edge fDGFF}---
\begin{align*}
    &
    \E^{\textnormal{UST}}_{\ddomain{\meshsize}}
    \big[\,
    \indicator_{\,\OEright}(\zd_1)
    \cdots
    \indicator_{\,\OEright}(\zd_n)
    \indicator_{\,\OEleftsubs}(\zdbis_1)
    \cdots
    \indicator_{\,\OEleftsubs}(\zdbis_m)
    \,\big]
    \phantom{\bigg\vert}
    \\
    & \mspace{200mu}
    = \;
    \BigfDGFFCorrFun{\ddomain{\meshsize}}{
    \ev^{\ddomain{\meshsize}}_{\zd_1}
    \big( P_{\OEright} \big)
    \cdots
    \ev^{\ddomain{\meshsize}}_{\zd_n}
    \big( P_{\OEright} \big)
    \ev^{\ddomain{\meshsize}}_{\zdbis_1}
    \big( P_{\OEleftsubs} \big)
    \cdots
    \ev^{\ddomain{\meshsize}}_{\zdbis_m}
    \big( P_{\OEleftsubs} \big)
    }\,.
\end{align*}
As we argued for the degree field, the more interesting local observables with non-constant correlation functions ---with respect to the insertion point--- in the scaling limit are
\begin{align*}
    \zd
    \,\longmapsto\,
    \indicator_{\,\OEright}(\zd)
    -
    \frac{1}{2}
    \mspace{40mu}
    \textnormal{ and }
    \mspace{40mu}
    \zd
    \,\longmapsto\,
    \indicator_{\,\OEleftsubs}(\zd)
    -
    \frac{1}{2}
\end{align*}
where we have subtracted the thermodynamic limit of the probability of one given edge being open in the UST---that is, $\frac{1}{2}$.
As expected, their associated fields 
\begin{align*}
    \mathsf{P}_{\OEright}
    \,\coloneqq\,
    P_{\OEright} - \frac{1}{2}
    + \Null
    \mspace{40mu}
    \textnormal{ and }
    \mspace{40mu}
    \mathsf{P}_{\OEleftsubs}
    \,\coloneqq\,
    P_{\OEleftsubs} - \frac{1}{2}
    + \Null
\end{align*}
have the same leading-order terms in the basis~\eqref{eq: basis}.
In particular, we have that both $\mathsf{P}_{\OEright}$ and $\mathsf{P}_{\OEleftsubs}$ are equal to ---see Section~\ref{subsec: ex of computation}---
\begin{align*}
    \frac{1}{4\pi}\,
    \big( \HolCurrModexi{-1} + \AntiHolCurrModexi{-1} \big)
    \big( \HolCurrModetheta{-1} + \AntiHolCurrModetheta{-1} \big) \fDGFFGround
    \mspace{15mu} 
    +
    \parbox{10em}{ \centering terms with higher scaling dimensions. }
\end{align*}
We define then the fDGFF field
$\horEdgeField \coloneqq
\frac{1}{4\pi}\,
\big( \HolCurrModexi{-1} + \AntiHolCurrModexi{-1} \big)
\big( \HolCurrModetheta{-1} + \AntiHolCurrModetheta{-1} \big) 
\fDGFFGround$
and we let $\contHorEdgeField\in\FullFock$ be its image via the isomorphism of Theorem~\ref{thm: isomorphism}.

\begin{ex}
Let $\zd^\meshsize$ be the closest vertex to $z\in\domain$ in the discrete domain~$\ddomain{\meshsize}$.
We have
\begin{align*}
    \frac{1}{\meshsize^2}
    \bigg(\,
    \P^{\textnormal{UST}}_{\ddomain{\meshsize}}
    \big[ \, \!\bullet\!\!\!\mathbf{-} \text{ open at } \zd^\meshsize \, \big]
    \,-\, \frac{1}{2}
    \, \bigg)
    \mspace{10mu}
    \xrightarrow{\mspace{20mu} \meshsize\downarrow0 \mspace{20mu}}
    \mspace{10mu}
    \BigSFCorrFun{\domain}{\thectt}{
    \contHorEdgeField(z)
    }
    \,=\,
    (\partial_1 + \partialBar_1)
    (\partial_2 + \partialBar_2)\,
    \HarmOfGreen_{\domain}(z,z)
\end{align*}
uniformly for $z$ on compacts of $\domain$.
Note again that the limit does not depend on $\thectt$.
\hfill$\diamond$\end{ex}

An interesting observable is the difference $\mathsf{P}_{\OEright}-\mathsf{P}_{\OEleftsubs}$.
Within correlation functions, one can think of this observable as the horizontal derivative of either $\mathsf{P}_{\OEright}$ or $\mathsf{P}_{\OEleftsubs}$ since it is the difference of the open-edge field among two horizontally neighbouring edges.
In the basis~\eqref{eq: basis}, we have
\begin{align*}
    \mathsf{P}_{\OEright} - \mathsf{P}_{\OEleftsubs}
    \, = \,
    \frac{1}{4\pi}
    \Big[
    \big( \HolCurrModexi{-2} + \AntiHolCurrModexi{-2} \big)
    \big( \HolCurrModetheta{-1} + \AntiHolCurrModetheta{-1} \big)
    \,+\,
    \big( \HolCurrModexi{-1} + \AntiHolCurrModexi{-1} \big)
    \big( \HolCurrModetheta{-2} + \AntiHolCurrModetheta{-2} \big)
    \Big]
    \fDGFFGround\,.
\end{align*}
As expected from its physical interpretation, this field has scaling dimension $3$ since it is the (horizontal) derivative of a field with scaling dimension $2$.

Using explicit formulae of the Sugawara construction ---Lemma~\ref{lemma: sugawara}---, it is straightforward to check that we have
\begin{align*}
    \mathsf{P}_{\OEright} - \mathsf{P}_{\OEleftsubs}
    \, = \,
    \big(\dL{-1}+\dLBar{-1}\big) \, \horEdgeField
    \,.
\end{align*}
This fact is in agreement with the postulates of CFT:
the Virasoro modes~$\dL{-1}$ and $\dLBar{-1}$ behave, within correlation functions, as the holomorphic and antiholomorphic Wirtinger derivatives respectively---their linear combination $\dL{-1}+\dLBar{-1}$ corresponds then to the derivative in the horizontal direction.

\vspace{-10pt}
\subsubsection{The Abelian sandpile model}
\label{subsubsec: ASM}
The model at hand\footnote{For simplicity, we consider here only the Abelian sandpile model with dissipative (or wired, or Dirichlet) boundary conditions.} was first introduced as a model for self-organized criticality in \cite{BTW-ASM}. 
Sandpile models are typically defined via a Markov chain.
Consider the set $(\Zpos)^{\interior{\ddomain{\meshsize}}}$ of positive-integer-valued height functions on the interior of $\ddomain{\meshsize}$.
We say two height functions $h,h'\in(\Zpos)^{\interior{\ddomain{\meshsize}}}$ are toppling-equivalent if they satisfy
\begin{align*}
    h'
    \, = \,
    h
    +
    \sum_{\zd \in \ddomain{\meshsize}}
    n_{\zd} T_{\zd}
\end{align*}
for some $n_{\zd}\in\Z$, where the \term{toppling} $T_{\zd}:^{\interior{\ddomain{\meshsize}}}\rightarrow\Z$ is defined via the Dirichlet Laplacian by
\begin{align*}
    T_{\zd}(\zdbis)
    \,\coloneqq\,
    (\Delta_{\ddomain{\meshsize}}^{\texttt{D}})_{\zd,\zdbis}\,.
\end{align*}

\vspace{-10pt}
The state space of the Markov chain is the set of height func\-tions $(\Z_{>0})^{\ddomain{\meshsize}}$ modulo topplings, which we denote by $\Heights{\ddomain{\meshsize}}$.
One can see that each class of toppling-equivalent height functions has a unique representative in $\{1,2,3,4\}^{\interior{\ddomain{\meshsize}}}$;
we hence identify each class with this unique representative.
The dynamics of the Markov chain is then as follows:
at each discrete time step, the height function increases by $1$ at a vertex selected uniformly at random among all the vertices in $\interior{\ddomain{\meshsize}}$.

The unique stationary measure of this Markov chain $\P_{\ddomain{\meshsize}}^{\text{ASM}}$ is uniform on the set $\Recurrent{\ddomain{\meshsize}}$ of re\-current states ---\cite{Dhar-ASM}---.
We let $\E_{\ddomain{\meshsize}}^{\text{ASM}}$ denote the expectation value with respect to this probability measure.

The \term{Abelian sandpile model} is then the probabilistic model whose space of configurations is $\Recurrent{\ddomain{\meshsize}}$ equipped with the uniform measure $\P_{\ddomain{\meshsize}}^{\text{ASM}}$.

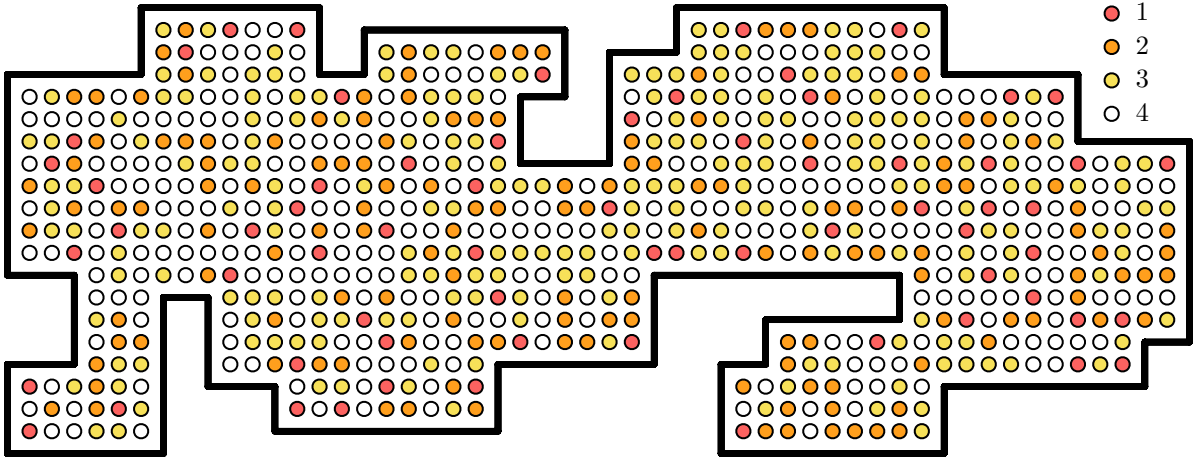
\begin{figure}[t!]
\centering
\begin{overpic}[scale=0.467, tics=10]{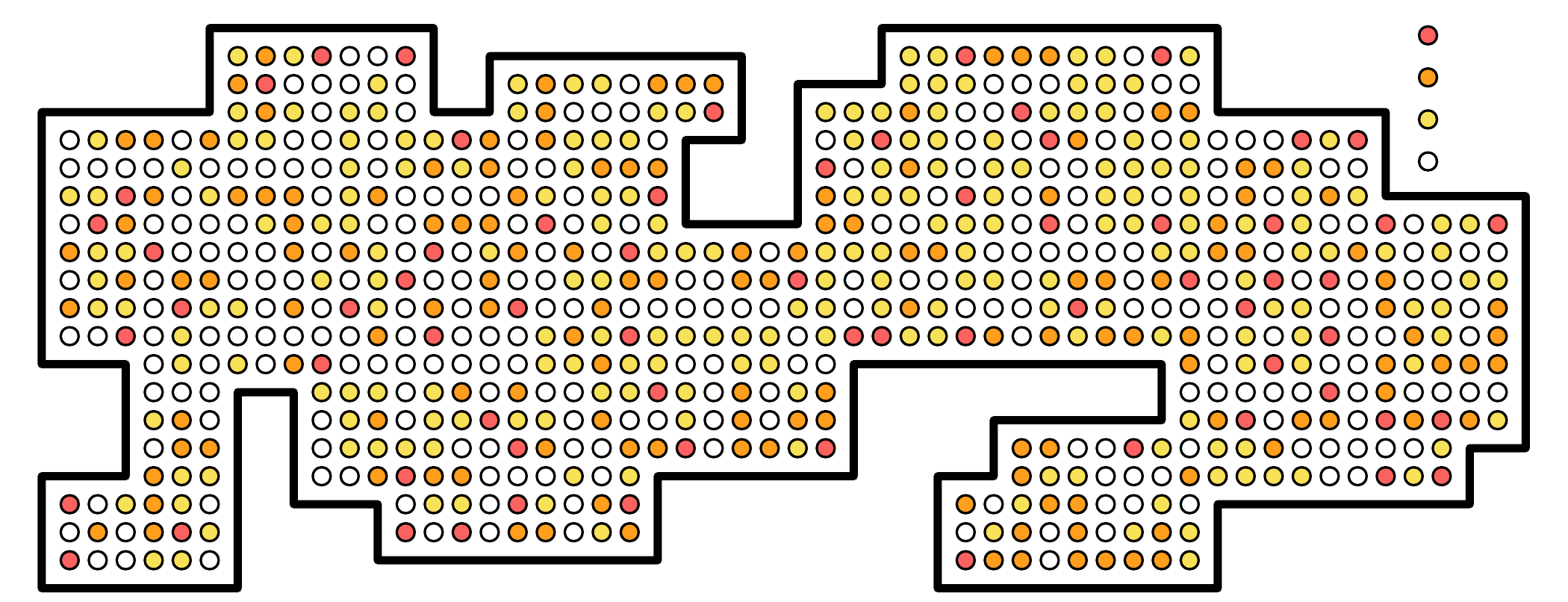}
    \put(93,35.6){\footnotesize $1$}
    \put(93,32.9){\footnotesize $2$}
    \put(93,30.2){\footnotesize $3$}
    \put(93,27.5){\footnotesize $4$}
\end{overpic}
\centering
\caption{ A configuration of the Abelian sandpile model with wired \\ boundary conditions in a discrete domain.
This configuration corresponds \\ to the spanning tree in Figure~\ref{fig: ust} via Dhar's burning algorithm.
}
\label{fig: asm}
\end{figure}

It is noteworthy that there is a non-local coupling ---that is, a bijection that preserves measures--- between the ASM and the UST known as the \textit{burning algorithm} ---\cite{MD-burning algorithm}---.
The (scaling limit of the) Abelian sandpile model is also conjectured to be described by a logarithmic CFT with central charge $-2$; see \cite{Ruelle-SM in the large} for a recent review.
Though the scaling-limit CFT is predicted not to be equivalent to the symplectic fermions, certain (local) fields in these two distinct logCFTs can be put in correspondence.

\vspace{-15pt}
\paragraph{The height-one field.}
Consider the observable
\begin{align*}
    \zd 
    \;\longmapsto \;
    \indicator_{\{h(\zd)=1\}}
\end{align*}
which takes values $1$ on ASM configurations in which the height at $\zd\in\ddomain{\meshsize}$ is $1$ and takes the value $0$ otherwise.
Its full-plane average  ---computed for the first time in \cite{MD-correlations}--- is $\frac{2(\pi-2)}{\pi^3}$, and its fluctuations around this average have been argued to be described by a local field in the symplectic fermions theory ---\cite{MahieuRuelle-observables}---.

By means of the fDGFF ---\cite[Theorem 3.5]{CCRR-fDGFF_sandpile_UST}---, we can compute the correlation functions 
\begin{align*}
    \E^{\textnormal{ASM}}_{\ddomain{\meshsize}}
    \big[\, \indicator_{\{h(\zd_1)=1\}} \cdots \indicator_{\{h(\zd_n)=1\}}
    \,\big]
    \;=\;
    \BigfDGFFCorrFun{\ddomain{\meshsize}}
    {H_1(\zd_1)
    \cdots
    H_1(\zd_n)
    }
\end{align*}
with the fDGFF field-polynomial
\begin{align*}
    H_1(\zd)
    \,\coloneqq\,
    \frac{1}{4\pi}
    \Bigg[ \,
    \prod_{e\sim \zd}
    \Big(\dgrad\fDGFFxi(e)\dgrad\fDGFFtheta(e) - 1 \Big)
    - 1 \, \Bigg]
    \,.
\end{align*}
We hence define the local observable~$\heightOne$ by
\begin{align*}
    \heightOne(\zd)
    \coloneqq
    \indicator_{\{h(\zd)=1\}}
    -
    \frac{2(\pi-2)}{\pi^3}\,,
\end{align*}
which is called the \term{height-one field}.
Its associated fDGFF local field is
\begin{align*}
    \dheightOne
    \,\coloneqq\,
    \frac{1}{4\pi}
    \Bigg[
    \prod_{x\in\{\pm 1, \pm\ii\}} \!\!
    \Big( \dgrad\Fieldxi \big( \tfrac{x}{2} \big)
    \dgrad\Fieldtheta \big( \tfrac{x}{2} \big)
    - 1 \Big) - 1 \,\Bigg]
    - \frac{2\,(\pi-2)}{\pi^3}
    \,+\, \Null\,,
\end{align*}
which, in terms of the CFT basis of the logarithmic Fock space ---Theorem~\ref{thm: isomorphism}--- is written as
\begin{align*}
    \dheightOne
    \;=\;
    -\frac{\pi-2}{\pi^3} \,
    \Big( \,\AntiHolCurrModexi{-1} \HolCurrModetheta{-1} + \HolCurrModexi{-1}\AntiHolCurrModetheta{-1} \Big)
    \fDGFFGround
    \mspace{15mu} 
    +
    \parbox{10em}{\centering terms with higher scaling dimensions.}
\end{align*}
We refer the reader to Section~\ref{subsec: ex of computation} to see an explicit example of such a computation.

We also define $\contHeightOne\coloneqq -\frac{\pi-2}{\pi^2}\big(\, \AlgChiBar{-1} \AlgEta{-1} + \AlgChi{-1}\AlgEtaBar{-1} \big) \Ground$ in the logarithmic Fock space~$\FullFock$---the image of $\dheightOne$ via the isomorphism in Theorem~\ref{thm: isomorphism}.
This local field coincides with the one predicted in~\cite{MahieuRuelle-observables}.

\begin{ex}
By Theorem~\ref{thm: scaling limit}, we have the following scaling-limit convergence of the one-point function:
\begin{align*}
    \frac{\E_{\ddomain{\meshsize}}^\textnormal{ASM}
    \big[\, \heightOne(\zd^\meshsize)\, \big]}{\meshsize^2}
    \mspace{10mu}
    \xrightarrow{\mspace{30mu}\meshsize\downarrow 0 \mspace{30mu}}
    \mspace{10mu}
    \BigSFCorrFun{\domain}{\thectt}{
    \contHeightOne (z)}
    \,=\,
    -
    \frac{8\,(\pi-2)}{\pi^2}\,
    \partialBar_1 \partial_2
    \HarmOfGreen_{\domain}(z,z)
\end{align*}
uniformly for $z$ on compacts of $\domain$, where $\zd^\meshsize$ is the vertex in $\ddomain{\meshsize}\subset\C$ that is closest to $z$.
Note the independence with respect to the constant $\thectt$.
\hfill$\diamond$
\end{ex}

\vspace{-10pt}
\paragraph{The dissipation field.}

Another interesting observable in this model is the insertion of dissipative sites in the bulk, that is, sites that do not topple.
This observable was fit into the CFT description of the ASM as a logarithmic local field in \cite{PirouxRuelle-observables}---we should revisit these ideas here.

The dissipation field is \emph{not} probabilistic in nature, but it rather modifies the configuration space of the ASM.
In precise terms, the configurations of the ASM on $\ddomain{\meshsize}$ with dissipative sites at the vertices $\zdbis_1,\ldots,\zdbis_n\in\ddomain{\meshsize}$ are the configurations of the ASM on $\ddomain{\meshsize}\setminus\{\zdbis_1,\ldots,\zdbis_n\}$.

Let $\mathbf{W}\coloneqq\{\zdbis_1,\ldots,\zdbis_n\}$ denote the set of dissipative sites in the bulk.
For a random variable $X$ of the ASM on $\ddomain{\meshsize}$ that is determined by the values of $h\vert_{\ddomain{\meshsize}\setminus\mathbf{W}}$, we introduce the following notation for expectations with dissipative sites:
\begin{align*}
    \E^{\textnormal{ASM}}_{\ddomain{\meshsize}}
    \big[\,
    X
    \,\bigASMseparator\,
    \dissipation(\zdbis_1) \cdots \dissipation(\zdbis_n)
    \,\big]
    \,\coloneqq\,
    \frac{1}{\vert \Recurrent{\ddomain{\meshsize}}\vert}
    \mspace{-30mu}
    \sum_{\mspace{50mu}h\in \Recurrent{\ddomain{\meshsize}\setminus\mathbf{W}}}
    \mspace{-30mu}
    X(h)\,.
\end{align*}
Note the sum runs over ASM configurations on $\ddomain{\meshsize}\setminus \mathbf{W}$, but the normalising constant is the number of ASM configurations on $\ddomain{\meshsize}$.
We refer to $\dissipation$ as the \term{dissipation field}.

\begin{ex}
We have
\begin{align*}
    \E^{\textnormal{ASM}}_{\ddomain{\meshsize}}
    \big[\,
    1
    \,\bigASMseparator\,
    \dissipation(\zdbis_1) \cdots \dissipation(\zdbis_n)
    \,\big]
    \,=\,
    \frac
    {\vert \Recurrent{\ddomain{\meshsize}\setminus\mathbf{W}} \vert}
    {\vert \Recurrent{\ddomain{\meshsize}}\vert}
    \,.
\end{align*}
where we let $1$ denote the constant random variable with constant value $1$.
\hfill$\diamond$
\end{ex}
We also define $\E^{\textnormal{ASM}}_{\ddomain{\meshsize}} [\, X \,\ASMseparator\, 1 \,] \coloneqq \E^{\textnormal{ASM}}_{\ddomain{\meshsize}} [\, X \,]$ and extend linearly in the following sense:
given the vertices $\xd_1,\ldots,\xd_\ell\in\ddomain{\meshsize}$ and the constants $\alpha,\alpha_1,\ldots,\alpha_\ell\in\C$,
we define
\begin{align*}
    & \,
    \E^{\textnormal{ASM}}_{\ddomain{\meshsize}}
    \big[\,
    X
    \,\bigASMseparator\,
    \bigg(
    \alpha \,+\,
    \sum_{j=1}^\ell \alpha_j \, \dissipation(\xd_j)
    \bigg)
    \dissipation (\zdbis_1) \cdots \dissipation(\zdbis_n)
    \,\big]
    \,\coloneqq\,
    \\
    & \mspace{100mu}
    \phantom{\bigg\vert}
    \alpha \cdot
    \E^{\textnormal{ASM}}_{\ddomain{\meshsize}}
    \big[\,
    X
    \,\bigASMseparator\,
    \dissipation (\zdbis_1) \cdots \dissipation(\zdbis_n)
    \,\big]
    \, + \,
    \sum_{j=1}^\ell
    \Big(
    \alpha_j \cdot
    \E^{\textnormal{ASM}}_{\ddomain{\meshsize}}
    \big[\,
    X
    \,\bigASMseparator\,
    \dissipation(\xd_j) \,
    \dissipation (\zdbis_1) \cdots \dissipation(\zdbis_n)
    \,\big]
    \Big)\,.
\end{align*}

The following result ---whose proof we postpone to Section~\ref{subsec: proof stat mech}--- allows one to compute the  correlation functions of an arbitrary number of height-one fields in the presence of an arbitrary number of dissipation fields.

\begin{lemma}
\label{lemma: corrfun h-one and diss}
Let $\zd_1,\ldots,\zd_n,\zdbis_1,\ldots,\zdbis_m\in\ddomain{\meshsize}$ be interior vertices of $\ddomain{\meshsize}$.
We have
\begin{align*}
    \E^{\textnormal{ASM}}_{\ddomain{\meshsize}}
    \big[\,
    \heightOne(\zd_1) \cdots \heightOne(\zd_n)
    \,\bigASMseparator\,
    \dissipation(\zdbis_1) \cdots \dissipation(\zdbis_m)
    \,\big]
    \,=\,
    \BigfDGFFCorrFun{\ddomain{\meshsize}}{
    H_1(\zd_1) \cdots H_1(\zd_n)
    D(\zdbis_1) \cdots D(\zdbis_m)
    }
\end{align*}
with $D(\zd) \coloneqq \dfrac{1}{4\pi}\fDGFFxi(\zd)\fDGFFtheta(\zd)$\,.
\end{lemma}

It is straightforward to see ---by the definitions in Example~\ref{ex: rep of ground states}--- that, in the basis \eqref{eq: basis}, the fDGFF field $\dissField \coloneqq \frac{1}{4 \pi} \Fieldxi(0)\Fieldtheta(0) + \Null$ associated to the dissipation field is written as
\begin{align*}
    \dissField
    \, = \,
    -  \frac{1}{4\pi}\fDGFFGroundPartner\,,
\end{align*}
which is a field with generalised scaling dimensions $0$.
However, we argue that the term of the logarithmic field $\fDGFFGroundPartner$ is the only relevant term to understand the physical behaviour of this field.
To see this, consider the statement of Theorem~\ref{thm: scaling limit} applied to the particular case of the correlation function in Lemma~\ref{lemma: corrfun h-one and diss}.

On the one hand, the height-one field was established to be a primary field with scaling dimension $2$---no logarithmic correction is involved in the scaling limit of this field.
As for the dissipation field, a field-polynomial representative of the field $\dissField - \log ( \lambda \, \meshsize ) \big[ (\dL{0}+\dLBar{0}) - 0 \, \big]\dissField$ is
\begin{align*}
    \frac{1}{4 \pi} \Fieldxi(0)\Fieldtheta(0)
    \, + \, \frac {\log (\lambda\,\meshsize)}{2\pi}\,.
\end{align*}
According to Theorem~\ref{thm: scaling limit}, we have
\begin{align*}
    & \frac{1}{{{\mathlarger{\meshsize}}}^{2 \cdot m}} \,
    \E^{\textnormal{ASM}}_{\ddomain{\meshsize}}
    \Big[\,
    \heightOne(\zd^\meshsize_1) \cdots \heightOne(\zd^\meshsize_m)
    \,\bigASMseparator\,
    \Big(\dissipation(\zdbis^\meshsize_1) + \tfrac {\log (\lambda\,\meshsize)}{2\pi} \Big)
    \cdots
    \Big(\dissipation(\zdbis^\meshsize_n) + \tfrac {\log (\lambda\,\meshsize)}{2\pi} \Big)
    \,\Big]
    \phantom{\Bigg\vert}
    \\
    & \mspace{210mu}
    \xrightarrow{\mspace{50mu} \meshsize\downarrow 0 \mspace{50mu}}
    \mspace{3mu} \phantom{\Bigg\vert}
    \frac{(-1)^n}{(4\pi)^n}\,
    \BigSFCorrFun{\domain}{\thectt(\lambda)}{
    \,\contHeightOne(z_1) \cdots \contHeightOne (z_m) \,
    \GroundPartner(w_1) \cdots \GroundPartner (w_n)
    }
\end{align*}
uniformly for $(z_1,\ldots,z_m,w_1,\ldots,w_n)$ on compacts of $\Conf{n+m}{\domain}$.

Note, once more, that the scaling constant $\lambda$ is arbitrary, and the scaling limit is unique only up to a choice of this constant.

\begin{ex}
\label{ex: dissipation field}
For one height-one field and one dissipation field, we have 
\begin{align*}
    & \frac{1}{{{\mathlarger{\meshsize}}}^{2}} \,
    \E^{\textnormal{ASM}}_{\ddomain{\meshsize}}
    \Big[\,
    \heightOne(\zd^\meshsize)
    \,\bigASMseparator\,
    \dissipation(\zdbis^\meshsize) + \tfrac {\log (\lambda\,\meshsize)}{2\pi}
    \,\Big]
    \mspace{10mu}
    \xrightarrow{\mspace{20mu} \meshsize\downarrow 0 \mspace{20mu}}
    \mspace{10mu} \phantom{\Bigg\vert}
    -\frac{1}{4\pi}\,
    \BigSFCorrFun{\domain}{\thectt(\lambda)}{
    \,\contHeightOne(z) \,
    \GroundPartner(w)
    }
    \\
    & \mspace{140mu}
    \phantom{\bigg\vert}
    = \, \frac{8\,(\pi-2)\,}{\pi^2}\,
    \Big(
    \partialBar_1 \partial_2 \HarmOfGreen_{\domain}(z,z)
    \big[\HarmOfGreen_\domain(w,w) + \thectt(\lambda) \big]
    +
    \partialBar_1 \HarmOfGreen_\domain(z,w) \,
    \partial_1 \HarmOfGreen_\domain(z,w)
    \Big)
\end{align*}
uniformly for $(z,w)$ in compacts of $\Conf{2}{\domain}$,
where $\zd^\meshsize$ and $\zdbis^\meshsize$ are the vertices in $\ddomain{\meshsize}$ that are closest to $z$ and $w$ respectively.
\hfill$\diamond$
\end{ex}
\black


\newpage
\section{Proofs}
\label{sec: proofs}
We now present the proofs that were omitted in the discussion of our results.

In Sections~\ref{subsec: lin field poly} and \ref{subsec: higher fields} we elaborate on preliminaries necessary for the proof of Theorem~\ref{thm: isomorphism}, which is presented in Section~\ref{subsec: proof of isomorphism}.
In Section~\ref{subsec: proof of scaling limit}, we give the proof of Theorem~\ref{thm: scaling limit}.
Finally, Section~\ref{subsec: proof stat mech} is devoted to the proof of Lemma~\ref{lemma: corrfun h-one and diss} regarding the dissipation field of the Abelian sandpile model.
In the final Section~\ref{subsec: ex of computation}, we give the details of one computation for expressing a local field in terms of the BPZ basis of the Fock space.

\vspace{-10pt}
\subsection{Linear local fields}
\label{subsec: lin field poly}
Before handling all fDGFF local fields, we consider a simple subspace of local fields that can be characterised more easily.
This subspace is the set of linear local fields, that is, fields that have a homogeneous field-polynomial representative of degree $1$.
The key insight is that one can build the rest of local fields via a certain product ---the normally-ordered product--- of linear local fields; this idea is developed in Section~\ref{subsec: higher fields}.

We define the space of \term{linear field polynomials} as
\begin{align*}
	\LinearFieldPoly
	\,\coloneqq\,
	\spn_\C
	\big\{\Fieldxi(\zu),\Fieldtheta(\zu)
	\,\colon \zu\in\Z^2
	\big\}\,,
\end{align*}
which is a vector subspace of $\FieldPoly$.
Then, the space of \term{linear local fields} is the quotient
\begin{align*}
	\LinearFields
	\,\coloneqq\,
	\LinearFieldPoly / \NullFieldPoly\,.
\end{align*}

\begin{ex}
\label{ex: fermion linear rep}
The fields
\begin{align*}
	\fDGFFGroundXi
	\,\coloneqq\,
	\Fieldxi (0) \,+\,\Null
	\mspace{60mu}
	\text{and}
	\mspace{60mu}
	\fDGFFGroundTheta
	\,\coloneqq\,
	\Fieldtheta (0) \,+\,\Null
\end{align*}
transparently belong in $\LinearFields$.
We define their representatives $\RepHolChi{0}\coloneqq \Fieldxi(0)$ and $\RepHolEta{0}\coloneqq \Fieldtheta(0)$.
\hfill$\diamond$
\end{ex}

\begin{ex}
\label{ex: rep of linear basis}
The fields
\begin{align*}
	\HolCurrModexi{-k} \fDGFFGround
	\,=\, &
	\frac{1}{2\pi}
	\sum_{\zu_\medial\in\Z^2_\medial}
	\ddeebar \zu_\medial^{[-k]}
	\pddee\Fieldxi (\zu_\medial)
	\,+\,
	\Null
	\,\eqqcolon\,
	\RepHolChi{-k}
	\,+\,
	\Null\,,
	\\
	\AntiHolCurrModexi{-k} \fDGFFGround
	\,=\, &
	\frac{1}{2\pi}
	\sum_{\zu_\medial\in\Z^2_\medial}
	\ddee \overline{\zu}_\medial^{[-k]}
	\pddeebar\Fieldxi (\zu_\medial)
	\,+\,
	\Null
	\,\eqqcolon\,
	\RepAntiHolChi{-k}
	\,+\,
	\Null\,,
	\\
	\HolCurrModetheta{-k} \fDGFFGround
	\,=\, &
	\frac{1}{2\pi}
	\sum_{\zu_\medial\in\Z^2_\medial}
	\ddeebar \zu_\medial^{[-k]}
	\pddee\Fieldtheta (\zu_\medial)
	\,+\,
	\Null
	\,\eqqcolon\,
	\RepHolEta{-k}
	\,+\,
	\Null\,,
	\mspace{20mu}
	\text{and}
	\\
	\AntiHolCurrModetheta{-k} \fDGFFGround
	\,=\, &
	\frac{1}{2\pi}
	\sum_{\zu_\medial\in\Z^2_\medial}
	\ddee \overline{\zu}_\medial^{[-k]}
	\pddeebar\Fieldtheta (\zu_\medial)
	\,+\,
	\Null
	\,\eqqcolon\,
	\RepAntiHolEta{-k}
	\,+\,
	\Null
\end{align*}
are linear local fields.
That is, they have linear field polynomial representatives.
\hfill$\diamond$
\end{ex}

\begin{lemma}
\label{lemma: linear basis}
The set
\begin{align*}
	\big\{
	\HolCurrModexi{-k} \fDGFFGround,
	\AntiHolCurrModexi{-k} \fDGFFGround,	\HolCurrModetheta{-k} \fDGFFGround,
	\AntiHolCurrModetheta{-k} \fDGFFGround,
	\,\colon
	k\in\Zpos
	\big\}
	\,\cup\,
	\big\{
	\fDGFFGroundTheta,
	\fDGFFGroundXi
	\big\}
\end{align*}
constitutes a basis of the space $\LinearFields$.
\end{lemma}
\begin{proof}
The proof follows the same lines as the proof of Theorem~5.1 in \cite{ABK-DGFF_local_fields} with the addition of the non-zero average linear fields $\fDGFFGroundTheta$ and $\fDGFFGroundXi$.
\end{proof}

In a similar fashion as in \cite{ABK-DGFF_local_fields} the previous lemma leads to a complete classification of the linear null fields of the fDGFF.

\begin{coro}
\label{coro: charac null linear}
The set
\begin{align*}
    \big\{
    \dlaplacian \Fieldxi ( \zu ) ,
    \dlaplacian \Fieldtheta ( \zu )
    \,\colon
    \zu \in \Z^2
    \big\}
\end{align*}
spans the subspace $\NullFieldPoly \cap \LinearFieldPoly$ of null linear local field polynomials.
\end{coro}

\vspace{-10pt}
\subsection{Higher-degree local fields}
\label{subsec: higher fields}
We now define a way of obtaining higher-degree field polynomials as a normally-ordered product of linear field polynomials,
which we have completely characterised in Section~\ref{subsec: lin field poly}.
To address this task, we need the \textbf{full-plane square-grid Green's function} ---\cite{LL-random_walk}---,
that is, the unique function $\Green_{\Z^2} \colon \Z^2\longrightarrow \R$
that satisfies
\begin{align*}
    \dlaplacian \Green_{\Z^2} (\zu)  \, = \, - \delta_{\zu,0}
    \mspace{20mu}
    \text{ and }
    \mspace{20mu}
    \Green_{\Z^2} (0) \, =  \, 0\,,
\end{align*}
and has the asymptotic behaviour
\begin{align*}
    \mspace{150mu}
    \Green_{\Z^2}(\zu) & \, = 
    -
    \frac{1}{2\pi} \log \vert \zu \vert
    -
    \frac{\mathsf{C}}{2\pi}
    +
    \mathcal{O} \big( \vert \zu \vert^{-2} \big)
    \mspace{50mu}
    \text{ as }
    \mspace{5mu}
    \vert \zu \vert \to \infty\,,
\end{align*}
with $\mathsf{C} = \gamma + \frac{3}{2} \log 2$, where $\gamma$ is the Euler--Mascheroni constant.

Via the full plane Green's function~$\Green_{\Z^2}$, we define the \term{Wick contractions} of linear field polynomials
\begin{align*}
\LinearFieldPoly \otimes \LinearFieldPoly \; \longrightarrow \; &\, \C \\
L_1 \otimes L_2 \; \longmapsto \; &\, \wick{\c L_1 \c L_2}
\end{align*}
by bilinear extension of the formulae
\begin{align*}
    \wick{\c \Fieldxi(\zu) \c \Fieldtheta(\zubis)}
    \, \coloneqq \,
    4\pi\,\Green_{\Z^2} (\zu-\zubis)
    \, \eqqcolon \,
    -\wick{\c \Fieldtheta(\zu) \c \Fieldxi(\zubis)}
    \mspace{30mu}
    \text{and}
    \mspace{30mu}
    \wick{\c \Fieldtheta(\zu) \c \Fieldtheta(\zubis)}
    \, = \,
    \wick{\c \Fieldxi(\zu) \c \Fieldxi(\zubis)}
    \, = \,
    0
\end{align*}
for $\zu,\zubis\in\Z^2$.
Note that Wick contractions are antisymmetric, i.e.~we have $\wick{\c L_1 \, \c L_2} = - \wick{\c L_2 \, \c L_1}$.

\begin{ex}\label{ex: Wick cont of Lap}
We have $\wick{\dlaplacian\c\Fieldxi(\zu) \c\Fieldtheta(\zubis)}
= \wick{\c\Fieldxi(\zu) \dlaplacian\c\Fieldtheta(\zubis)}
  = -4 \pi \delta_{\zu,\zubis}$.
In particular, we have
\begin{align*}
    \BigfDGFFCorrFun{\ddomain{\meshsize}}{
    \ev^{\ddomain{\meshsize}}_{\zd}
    \big( \dlaplacian\Fieldxi(\zu) \Fieldtheta(\zubis)\big)\,}
    - \wick{ \dlaplacian \c\Fieldxi(\zu) \c\Fieldtheta(\zubis)}
    \,=\,
    0
\end{align*}
and 
\begin{align*}
    \BigfDGFFCorrFun{\ddomain{\meshsize}}{
    \ev^{\ddomain{\meshsize}}_{\zd}
    \big( \Fieldxi(\zu)\dlaplacian\Fieldtheta(\zubis)\big)\,}
    - \wick{\c\Fieldxi(\zu) \dlaplacian\c\Fieldtheta(\zubis)}\,
    \,=\,
    0\,,
\end{align*}
for any discrete domain~$\ddomain{\meshsize}$ and point $\zd\in\ddomain{\meshsize}$.
\hfill
$\diamond$
\end{ex}

\term{Normal ordering} is a linear map $\FieldPoly \longrightarrow \FieldPoly$.
However, we conceive it as a linear map
\begin{align*}
    \no{\cdots} \colon
    \bigoplus_{d=0} \big(\LinearFieldPoly\big)^{\otimes d}
    \longrightarrow
    \FieldPoly
\end{align*}
by the natural identification of $\bigoplus_{d=0} \big(\LinearFieldPoly\big)^{\otimes d}$ with $\FieldPoly$.
It is defined as follows:
for linear field polynomials $L_1,\cdots,L_n\in\LinearFieldPoly$,
we define
\begin{align*}
    \no{ L_1 \cdots L_n }
    \coloneqq
    \sum_{P\in\Pair_n}
    (-1)^P
    (-1)^{\vert P \vert}
    \prod_{(i,j)\in P}
    \wick{\c L_i \c L_j}
    \prod_{k\in [n]\setminus \cup P}^{\rightarrow}
    L_k
\end{align*}
where $\Pair_n$ is the set of (ordered) \term{partial pairings} of the set $[n]=\{1,\ldots,n\}$, that is, an element of this set is $P=\big\{(i_1,j_1),\ldots,(i_{\vert P \vert},j_{\vert P \vert})\big\}$ with $i_p < j_p$ and such that the set $\cup P \coloneqq \{i_1,j_1,\ldots,i_{\vert P \vert},j_{\vert P \vert} \}$ is of order $2\vert P\vert $.
We also write $[n]\setminus\cup P = \{ k_1,\ldots,k_{n-\vert P \vert}\}$ with $k_\alpha < k_\beta$ for $\alpha < \beta$.
Then, we let $(-1)^P$ denote the signature of the permutation $(1\;2\;\cdots\; n)\mapsto(i_1\  j_1 \, \cdots \, i_{\vert P \vert}\  j_{\vert P \vert}\  k_1\, \cdots \, k_{n-\vert P \vert} )$.
Finally, the notation $\prod^\rightarrow$ indicates the product is taken in increasing order of the indices, that is, we have
$\prod_{k\in [n]\setminus \cup P}^{\rightarrow}L_k = L_{k_1}\,L_{k_2}\cdots L_{k_{n-\vert P \vert}}$.

\begin{ex}
\label{ex: norm ord of two fields}
Given two linear field polynomials $L_1,L_2\in\LinearFieldPoly$, their normal ordering is given by $\no{ L_1 \, L_2 } = L_1 \, L_2 - \wick{\c L_1 \, \c L_2}$.
In particular, by Example~\ref{ex: Wick cont of Lap}, we have
\begin{align*}
    \BigfDGFFCorrFun{\ddomain{\meshsize}}{
    \ev^{\ddomain{\meshsize}}_{\zd}
    \big( \no{\dlaplacian\Fieldxi(\zu)\Fieldtheta(\zubis)}\big)\,}
    \,=\, 0
    \mspace{50mu}
    \textnormal{ and }
    \mspace{50mu}
    \BigfDGFFCorrFun{\ddomain{\meshsize}}{
    \ev^{\ddomain{\meshsize}}_{\zd}
    \big( \no{\Fieldxi(\zu)\dlaplacian\Fieldtheta(\zubis)}\big)\,}
    \,=\, 0
\end{align*}
for any discrete domain~$\ddomain{\meshsize}$ and point $\zd\in\ddomain{\meshsize}$.
\hfill$\diamond$
\end{ex}

\begin{rmk}
\label{rmk: recursive formula norm ord}
The normal ordering of $n>2$ linear field polynomials $L_1,\ldots,L_n\in\LinearFieldPoly$ can be inductively proven to satisfy the formula
\begin{align*}
    \no{ L_n \, L_{n-1} \cdots L_1 }
    \,=\,
    L_n \, \no{ L_{n-1} \cdots L_1 }
    \,+\,
    \sum_{i=1}^{n-1}
    (-1)^i\,
    \wick{\c L_n \, \c L_{n-i}}\,
    \no{ L_{n-1} \cdots \widehat{L}_{n-i} \cdots L_1 }\,,
\end{align*}
where the hat ($\,\widehat{\phantom{m}}\;\!$) indicates the factor is not in the product.
\hfill$\diamond$
\end{rmk}

\begin{rmk}
\label{rmk: norm ord surjective}
The map $\no{\cdot}\colon \FieldPoly\longrightarrow\FieldPoly$ is surjective.
To see this, note that, for a field polynomial~$P$ of degree~$d$, the field polynomial $\no{P}-P$ is of degree at most $d-2$.
That is, the map $\no{\cdot}$ is upper triangular with $1$s on the diagonal.
\hfill$\diamond$
\end{rmk}

\begin{rmk}\label{rmk: norm ord factors through null fields}
Let $L_1,\ldots,L_n\in\LinearFieldPoly$ be linear field polynomials.
If $L_i$ is null for some $i\in\{1,\ldots,n\}$, then $\no{L_1 \cdots L_n}$ is a null field polynomial too.
To prove this, one first needs to expand the expression of the correlation function
\begin{align*}
    \BigfDGFFCorrFun{\ddomain{\meshsize}}{
    \ev^{\ddomain{\meshsize}}_{\zd}
    \big( \no{L_1 \cdots L_n} \big)\,
    \fDGFFxi(\zd_1)\,\fDGFFtheta(\zdbis_1)
    \cdots
    \fDGFFxi(\zd_m)\,\fDGFFtheta(\zdbis_m)
    }
\end{align*}
by means of Wick's formula ---Remark~\ref{rmk: 2n pt function fDGFF}--- and the definition of normal ordering.
The proof eventually boils down to the observation that we have
\begin{align*}
    \BigfDGFFCorrFun{\ddomain{\meshsize}}{
    \ev^{\ddomain{\meshsize}}_{\zd}
    \big( \no{L_1L_2}\big)\,}
    \,=\,0
\end{align*}
whenever $L_1$ or $L_2$ are null linear fields---recall the characterisation of null linear fields in Corollary~\ref{coro: charac null linear} and the Example~\ref{ex: norm ord of two fields}.
We refer the reader to \cite[Lemma~6.2]{ABK-DGFF_local_fields} for a detailed proof of this statement in the bosonic setting.
\hfill$\diamond$
\end{rmk}

In other words, Remark~\ref{rmk: norm ord factors through null fields} asserts that the map $\no{\cdots}$ factors through the quotient with null fields.
In particular, it means that,
for $F_i\in\LinearFields$ linear local fields with linear representatives $L_i\in\LinearFieldPoly$ ---that is, we have $F_i = L_i + \Null$---,
the map
\begin{align*}
    F_1 \otimes \cdots \otimes F_n
    \;\,\longmapsto\;
    \no{L_1\,\cdots\,L_n} \,+\, \Null
    \,\eqqcolon\,
    \noQuo{F_1 \cdots F_n}
\end{align*}
is well-defined, that is, the normally ordered product does not depend on the choice of linear representatives.

The following lemma is the last piece we need for the proof of Theorem~\ref{thm: isomorphism}.
Given the basis~\eqref{eq: basis}, it essentially asserts that all local fields in $\Fields$ can be constructed as the normally-ordered product of enough linear fields.

\begin{lemma}
\label{lemma: normal order of basis}
Denote by $T$ the tensor product
\begin{align*}
\HolCurrModexi{-k_1}\fDGFFGround
\otimes \cdots \otimes
\HolCurrModexi{-k_n}\fDGFFGround
\otimes
\HolCurrModetheta{-\ell_1}\fDGFFGround
\otimes \cdots \otimes
\HolCurrModetheta{-{\ell}_m}\fDGFFGround 
\otimes
\AntiHolCurrModexi{-k_1}\fDGFFGround
\otimes \cdots \otimes
\AntiHolCurrModexi{-k_{\overline{n}}}\fDGFFGround
\otimes
\AntiHolCurrModetheta{-\overline{\ell}_1}\fDGFFGround
\otimes \cdots \otimes
\AntiHolCurrModetheta{-\overline{\ell}_{\overline{m}}}\fDGFFGround
\end{align*}
of linear fields.
Then we have
\begin{align*}
    \noQuo{\,T\,}
    \,=\, & \phantom{\Big\vert}
    \HolCurrModexi{-k_1}
    \cdots
    \HolCurrModexi{-k_n}
    \HolCurrModetheta{-\ell_1}
    \cdots
    \HolCurrModetheta{-{\ell}_m} 
    \AntiHolCurrModexi{-k_1}
    \cdots
    \AntiHolCurrModexi{-k_{\overline{n}}}
    \AntiHolCurrModetheta{-\overline{\ell}_1}
    \cdots
    \AntiHolCurrModetheta{-\overline{\ell}_{\overline{m}}}\fDGFFGround\,,
    \\
    \noQuo{\,T \otimes { \fDGFFGroundTheta}\,}
    \,=\, & \phantom{\Big\vert}
    \HolCurrModexi{-k_1}
    \cdots
    \HolCurrModexi{-k_n}
    \HolCurrModetheta{-\ell_1}
    \cdots
    \HolCurrModetheta{-{\ell}_m} 
    \AntiHolCurrModexi{-k_1}
    \cdots
    \AntiHolCurrModexi{-k_{\overline{n}}}
    \AntiHolCurrModetheta{-\overline{\ell}_1}
    \cdots
    \AntiHolCurrModetheta{-\overline{\ell}_{\overline{m}}}
    \fDGFFGroundTheta\,
    \\
    \noQuo{\,T \otimes {\fDGFFGroundXi}\,}
    \,=\, & \phantom{\Big\vert}
    \HolCurrModexi{-k_1}
    \cdots
    \HolCurrModexi{-k_n}
    \HolCurrModetheta{-\ell_1}
    \cdots
    \HolCurrModetheta{-{\ell}_m} 
    \AntiHolCurrModexi{-k_1}
    \cdots
    \AntiHolCurrModexi{-k_{\overline{n}}}
    \AntiHolCurrModetheta{-\overline{\ell}_1}
    \cdots
    \AntiHolCurrModetheta{-\overline{\ell}_{\overline{m}}}
    \fDGFFGroundXi\,,
    \mspace{15mu} \text{ and }
    \\
    \noQuo{\,T \otimes { \fDGFFGroundTheta \otimes \fDGFFGroundXi}\,}
    \,=\, & \phantom{\Big\vert}
    \HolCurrModexi{-k_1}
    \cdots
    \HolCurrModexi{-k_n}
    \HolCurrModetheta{-\ell_1}
    \cdots
    \HolCurrModetheta{-{\ell}_m} 
    \AntiHolCurrModexi{-k_1}
    \cdots
    \AntiHolCurrModexi{-k_{\overline{n}}}
    \AntiHolCurrModetheta{-\overline{\ell}_1}
    \cdots
    \AntiHolCurrModetheta{-\overline{\ell}_{\overline{m}}}
    \fDGFFGroundPartner\,.
\end{align*}
\end{lemma}

However, in the proof of this lemma, we use the following rather technical result regarding the linear local fields ---rather, their representatives--- in Examples~\ref{ex: rep of linear basis} and \ref{ex: fermion linear rep}.

\begin{lemma}
\label{lemma: technical}
Let $\gamma$ be a large enough corner contour encircling the origin.
We have,
\begin{align*}
    \sqint_\gamma
    \zu_\diamond^{[-k]} \,
    \wick{ \pddee \c\Fieldxi (\zu_\medial) \c\RepHolEtaNoIndex_{-\ell} }
    \,\dd \zu
    \,=\; 0 \; =\,
    \sqint_\gamma
    \overline{\zu}_\diamond^{[-k]} \,
    \wick{ \pddeebar \c\Fieldxi (\zu_\medial) \c\RepHolEtaNoIndex_{-\ell} }
    \,\dd \zu
\end{align*}
and
\begin{align*}
    \sqint_\gamma
    \zu_\diamond^{[-k]} \,
    \wick{ \pddee \c\Fieldtheta (\zu_\medial) \c\RepHolChiNoIndex_{-\ell} }
    \, \dd \zu
    \,=\; 0 \; = \,
    \sqint_\gamma
    \overline{\zu}_\diamond^{[-k]} \,
    \wick{ \pddeebar \c\Fieldtheta (\zu_\medial) \c\RepHolChiNoIndex_{-\ell} }
    \,\dd \zu
\end{align*}
for $k > 0$ and $\ell \geq 0$.
We also have
\begin{align*}
    \sqint_\gamma
    \zu_\diamond^{[-k]} \,
    \wick{ \pddee \c\Fieldxi (\zu_\medial) \c\RepAntiHolEtaNoIndex_{-\ell} }
    \,\dd \zu
    \,=\; 0 \; =\,
    \sqint_\gamma
    \overline{\zu}_\diamond^{[-k]} \,
    \wick{ \pddeebar \c\Fieldxi (\zu_\medial) \c\RepAntiHolEtaNoIndex_{-\ell} }
    \,\dd \zu
\end{align*}
and
\begin{align*}
    \sqint_\gamma
    \zu_\diamond^{[-k]} \,
    \wick{ \pddee \c\Fieldtheta (\zu_\medial) \c\RepAntiHolChiNoIndex_{-\ell} }
    \, \dd \zu
    \,=\; 0 \; = \,
    \sqint_\gamma
    \overline{\zu}_\diamond^{[-k]} \,
    \wick{ \pddeebar \c\Fieldtheta (\zu_\medial) \c\RepAntiHolChiNoIndex_{-\ell} }
    \,\dd \zu
\end{align*}
for $k > 0$ and $\ell > 0$.
\end{lemma}

\begin{proof}
We prove the case $\ell = 0$ for general $k > 0$.
Recall that we defined $X_0=\Fieldxi(0)$ in Example~\ref{ex: fermion linear rep}.
We have
\begin{align*}
    \sqint_\gamma
    \zu_\diamond^{[-k]} \,
    \wick{ \pddee \c\Fieldtheta (\zu_\medial) \c\RepHolChiNoIndex_{-\ell} }
    \, \dd \zu
    \, = \,
    -4\pi
    \sqint_\gamma
    \zu_\diamond^{[-k]} \,
    \pddee\Green_{\Z^2}(\zu_\medial)
    \, \dd \zu\,.
\end{align*}
Note that, away from the origin, both integrands are discrete holomorphic, so the above integral does not depend on the (large enough) discrete contour $\gamma$.
Recall that, as $\vert \zu \vert \to \infty$, we have $\zu^{[-k]}=\mathcal{O}\big(\vert\zu\vert^{-k}\big)$ and $\pddee\Green_{\Z^2}(\zu)=\mathcal{O}\big(\vert\zu\vert^{-1}\big)$.
Let $\gamma(R)$ be a discrete contour with a square shape of sidelength $R$.
The above integral has the asymptotic behaviour $\mathcal{O}\big(R^{-k}\big)$ as $ R \to \infty$, but, since it does not depend on $R$, taking the limit $R\to\infty$ we see that it  must vanish.
The case $\ell > 0$ can be handled by similar arguments---we refer the reader to \cite[Lemma~6.5]{ABK-DGFF_local_fields} for the details of these arguments in the bosonic setting.
\end{proof}

\begin{proof}[Proof of Lemma~\ref{lemma: normal order of basis}]
We write down the details only for the case $\overline n = \overline m = 0$;
the proof in the general case follows similar arguments.
We proceed by induction on $n+m$.
The linear base cases
\begin{align*}
    \noQuo{ \HolCurrModexi{-k} \fDGFFGround }
    \,=\,
    \HolCurrModexi{-k} \fDGFFGround\,,
    \mspace{40mu}
    \noQuo{ \HolCurrModetheta{-k} \fDGFFGround }
    \,=\,
    \HolCurrModetheta{-k} \fDGFFGround\,,
    \mspace{40mu}
    \noQuo{ \fDGFFGroundXi }
    \,=\,
    \fDGFFGroundXi\,,
    \mspace{20mu}
    \text{ and }
    \mspace{20mu}
    \noQuo{ \fDGFFGroundTheta }
    \,=\,
    \fDGFFGroundTheta
\end{align*}
are straightforward to check since normal ordering $\no{\cdots}$ is the identity map when restricted to linear field polynomials.
The value $\Green_{\Z^2}(0)=0$ implies that we also have
\begin{align*}
    \noQuo{ \fDGFFGroundTheta\otimes\fDGFFGroundXi }
    \,=\,
    \no{ \Fieldtheta(0)\Fieldxi(0)} + \Null
    \,=\,
    \Fieldtheta(0)\Fieldxi(0) + \,\Null
    \,=\,
    \fDGFFGroundPartner\,.
\end{align*}
For the rest of the proof we focus only on the last case ---that is, the $\noQuo{T\otimes\fDGFFGroundTheta\otimes\fDGFFGroundXi}$ case---, the other three can be handled in a similar manner.
By the induction hypothesis, we have
\begin{align*}
    \HolCurrModexi{-k_n}\cdots\HolCurrModexi{-k_1}
    \HolCurrModetheta{-\ell_m}\cdots\HolCurrModetheta{-{\ell}_1}
    \fDGFFGroundPartner
    =
    \no{ F \,}
    + \Null
\end{align*}
with 
$F \, \coloneqq \,
\RepHolChi{-k_n} \cdots \RepHolChi{-k_1}
\RepHolEta{-\ell_m} \cdots \RepHolEta{-\ell_1}
\RepHolEta{0} \, \RepHolChi{0}$\,,
where $\RepHolChi{k}$ and $\RepHolEta{\ell}$ are the field-polynomial representatives of the linear fields in Examples~\ref{ex: fermion linear rep} and \ref{ex: rep of linear basis}.
In what follows, we also use the notation
\begin{align*}
    F_0^{\hat \RepHolEtaNoIndex} &\; \coloneqq \,
    \RepHolChi{-k_n} \cdots \RepHolChi{-k_1}
    \RepHolEta{-\ell_m} \cdots \RepHolEta{-\ell_1}
    \RepHolChi{0}\,,
    \mspace{20mu}
    \text{ and }
    \\
    F_j^{\hat \RepHolEtaNoIndex} &\; \coloneqq \,
    \RepHolChi{-k_n} \cdots  \RepHolChi{-k_1}
    \RepHolEta{-\ell_m} \cdots \widehat{\RepHolEtaNoIndex}_{-\ell_j} \cdots \RepHolEta{-\ell_1}
    \RepHolEta{0} \, \RepHolChi{0}\,,
\end{align*}
where the hat ($\,\widehat{\phantom{m}}\;\!$) indicates the factor is not in the product.
We prove the inductive step for the induction in $n$---the induction in  $m$ can be handled almost identically.
Then, on the one hand, using the recursive formula in by Remark~\ref{rmk: recursive formula norm ord}, we have
\begin{align}
\label{eq: one hand}
    \no{
    \RepHolChi{-k_{n+1}} F \,
    }
    \;=\; 
    \RepHolChi{-k_{n+1}} \,
    \no{ F \, }
    & \,
    - \;
    \sum_{i=0}^{m-1} (-1)^{n+i} \wick{ \c \RepHolChiNoIndex_{-k_{n+1}} \c \RepHolEtaNoIndex_{-k_{m-i}} } \,
    \no{F^{\hat \RepHolEtaNoIndex}_{m-i}}
    \\
    \nonumber
    & \,
    - \;
    (-1)^{n+m}\,
    \wick{ \c \RepHolChiNoIndex_{-k_{n+1}} \c \RepHolEtaNoIndex_{ 0 } } \,
    \no{F^{\hat \RepHolEtaNoIndex}_{0}}\,,
\end{align}
where we have already removed the terms with vanishing Wick contractions of the type $\wick{\c X \c X}$.
On the other hand, by the definition of the current modes ---Proposition~\ref{prop: current modes}--- and using the discrete Stokes' formula ---Section~\ref{subsubsec: discrete integration}---, we have
\begin{align*}
    \HolCurrModexi{-k_{n+1}} \big( \no{F \,} + \Null\big)
    \;=\; & \,
    \RepHolChi{-k_{n+1}} \,
    \no{ F \, }
    \;+\;
    \frac{1}{2\pi}
    \sum_{\zu_\diamond}\zu_\diamond^{[-k_{n+1}]} \,\ddeebar\pddee\Fieldxi(\zu_\diamond) \no{ F \, }
    + \Null
    \,,
\end{align*}
where the sum runs over a big enough neighbourhood of the origin.
Using the recursive formula in Remark~\ref{rmk: recursive formula norm ord} again, we can massage the last expression into
\begin{align}
\label{eq: other hand}
    \HolCurrModexi{-k_{n+1}} \big( \no{F \,} + \Null\big)
    \;=\; & \,
    \RepHolChi{-k_{n+1}} \,
    \no{ F \, }
    \; + \;
    \frac{1}{2\pi}
    \sum_{\zu_\diamond}\zu_\diamond^{[-k_{n+1}]}
    \mspace{-40mu}
    \overbrace{
    \no{ \ddeebar\pddee\Fieldxi(\zu_\diamond) \phantom{\bigg\vert}\! F \, }}^{\text{null by Remark~\ref{rmk: factorisation laplacian} and Remark~\ref{rmk: norm ord factors through null fields}}}
    \\
    \nonumber
    &
    \mspace{30mu}
    \;+\;
    \frac{1}{2\pi}
    \sum_{\zu_\diamond}\zu_\diamond^{[-k_{n+1}]} \,
    \Bigg(
    \sum_{i=0}^{m-1} (-1)^{m+i}\,
    \wick{ \ddeebar\pddee\c \Fieldxi(\zu_\diamond) \c \RepHolEtaNoIndex_{-k_{m-i}} } \,
    \no{F^{\hat \RepHolEtaNoIndex}_{m-i}}
    \\
    \nonumber
    &
    \mspace{220mu}
    + \;
    (-1)^{n+m}\,
    \wick{ \ddeebar\pddee\c \Fieldxi(\zu_\diamond) \c \RepHolEtaNoIndex_{ 0 } } \,
    \no{F^{\hat \RepHolEtaNoIndex}_{0}}
    \Bigg)\,.
\end{align}
Comparing~\eqref{eq: one hand} and \eqref{eq: other hand} we see that
if we have
\begin{align}
\label{eq: key ingredient}
    \frac{1}{2\pi}
    \sum_{\zu_\diamond}\zu_\diamond^{[-k]} \,
    \wick{ \ddeebar\pddee\c \Fieldxi(\zu_\diamond) \c \RepHolEtaNoIndex_{-\ell} }
    \,+\,
    \wick{ \c \RepHolChiNoIndex_{-k} \c \RepHolEtaNoIndex_{-\ell } }
    \,=\, 0
\end{align}
for all $k>0$ and $\ell\geq 0$, we get the desired
\begin{align*}
    \HolCurrModexi{-k_{n+1}} \big( \no{F \,} + \Null\big)
    \, = \,
    \no{
    \RepHolChi{-k_{n+1}} F }
    \,+\, \Null\,.
\end{align*}
Using the expression of the linear representatives ---Example~\ref{ex: rep of linear basis}--- and using Stokes' formula again, \eqref{eq: key ingredient} is equivalent to
\begin{align*}
    0 \; = \; & \,
    \frac{1}{2\pi}
    \sum_{\zu_\diamond}\zu_\diamond^{[-k]} \,
    \wick{ \ddeebar\pddee\c \Fieldxi(\zu_\diamond) \c \RepHolEtaNoIndex_{-\ell} }
    \,+\,
    \frac{1}{2\pi}
    \sum_{\zu_\medial}
    \ddeebar \zu_\medial^{[-k]}
    \wick{ \pddee \c \Fieldxi (\zu_\medial) \c \RepHolEtaNoIndex_{-\ell } }
    \\
    = \; & \,
    \frac{1}{2\pi\ii}
    \sqint_\gamma
    \zu_\diamond^{[-k]}
    \wick{ \pddee \c \Fieldxi (\zu_\medial) \c \RepHolEtaNoIndex_{-\ell } }\,
    \dd \zu\,,
\end{align*}
which is indeed true by virtue of Lemma~\ref{lemma: technical}.
\end{proof}

\vspace{-10pt}
\subsection{Proof of Theorem~\ref{thm: isomorphism}}
\label{subsec: proof of isomorphism}
At this point we are ready to prove the isomorphism $\Fields \cong \FullFock$.
Proving that we have $\FullFock\subset\Fields$ as a subrepresentation is straightforward to prove with an algebra-theoretic argument presented in the following lemma.
The opposite inclusion requires more work and is the essence of the proof of Theorem~\ref{thm: isomorphism} below.

In the lemma, we use the following common notation:
given a vector $v$ in a $\FullSymFer$ representation $V$, we let $(\FullEnvSymFer)v$ denote the subrepresentation of $V$ cyclically generated by the action of $\FullSymFer$ on $v$.

\begin{lemma}\label{lemma: univ prop Fock}
Let $V$ be a representation\footnote{We use the module notation for simplicity.} of $\FullSymFer$, and let $v\in V\setminus \{0\}$  be a non-zero vector satisfying
\begin{align*}
    \mspace{130mu}
    \AlgChi{k}v
    \,=\,
    \AlgEta{k}v
    \,=\,
    \AlgChiBar{k}v
    \,=\,
    \AlgEtaBar{k}v
    \,=\,
    0 & \phantom{\Big\vert}
    \mspace{50mu}
    \text{ for } \mspace{5mu} k>0\,,
    \\
    \big( \AlgEta{0} - \AlgEtaBar{0} \big) v \,=\, 0\,,
    \mspace{50mu}
    \big( \AlgChi{0} - \AlgChiBar{0} \big) v \, &\,=\, 0\,,\phantom{\Big\vert}
\end{align*}
and
\begin{align*}
    \AlgChi{0} \AlgEta{0} v \,\neq\, 0 \,.
\end{align*}
Then there exists a unique $(\FullSymFer)$-homomorphism $\FullFock\rightarrow V$ determined by 
\begin{align*}
    \GroundPartner \longmapsto v\,.
\end{align*}
This map is injective and constitutes an isomorphism between $\FullFock$ and $(\FullEnvSymFer)v$ as representations of $\FullSymFer$.
\end{lemma}

\begin{proof}
The existence of such a map follows from the universal property of the quotient~$\FullFock$.
Uniqueness follows from the fact that the values of the map on the basis vectors are fixed by $\GroundPartner \mapsto v$ and the homomorphism property.
Let $\phi\colon\FullFock\rightarrow V$ denote such a unique homomorphism.
By construction, the map $\phi$ maps $\FullFock$ surjectively onto the submodule $(\FullEnvSymFer)v$.
As for injectivity, since it is a non-zero map, its kernel must be a proper submodule of $\FullFock$.
Recall from Remark~\ref{rmk: subrepresentations} that the only three proper submodules of $\FullFock$
are $(\FullEnvSymFer)\GroundChi$, $(\FullEnvSymFer)\GroundEta$ and $(\FullEnvSymFer)\Ground$; the third one, in turn, is a proper submodule of each of the first two.
However, since we assumed
\begin{align*}
    \phi ( \Ground )
    \,=\,
    \phi ( \AlgChi{0} \AlgEta{0} \GroundPartner )
    \,=\,
    \AlgChi{0} \AlgEta{0} \phi ( \GroundPartner )
    \,=\,
    \AlgChi{0} \AlgEta{0} v
    \,\neq\,
    0\,,
\end{align*}
we must have $\ker \phi = \{0\}$, i.e.~$\phi$ is injective.
We conclude that $\phi$ is an isomorphism between $\FullFock$ and $(\FullEnvSymFer)v$.
\end{proof}

\begin{proof}[Proof of Theorem~\ref{thm: isomorphism}]
The local field~$\fDGFFGroundPartner$ satisfies
\begin{align*}
    \HolCurrModexi{k}\fDGFFGroundPartner
    \,=\,
    \HolCurrModetheta{k}\fDGFFGroundPartner
    \,=\,
    \AntiHolCurrModexi{k}\fDGFFGroundPartner
    \,=\,
    \AntiHolCurrModetheta{k}\fDGFFGroundPartner
    \,=\,
    0
\end{align*}
for all $k\in\Zpos$ and
\begin{align*}
    \big(\HolCurrModexi{0}-\AntiHolCurrModexi{0}\big)
    \fDGFFGroundPartner
    \,=\,
    0
    \,=\,
    \big(\HolCurrModetheta{0}-\AntiHolCurrModetheta{0}\big)
    \fDGFFGroundPartner\,.
\end{align*}
Moreover, we have $\HolCurrModexi{0}\HolCurrModetheta{0}\fDGFFGroundPartner=\fDGFFGround\neq0$ ---Example~\ref{ex: 1 is not 0}---.
Hence, by Lemma~\ref{lemma: univ prop Fock}, the subrepresentation $(\LattSFAsso)\fDGFFGroundPartner$ generated by the cyclic action of the current modes on $\fDGFFGroundPartner$ is isomorphic to $\FullFock$.
Since we have the inclusion $(\LattSFAsso)\fDGFFGroundPartner\subset \Fields$, it only remains to prove $\Fields \subset (\LattSFAsso)\fDGFFGroundPartner$.
Take any field $F\in\Fields$.
Since normal ordering is a surjective map ---Remark~\ref{rmk: norm ord surjective}---, we can find linear field polynomials $L_1,\ldots,L_n\in\LinearFieldPoly$ such that we have
\begin{align*}
    F
    \,=\,
    \noQuo{
    \big( L_1 + \Null\big)
    \cdots
    \big( L_n + \Null\big)
    }\,.
\end{align*}
Since $L_i$ are linear, they must be a linear combination of the elements in the basis
\begin{align*}
	\big\{
	\HolCurrModexi{-k} \fDGFFGround,
	\AntiHolCurrModexi{-k} \fDGFFGround,	\HolCurrModetheta{-k} \fDGFFGround,
	\AntiHolCurrModetheta{-k} \fDGFFGround,
	\,\colon
	k\in\Zpos
	\big\}
	\,\cup\,
	\big\{
	\fDGFFGroundTheta,
	\fDGFFGroundXi
	\big\}\,.
\end{align*}
Finally, by Lemma~\ref{lemma: normal order of basis}, the multilinearlity of normal ordering, and the relations between ground states via the current modes ---Remark~\ref{rmk: action currents on ground states}---, the local field $F$ can be written as an element of the submodule $(\LattSFAsso)\fDGFFGroundPartner$.
\end{proof}

\vspace{-10pt}
\subsection{Proof of Theorem~\ref{thm: scaling limit}}
\label{subsec: proof of scaling limit}
We now move on to the scaling limit of correlation functions.
In the proof of our second main theorem, we use the following self-contained results.
Consider the scaling-limit setting as in Theorem~\ref{thm: scaling limit}.
For the rest of the section we define $\mathsf C \coloneqq \gamma + \frac{3}{2} \log 2$ where $\gamma$ is the Euler--Mascheroni constant.

\begin{lemma}
\label{lemma: converg diagonal greens fnctn}
Let $\zd^\meshsize\in\ddomain{\meshsize}$ be the closest vertex to $z\in\domain$ in $\ddomain{\meshsize}$.
We have
\begin{align*}
    \Green_{\ddomain{\meshsize}}(\zd^\meshsize,\zd^\meshsize)
    +
    \frac{1}{2\pi}\log \meshsize
    \mspace{20mu}
    \xrightarrow{\mspace{40mu}\meshsize \downarrow 0 \mspace{40mu}}
    \mspace{20mu}
    \HarmOfGreen_{\domain}(z,z)
    +
    \frac{\mathsf{C}}{2\pi}
    \,,
\end{align*}
uniformly for $z$ in compacts of $\domain$, where $\HarmOfGreen_{\domain}$ is the harmonic part of the Green's function on the domain~$\domain$ as in Section~\ref{subsec: corr fun symplectic fermions}.
\end{lemma}

\begin{proof}
The function,
\begin{align*}
    H^\meshsize
    \,\colon\,
    \zdbis
    \longmapsto
    \Green_{\ddomain{\meshsize}}\big(\zdbis,\zd^\meshsize \big)
    -
    \Green_{\Z^2}\bigg( \frac{\zdbis-\zd^\meshsize}{\meshsize} \bigg)
    +
    \frac{1}{2\pi}\log \meshsize
    -
    \frac{\mathsf{C}}{2\pi}
\end{align*}
is discrete harmonic on $\ddomain{\meshsize}$ and has boundary values $\frac{1}{2\pi}\log \vert \xd^\meshsize - \zd^\meshsize \vert + \mathcal{O}(\meshsize^2)$ as $\meshsize\downarrow0$ for $\xd^\meshsize$ at the boundary of $\ddomain{\meshsize}$.
It is well-known that (any reasonable extension to $\domain$ of) $H^\meshsize$ converges to the harmonic extension of the function $x\mapsto\frac{1}{2\pi}\log \vert x - z\vert$ on $\bdry\domain$ to the interior $\domain$.
This limiting function on $\domain$ is the harmonic part of the Green's function $\HarmOfGreen_\domain( \,\cdot\, ,z)$.
Evaluating $H^\meshsize$ at $\zd^\meshsize$, we get the desired scaling limit convergence since we have $\Green_{\Z^2}(0)=0$.
\end{proof}

\begin{coro}
\label{coro: converg of partner}
Let $\zd^\meshsize\in\ddomain{\meshsize}$ be the closest vertex to $z\in\domain$ in $\ddomain{\meshsize}$ and fix $\lambda > 0$.
We have
\begin{align*}
    -\BigfDGFFCorrFun{\ddomain{\meshsize}}{\fDGFFxi(\zd^\meshsize)\fDGFFtheta(\zd^\meshsize)}
    - 2\, \log \big( \lambda \meshsize \big)
    \mspace{20mu}
    \xrightarrow{\mspace{40mu}\meshsize \downarrow 0 \mspace{40mu}}
    \mspace{20mu}
    \BigSFCorrFun{\domain}{\thectt(\lambda)}{\GroundPartner(z)}
\end{align*}
uniformly for $z$ in compacts of $\Conf{n}{\domain}$, where $\thectt(\lambda)\coloneqq 2 \big( \log \lambda + \mathsf{C} \big)$.
\end{coro}
\begin{proof}
Straightforward from Lemma~\ref{lemma: converg diagonal greens fnctn} and Example~\ref{ex: 1pt fctn log partner}.
\end{proof}

We are now ready to prove the main theorem.

\begin{proof}[Proof of Theorem~\ref{thm: scaling limit}]
By linearity of the correlation functions, we only need to prove the statement for fields $\field_i$ in the basis~\eqref{eq: basis}.
For the sake of brevity, we detail the proof in the case where we have
\begin{align*}
\field_i
\,\coloneqq\,
\HolCurrModexi{-k_{r_i}^{(i)}}\cdots
&
\HolCurrModexi{-k_1^{(i)}}
\HolCurrModetheta{-\ell_{s_i}^{(i)}}\cdots\HolCurrModetheta{-\ell_1^{(i)}}
\fDGFFGroundPartner
\in \Fields\,
\end{align*}
with
${0 \leq k_1^{(i)} < k_2^{(i)} < \cdots < k_{r_i}^{(i)}}$
and
${0 \leq \ell_1^{(i)} < \ell_2^{(i)} < \cdots < \ell_{s_i}^{(i)}}$---the more general case can be \linebreak[4] handled in a similar way.

\noindent
The CFT correlation functions of the fields $\field_1,\ldots,\field_n\in\FullFock$ ($\cong\Fields$) can be expressed ---recall Remark~\ref{rmk: integral formula}--- as
\begin{align} \label{eq: CFT corr fun}
    &
    \BigSFCorrFun{\domain}{\thectt}{
    \field_1(z_1)
    \,\cdots\,
    \field_n(z_n)
    }
    \,=\phantom{\Bigg\vert}
    \\ \nonumber
    &
    \mspace{60mu}
    =\,
    \oint \!\cdots\! \oint
    \prod_{i=1}^n
    \Bigg(
    \prod_{a=1}^{r_i}
    \frac{\cd \zeta_{i;a}^{\eta}}{2\pi\ii}
    \prod_{b=1}^{s_i}
    \frac{\cd \zeta_{i;b}^{\chi}}{2\pi\ii}
    \Bigg)
    \prod_{i=1}^n
    \Bigg(
    \prod_{a=1}^{r_i}
    \, \big( \zeta_{i;a}^{\eta} - z_i \big)^{-k_a^{(i)}}
    \prod_{b=1}^{s_i}
    \, \big( \zeta_{i;b}^{\chi} - z_i \big)^{-\ell_b^{(i)}}
    \Bigg)
    \times
    \\ \nonumber
    &
    \mspace{150mu}
    \phantom{\Bigg\vert}
    \times
    \bigg\langle
    \HolCurrentEta\big( \zeta_{1;1}^{\eta} \big) \cdots \HolCurrentEta\big( \zeta_{1;r_1}^{\eta} \big)
    \HolCurrentChi\big( \zeta_{1;1}^{\chi} \big) \cdots \HolCurrentChi\big( \zeta_{1;s_1}^{\chi} \big)\,
    \GroundPartner(z_1)
    \;\cdots
    \\ \nonumber
    &
    \mspace{290mu}
    \cdots\;
    \HolCurrentEta\big( \zeta_{n;1}^{\eta} \big) \cdots \HolCurrentEta\big( \zeta_{n;r_n}^{\eta} \big)
    \HolCurrentChi\big( \zeta_{n;1}^{\chi} \big) \cdots \HolCurrentChi\big( \zeta_{n;s_n}^{\chi} \big)\,
    \GroundPartner(z_n)
    \bigg\rangle_{\domain;\thectt}^{\SFtag}\,,
\end{align}
where the variables $\zeta_{i;a}^{\eta}$ and $\zeta_{i;b}^{\chi}$ are integrated, respectively, along contours $\Gamma_{i;a}^{\eta}$ and $\Gamma_{i;b}^{\chi}$ in $\domain$ that encircle the point $z_i$ and no other insertion point, and such that all contours are pairwise non-intersecting.
An example of this scenario is shown in Figure~\ref{fig: proof}.

\noindent
In what remains of this proof, we omit some details---we refer the reader to the proof of Theorem~7.2 in \cite{ABK-DGFF_local_fields}.
In particular, in there, the contours and their discrete approxi\-mations are deliberately chosen to be of square shape so that some technical aspects substantially simplify.

\noindent
To set up the scaling-limit convergence,
we let $\gamma_{i;a}^{\eta}(\meshsize)$ and $\gamma_{i;b}^{\chi}(\meshsize)$ be discrete corner contours\footnote{Naturally, one needs to impose certain requirements to these discrete contours so that in the limit $\meshsize\downarrow0$ one gets proper Riemann-sum approximations of continuum contour integrals.} in $\Z^2$ that are, respectively, close to the blown-up contours $\frac{1}{\meshsize}\big(\Gamma_{i;a}^{\eta}-z_i\big)$ and $\frac{1}{\meshsize}\big(\Gamma_{i;b}^{\chi}-z_i\big)$.
Then, for the field~$\field_i\in\Fields$, consider the $\FieldPoly$-representative
\begin{align*}
    &
    P_i^\meshsize \,\coloneqq\,
    \prod_{a=1}^{r_i}
    \Bigg(
    \dsqint_{\gamma^\eta_{i;a}(\meshsize)}
    \frac{\dd \zu_{i;a}}{2\pi \ii}
    (\zu_{i;a})_\diamond^{[-k_a^{(i)}]}
    \pddee \Fieldtheta \big( (\zu_{i;a})_\medial \big)
    \Bigg)
    \times
    \\
    &
    \mspace{150mu}
    \times\,
    \prod_{b=1}^{s_i}
    \Bigg(
    \dsqint_{\gamma^\chi_{i;b}(\meshsize)} \,
    \frac{\dd \zubis_{i;b}}{2\pi \ii}
    (\zubis_{i;b})_\diamond^{[-\ell_b^{(i)}]}
    \pddee \Fieldxi \big( (\zubis_{i;b})_\medial \big)
    \Bigg)
    \Big( -\Fieldxi(0)\Fieldtheta(0) - 2\log\big(\lambda\meshsize\big)\Big)
\end{align*}
of the field~$\field_i-\log(\lambda\meshsize)\big[(\dL{0}-\dLBar{0})-(\Delta_i+\bar\Delta_i)\big]\field_i$,
where we have used the basis expression in Example~\ref{ex: good basis scaling fields}, the representative~$\fDGFFGroundPartner = -\Fieldxi(0)\Fieldtheta(0)+\Null$ ---Example~\ref{ex: rep of ground states}---, and the integral definition of the current modes ---Proposition~\ref{prop: current modes}---.
\begin{figure}[b!]
\centering
\begin{overpic}[scale=0.48, tics=10]{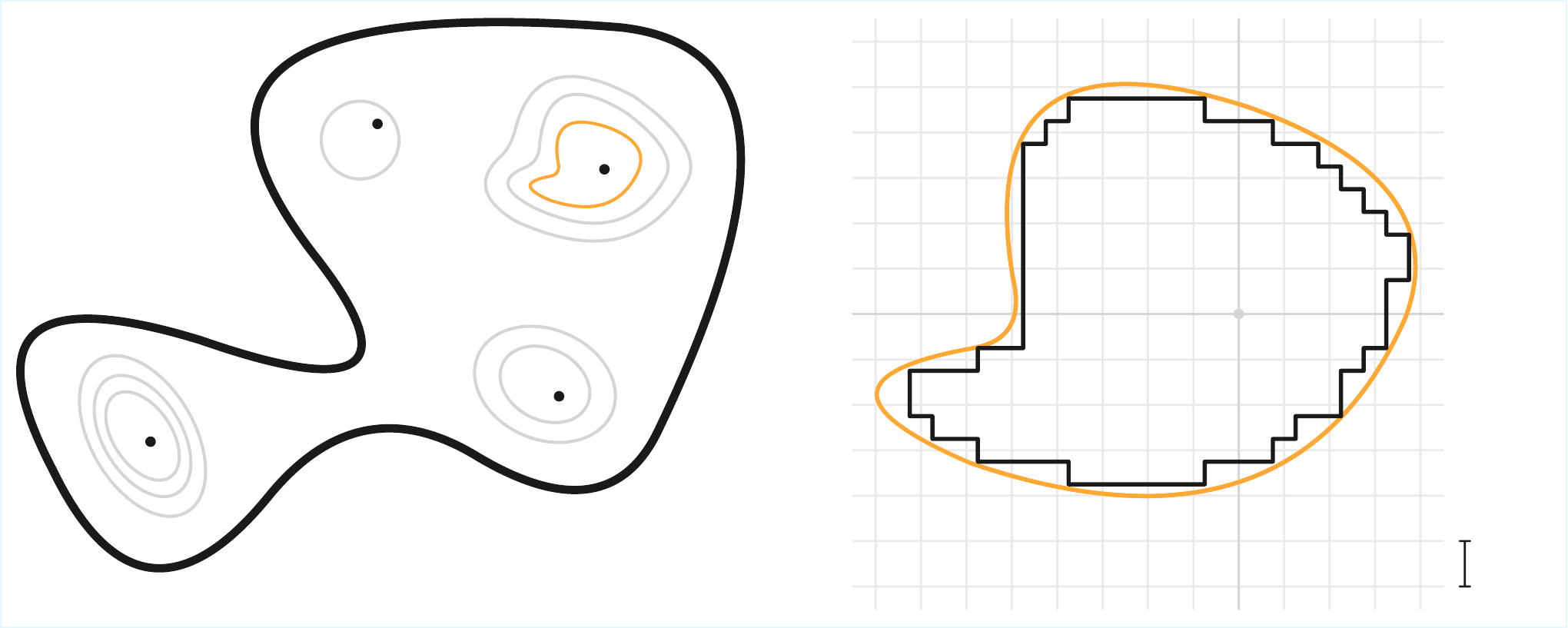}
    \put(95,3.3){$\meshsize$}
    \put(93,36){\large $\meshsize\Z^2$}
    \put(13,36){\large $\domain$}
    \put(36.5,30){$z_i$}
    \put(34,24.3){
		\pgfsetfillopacity{0.8}
		\colorbox{white}{\centering
		\parbox{0pt}{\pgfsetfillopacity{1}}\color{black} \hspace{0pt}$\Gamma_{i;1}$}}
    \put(63,5){
		\pgfsetfillopacity{0.8}
		\colorbox{white}{\centering
		\parbox{0pt}{\pgfsetfillopacity{1}}\color{black} \hspace{0pt}$\frac{1}{\meshsize}\big(\Gamma_{i;1}-z_i\big)$}}
        \put(63,12.3){
		\pgfsetfillopacity{0.8}
		\colorbox{white}{\centering
		\parbox{0pt}{\pgfsetfillopacity{1}}\color{black} \hspace{0pt}$\gamma_{i;1}(\meshsize)$}}
\end{overpic}
\centering
\caption{
On the left, a domain~$\domain$ with insertion points and contours around \\ them.
One of the contours ($\Gamma_{i;1}$) around the point $z_i$ is highlighted in orange.
\\
On the right, the shifted and blown-up contour~$\frac{1}{\meshsize}\big(\Gamma_{i;1}-z_i\big)$ (orange) and \\ its closest discrete contour~$\gamma_{i;1}(\meshsize)$.}
\label{fig: proof}
\end{figure}
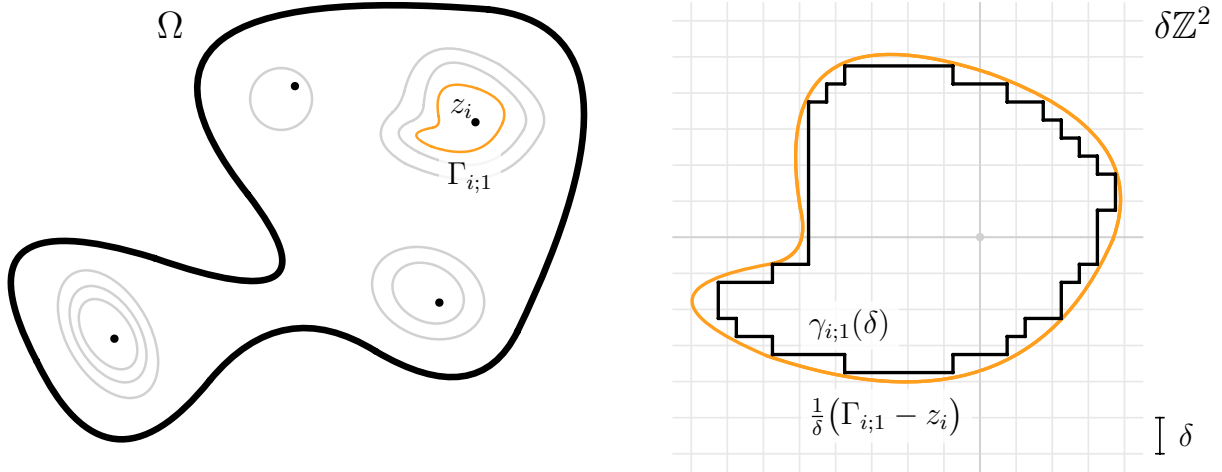
The relevant fDGFF correlation functions are then
\begin{align}
\label{eq: fDGFF corr fun}
    &
    \BigfDGFFCorrFun{\ddomain{\meshsize}}{
    \ev_{\zd_1^\meshsize}^{\ddomain{\meshsize}}\big( P_1 \big)
    \cdots
    \ev_{\zd_n^\meshsize}^{\ddomain{\meshsize}}\big( P_n \big)
    }
    \, = \phantom{\Bigg\vert}
    \\ \nonumber
    &
    \mspace{20mu}
    =\,
    \prod_{i=1}^n
    \Bigg(
    \prod_{a=1}^{r_i}
    \dsqint_{\gamma^\eta_{i;a}(\meshsize)}
    \frac{\dd \zu_{i;a}}{2\pi\ii} \,
    (\zu_{i;a})_\diamond^{[-k_a^{(i)}]}
    \prod_{b=1}^{s_i}
    \;
    \dsqint_{\gamma^\chi_{i;b}(\meshsize)} \,
    \frac{\dd \zubis_{i;b}}{2\pi\ii}
    (\zubis_{i;b})_\diamond^{[-\ell_b^{(i)}]}
    \Bigg)
    \times
    \\ \nonumber
    &
    \mspace{30mu}
    \phantom{\Bigg\vert}
    \times
    \bigg\langle
    \pddee\fDGFFtheta\big( \zd_1^\meshsize + \meshsize (\zu_{1;1})_\medial \big)
    \cdots
    \pddee\fDGFFtheta\big( \zd_1^\meshsize + \meshsize (\zu_{1;r_1})_\medial \big)
    \,\cdot
    \\ \nonumber
    &
    \mspace{90mu}
    \cdot
    \pddee\fDGFFxi\big( \zd_1^\meshsize + \meshsize (\zubis_{1;1})_\medial \big)
    \cdots
    \pddee\fDGFFxi\big( \zd_1^\meshsize + \meshsize (\zubis_{1;s_1})_\medial\big)
    \Big(
    - \fDGFFxi\big( \zd_1^\meshsize \big) \fDGFFtheta\big( \zd_1^\meshsize \big) 
    - 2\log\big(\lambda\meshsize\big)
    \Big)
    \cdots
    \\ \nonumber
    &
    \mspace{90mu}
    \cdots\;
    \pddee\fDGFFtheta\big( \zd_n^\meshsize + \meshsize (\zu_{n;1})_\medial \big)
    \cdots
    \pddee\fDGFFtheta\big( \zd_n^\meshsize + \meshsize (\zu_{n;r_n})_\medial \big)
    \cdot\phantom{\Bigg\vert}
    \\ \nonumber
    &
    \mspace{90mu}
    \cdot
    \pddee\fDGFFxi\big( \zd_n^\meshsize + \meshsize (\zubis_{n;1})_\medial \big)
    \cdots
    \pddee\fDGFFxi\big( \zd_n^\meshsize + \meshsize (\zubis_{n;s_n})_\medial\big)
    \Big(
    - \fDGFFxi\big( \zd_n^\meshsize \big) \fDGFFtheta\big( \zd_n^\meshsize \big) 
    - 2\log\big(\lambda\meshsize\big)
    \Big)
    \bigg\rangle_{\ddomain{\meshsize}}^{\fDGFFtag}
    \!\!\!.
\end{align}
Let us focus on the integrand that is an fDGFF correlation function.
By the fDGFF generalised Wick's formula ---Remark~\ref{rmk: discrete Wick with log corrections}---, and by linearity of the fDGFF correlation functions, it can be rewritten as sums of products of factors of the two following types:
\begin{align*}
    \dsqint_{\gamma^\eta_{i;a}(\meshsize)}
    \!\!\!\!
    \frac{\dd \zu_{i;a}}{2\pi\ii} \,
    (\zu_{i;a})_\diamond^{[-k_a^{(i)}]}
    \dsqint_{\gamma^\chi_{j;b}(\meshsize)} \,
    \!\!\!\!
    \frac{\dd \zubis_{j;b}}{2\pi\ii}
    (\zubis_{i;b})_\diamond^{[-\ell_b^{(j)}]}
    \big(\pddee\big)_{\!\zu}
    \big(\pddee\big)_{\!\zubis}
    \BigfDGFFCorrFun{\ddomain{\meshsize}}{
    \fDGFFtheta ( \zd_i^\meshsize+\meshsize(\zu_{i;a})_\medial )
    \fDGFFxi ( \zd_j^\meshsize+\meshsize(\zubis_{j;b})_\medial )
    }
\end{align*}
and
\begin{align*}
    - \BigfDGFFCorrFun{\ddomain{\meshsize}}{
    \fDGFFxi \big( \zd_i^\meshsize \big)
    \fDGFFtheta \big( \zd_i^\meshsize \big)
    }
    - 2\log\big(\lambda\meshsize\big)\,.
\end{align*}
The latter type of factor, converges to 
\begin{align*}
    \BigSFCorrFun{\domain}{\alpha(\lambda)}{\GroundPartner(x)}
\end{align*}
uniformly for $x$ on compacts of $\domain$ by Corollary~\ref{coro: converg of partner}.
The former type, when renormalised by $\meshsize^{-(k_a^{(i)}+\ell_b^{(j)})}$, converges uniformly for $(z_i,z_j)$ on compacts of $\Conf{2}{\domain}$ to
\begin{align*}
    \oint_{\Gamma^\eta_{i;a}}
    \!\!
    \frac{\cd \zeta^{\eta}_{i;a}}{2\pi\ii} \,
    \big( \zeta^{\eta}_{i;a} \big)^{-k_a^{(i)}}
    \dsqint_{\Gamma^\chi_{j;b}} \,
    \!\!
    \frac{\cd \zeta^{\chi}_{j;b}}{2\pi\ii} \,
    \big( \zeta^{\chi}_{j;b} \big)^{-\ell_b^{(j)}}
    \partial_{\zeta^{\eta}_{i;a}}
    \partial_{\zeta^{\chi}_{j;b}}
    \BigSFCorrFun{\domain}{\thectt(\lambda)}{
    \GroundEta \big( z_i+\zeta^{\eta}_{i;a} \big)\,
    \GroundChi \big( z_j+\zeta^{\chi}_{j;b} \big)
    }\,,
\end{align*}
where we took $\alpha = \alpha(\lambda)$ since the CFT correlation does not depend on this constant.
To see this, recall that for $\zu^\meshsize\in\meshsize\Zdiamond$ satisfying $\meshsize\zu^\meshsize\to\zeta \in \C$ as $\meshsize\downarrow0$, we have ---Proposition~\ref{prop: monomials} (Property~8)--- the convergence of the discrete monomials
\begin{align*}
    \meshsize^k \big(\zu^\meshsize\big)^{[k]}\;
    \xlongrightarrow{\;\;\; \meshsize\downarrow 0\;\;\;}
    \;\zeta^k
\end{align*}
uniformly for $\zeta$ on compacts of $\C\setminus\{0\}$.
We also need the classical result of the scaling-limit convergence of the double derivatives of the discrete (Dirichlet) Green's function to its continuum counterpart when renormalized by $\meshsize^2$---see \cite[Section~3.3]{ABK-DGFF_local_fields}.
Then, when each of the discrete differentials $\dd$ is renormalised by $\meshsize$, the expression $\eqref{eq: fDGFF corr fun}$ becomes a Riemann-sum approximation of double integrals.
Collecting the powers of $\meshsize$ we find that, if we renormalise~\eqref{eq: fDGFF corr fun} by
\begin{align*}
    {{\mathlarger{\meshsize}}^{\,-\sum_i\big( \sum_a k_{i;a} + \sum_b \ell_{i;b} \big)}},
\end{align*}
it converges uniformly on compacts to \eqref{eq: CFT corr fun} ---after a shift in the integration variables---.
Note that $\sum_a k_{i;a} + \sum_b \ell_{i;b}$ is the generalised scaling dimension of the field $\field_i$.
\end{proof}

\vspace{-10pt}
\subsection{Proof of Lemma~\ref{lemma: corrfun h-one and diss}}
\label{subsec: proof stat mech}
We now present the proof of the connection between the corelation functions of the fDGFF and those of the height-one field in presence of dissipations in the Abelian sandpile model.

\begin{proof}[Proof of Lemma~\ref{lemma: corrfun h-one and diss}]
Via the burning algorithm, recurrent configuration with dissipative sites at $\mathbf{W}$ are in bijection with spanning forests of $\ddomain{\meshsize}$ where $\bdry\ddomain{\meshsize}$ and each of the dissipative sites $\zdbis_i$ belongs in a different tree.
Equivalently, these configurations are in bijection with spanning trees of the graph obtained from $(V_{\ddomain{\meshsize}},E_{\ddomain{\meshsize}})$ by fusing all dissipative sites $\zdbis_i$ with the boundary vertex $\bdry\ddomain{\meshsize}$---let $\tilde{\mathcal{G}}$ denote this modified graph.
The Dirichlet Laplacian matrix $\tilde\Delta_{\ddomain{\meshsize}}^{\texttt{D}}$ of $\tilde{\mathcal{G}}$ is just $\tilde\Delta_{\ddomain{\meshsize}}^{\texttt{D}}$ with the rows and columns of the dissipative sites removed.

\noindent
Let us check that the Green's function of $\tilde{\mathcal{G}}$ is given by
\begin{align*}
    \tilde\Green (\xd,\xdbis)
    \,=\,
    \frac{1}{4\pi}
    \frac{
    \bigfDGFFCorrFun{\ddomain{\meshsize}}{
    \fDGFFxi(\xd)\fDGFFtheta(\xdbis)
    D(\zdbis_1)\cdots D(\zdbis_m)
    }}
    {\bigfDGFFCorrFun{\ddomain{\meshsize}}{
    D(\zdbis_1)\cdots D(\zdbis_m)
    }}\,,
\end{align*}
where, recall, that $D(\zdbis_i) = \tfrac{1}{4\pi}\fDGFFxi(\zdbis_i)\fDGFFtheta(\zdbis_i)$.
From Remark~\ref{rmk: 2n pt function fDGFF}, we see that it is symmetric, we have $\bigfDGFFCorrFun{\ddomain{\meshsize}}{
    \fDGFFxi(\xd)\fDGFFtheta(\xdbis)
    D(\zdbis_1)\cdots D(\zdbis_m)}
    \,=\, 0$
for $x=\zdbis_i$, and
\begin{align*}
    \bigfDGFFCorrFun{\ddomain{\meshsize}}{
    \fDGFFxi(\xd)\fDGFFtheta(\xdbis)
    D(\zdbis_1)\cdots D(\zdbis_m)}
    \,=\, & \,
    \bigfDGFFCorrFun{\ddomain{\meshsize}}{
    \fDGFFxi(\xd)\fDGFFtheta(\xdbis)}\,
    \bigfDGFFCorrFun{\ddomain{\meshsize}}{
    D(\zdbis_1)\cdots D(\zdbis_m)}
    \phantom{\bigg\vert}
    \\
    & \mspace{70mu} \phantom{\Bigg\vert}
    + \;
    \parbox{150pt}{\centering discrete harmonic terms \linebreak w.r.t.~$\xd$ at $\ddomain{\meshsize}\setminus \mathbf{W}$\,,}
\end{align*}
which straightforwardly leads to $\big[\tilde\Delta_{\ddomain{\meshsize}}^{\texttt{D}}\tilde\Green\big](\xd,\xdbis)=-\delta_{\xd,\xdbis}$ for $\xd,\xdbis\in\ddomain{\meshsize}\setminus\mathbf{W}$.

\noindent
At this point, similarly as in the case of the height-one field without the presence of dissipation fields, one can prove that we have
\begin{align*}
    \E^{\textnormal{ASM}}_{\ddomain{\meshsize}\setminus\mathbf{W}}
    \big[\,
    \heightOne(\zdbis_1)
    \cdots
    \heightOne(\zdbis_m)
    \big]
    \,=\,
    \BigfDGFFCorrFun{\ddomain{\meshsize}}{
    H_1(\zd_1) \cdots H_1(\zd_m)
    D(\zdbis_1) \cdots D(\zdbis_n)
    }\,,
\end{align*}
where, recall, we defined $H_1(\zd) \coloneqq \frac{1}{4\pi}\big[\prod_{e \sim \zd}\big(\dgrad\fDGFFxi(e)\dgrad\fDGFFtheta(e) - 1\big) - 1\big]$.

\noindent
To complete the proof, we just need to prove
\begin{align*}
    \BigfDGFFCorrFun{\ddomain{\meshsize}}{
    D(\zdbis_1) \cdots D(\zdbis_n)
    }
    \,=\,
    \frac
    {\big\vert
    \Recurrent{\ddomain{\meshsize}}
    \big\vert}
    {\big\vert
    \Recurrent{\ddomain{\meshsize}\setminus\mathbf{W}}
    \big\vert}
    \,.
\end{align*}
Recall that recurrent configurations are in bijection with spanning trees via the burning algorithm.
Then, on the one hand, using the matrix-tree theorem, we get
\begin{align}
\label{eq: statmech1}
    \frac{\big\vert 
    \Recurrent{\ddomain{\meshsize}\setminus\mathbf{W}}
    \big\vert}
    {\big\vert \Recurrent{\ddomain{\meshsize}} \big\vert}
    \,=\, & \,
    \frac{\vert \det \tilde\Delta_{\ddomain{\meshsize}}^{\texttt{D}} \,\vert}
    {\vert \det \Delta_{\ddomain{\meshsize}}^{\texttt{D}} \,\vert}
    \\ \nonumber
    \,=\, & \,
    \frac{\big\vert \det \big( \, \Delta_{\ddomain{\meshsize}}^{\texttt{D}} - \sum\nolimits_i \One^{\zdbis_i} \big)\big\vert}
    {\vert \det \Delta_{\ddomain{\meshsize}}^{\texttt{D}} \,\vert}
    \phantom{\Bigg\vert}
\end{align}
\begin{align*}    
    \,=\, & \,
    \frac{
    \vert \det \Delta_{\ddomain{\meshsize}}^{\texttt{D}} \,\vert
    \cdot
    \big\vert
    \det \big( \, \identity + \sum\nolimits_i \Green_{\ddomain{\meshsize}} \One^{\zdbis_i} \big)
    \big\vert }
    {\vert \det \Delta_{\ddomain{\meshsize}}^{\texttt{D}} \,\vert}
    \\ \nonumber
    \,=\, & \,
    \Big\vert
    \det \Big( \, \identity + \sum\nolimits_i \Green_{\ddomain{\meshsize}} \One^{\zdbis_i} \Big)
    \Big\vert
    \\ \nonumber
    \,=\, & \,
    \big\vert
    \det\nolimits_{\mathbf{W}} ( \, \identity + \Green_{\ddomain{\meshsize}} )
    \,\big\vert\phantom{\Bigg\vert}
    \\ \nonumber
    \,=\, & \,
    \sum_{\sigma\in\mathfrak{S}_n}
    (-1)^\sigma
    \prod_{i=\sigma(i)}
    \Big(
    1 + \Green_{\ddomain{\meshsize}} ( \zdbis_i , \zdbis_i )
    \Big)
    \prod_{j\neq\sigma(j)}
    \Green_{\ddomain{\meshsize}} ( \zdbis_j , \zdbis_{\sigma(j)} )
\end{align*}
where $\det_{\mathbf{W}}$ denotes the determinant of the submatrix with rows and columns indexed by $\mathbf{W}=\{\zdbis_1,\ldots,\zdbis_n\}$\,.
On the other hand, from the fermionic Wick's formula, it is straightforward to get get
\begin{align*}
    \BigfDGFFCorrFun{\ddomain{\meshsize}}{
    \fDGFFxi(\zdbis_1) \fDGFFtheta(\zdbis_1)
    \cdots
    \fDGFFxi(\zdbis_n) \fDGFFtheta(\zdbis_n)
    }
    = & \phantom{\bigg\vert}
    \BigfDGFFCorrFun{\ddomain{\meshsize}}{
    \fDGFFxi(\zdbis_1) \fDGFFtheta(\zdbis_1)
    }\!
    \BigfDGFFCorrFun{\ddomain{\meshsize}}{
    \fDGFFxi(\zdbis_2) \fDGFFtheta(\zdbis_2)
    \cdots
    \fDGFFxi(\zdbis_n) \fDGFFtheta(\zdbis_n)
    }
    \\
    & \mspace{40mu} \phantom{\Bigg\vert} + \,
    \sum_{\underset{\sigma(1)\neq 1}{\sigma \in \mathcal{S}_n}}
    (-1)^{\sigma}
    \prod_{i=1}^{n}
    \BigfDGFFCorrFun{\ddomain{\meshsize}}{
    \fDGFFxi(\zdbis_i) \fDGFFtheta(\zdbis_{\sigma(i)}
    }\!\!\!,
\end{align*}
from which one can compute
\begin{align}
\label{eq: statmech2}
    \BigfDGFFCorrFun{\ddomain{\meshsize}}{
    D(\zdbis_1) \cdots D(\zdbis_n)
    }
    \,=\, & \;
    \sum_{\sigma\in\mathfrak{S}_n}
    (-1)^\sigma
    \prod_{i=\sigma(i)}
    \bigg(
    1 +  \frac{1}{4\pi}\BigfDGFFCorrFun{\ddomain{\meshsize}}{
    \fDGFFxi(\zdbis_i) \fDGFFtheta(\zdbis_i)
    }\,
    \bigg)\,
    \times
    \\ \nonumber
    & \mspace{220mu}
    \phantom{\Bigg\vert}
    \times
    \prod_{j\neq\sigma(j)}
    \frac{1}{4\pi}
    \BigfDGFFCorrFun{\ddomain{\meshsize}}{
    \fDGFFxi(\zdbis_j) \fDGFFtheta(\zdbis_{\sigma(j)})
    }
    \!\!\!.
\end{align}
The complete the proof we just need to compare~\eqref{eq: statmech1} and \eqref{eq: statmech2} since we have the $2$-point function $\bigfDGFFCorrFun{\ddomain{\meshsize}}{
    \fDGFFxi(\zd) \fDGFFtheta(\zdbis)
    } = 4\pi \, \Green_{\ddomain{\meshsize}} (\zd,\zdbis)$.
\end{proof}

\vspace{-10pt}
\subsection{Example of a basis computation}
\label{subsec: ex of computation}
In this subsection we present a detailed computation for expressing an fDGFF local field in the basis \eqref{eq: basis}.
We consider such a computation for the open-edge field of the UST---its associated fDGFF local field is
\begin{align*}
    \mathsf{P}_{\OEright}
    \,=\,
    \frac{1}{4\pi}\,
    \dgrad\Fieldxi \big( \tfrac{1}{2} \big)
    \dgrad\Fieldtheta \big( \tfrac{1}{2} \big)
    \,-\, \frac{1}{2}
    \,+\, \Null\,.
\end{align*}
We consider first the linear local fields $\dgrad\Fieldxi \big( \tfrac{1}{2} \big) + \Null$ and $\dgrad\Fieldtheta \big( \tfrac{1}{2} \big) + \Null$ separately.
We have
\begin{align}
\label{eq: dgrad repre}
    \dgrad\Fieldxi \big( \tfrac{1}{2} \big)
    \,+\, \Null
    \, = \, & \;
    \Fieldxi ( 1 ) - \Fieldxi ( 0 ) 
    \,+\, \Null
    \phantom{\Big\vert}
\end{align}
\begin{align*}
    \, = \, & \;
    \Fieldxi ( 1 )
    \,-\, 
    \frac{1}{4}
    \Big[
    \Fieldxi(1) + \Fieldxi(\ii) + \Fieldxi(-1) + \Fieldxi(-\ii)
    \Big]
    \,+\, \Null
    \phantom{\bigg\vert}
    \\
    \, = \, & \;
    \frac{\Fieldxi(1)-\Fieldxi(-1)}{2}
    \,+\,
    \frac{1}{4}
    \Big[
    \Big(\Fieldxi(1) + \Fieldxi(\ii)\Big)
    -
    \Big(\Fieldxi(-1) + \Fieldxi(-\ii)\Big)
    \Big]
    \,+\, \Null\,
\end{align*}
where we used the fact that $\dlaplacian\Fieldxi(0)$ is null.

A priori, a linear local field with only $\Fieldxi$-insertions, can be a linear combination of the basis fields $\fDGFFGroundXi$, $\HolCurrModexi{-k}\fDGFFGround$ and $\AntiHolCurrModexi{-k}\fDGFFGround$ with $k\in\Zpos$.
However, from \cite{ABK-DGFF_local_fields} we have the following insight:
a zero-average linear local fields with $\Fieldxi$-inser\-tions in a ball of radius $1$ can be expressed as a linear combination of the local fields $\HolCurrModexi{-1} \fDGFFGround$, $\AntiHolCurrModexi{-1} \fDGFFGround$ and $\big( \HolCurrModexi{-2} + \AntiHolCurrModexi{-2} \big) \fDGFFGround$.

From Example~\ref{ex: rep of linear basis}, we can find the representatives
\begin{align*}
	\HolCurrModexi{-1} \fDGFFGround
	\,=\, & \;
	\frac{1}{2\pi}
	\sum_{\zu_\medial\in\Z^2_\medial}
	\ddeebar \zu_\medial^{[-1]}
	\pddee\Fieldxi (\zu_\medial)
	\,+\,
	\Null
	\,=\,
	\frac{1}{2}
	\bigg[
	\frac{\Fieldxi (1) - \Fieldxi (-1)}{2}
	-\ii \,
	\frac{\Fieldxi (\ii) - \Fieldxi (-\ii)}{2}
	\bigg]
	\,+\,
	\Null\,,
	\\
	\AntiHolCurrModexi{-1} \fDGFFGround
	\,=\, & \;
	\frac{1}{2\pi}
	\sum_{\zu_\medial\in\Z^2_\medial}
	\ddee \overline{\zu}_\medial^{[-1]}
	\pddeebar\Fieldxi (\zu_\medial)
	\,+\,
	\Null
	\,=\, 
	\frac{1}{2}
	\bigg[
	\frac{\Fieldxi (1) - \Fieldxi (-1)}{2}
	+\ii \,
	\frac{\Fieldxi (\ii) - \Fieldxi (-\ii)}{2}
	\bigg]
	\,+\,
	\Null\,,
\end{align*}
and
\begin{align*}
	\big( \HolCurrModexi{-2}
    +
    \AntiHolCurrModexi{-2} \big) \fDGFFGround
	\,=\, & \;
	\frac{1}{2\pi}
	\sum_{\zu_\medial\in\Z^2_\medial}
	\Big(
    \ddeebar \zu_\medial^{[-2]}
	\pddee\Fieldxi (\zu_\medial)
    +
    \ddee \overline{\zu}_\medial^{[-2]}
	\pddeebar\Fieldxi (\zu_\medial)
    \Big)
	\,+\,
	\Null
	\phantom{\Bigg\vert}
	\\
	\,=\, & \; \;
	\frac{1}{2}\;
	\bigg[
	\Big( \Fieldxi (1) + \Fieldxi (-1) \Big)
	-
	\Big( \Fieldxi (\ii) + \Fieldxi (-\ii) \Big)
	\bigg]
	\,+\,
	\Null\,,
	\phantom{\Bigg\vert}
\end{align*}
where we have used the values of the discrete residues of the first two negative-power discrete monomials from Proposition~\ref{prop: monomials}---these values are shown in Figure~\ref{fig: discrete residues}.

\begin{figure}[t!]
	\centering
	\begin{overpic}[scale=0.48, tics=10]{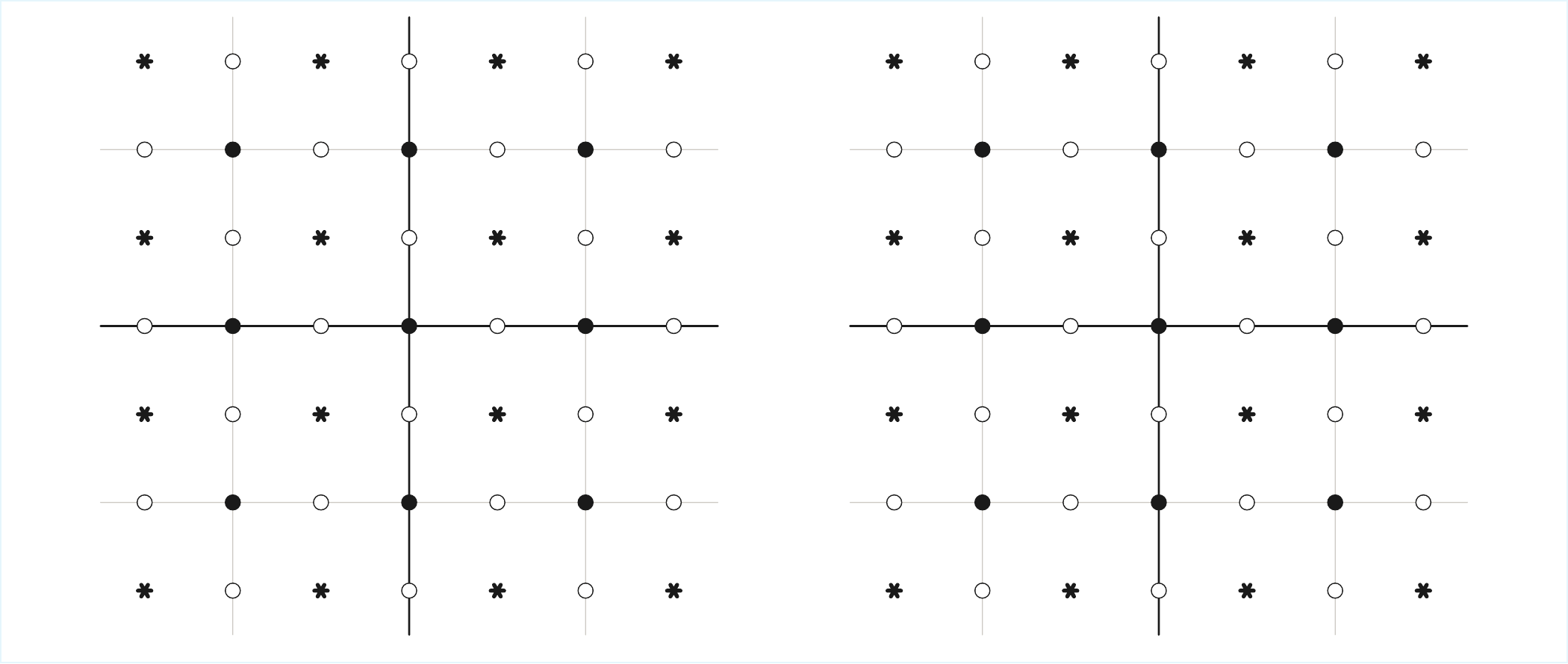}
		\put(7.5,39.15){\niceBlue\scriptsize$0$}
		\put(7.5,33.525){\niceBlue\scriptsize$0$}
		\put(7.5,27.9){\niceBlue\scriptsize$0$}
		\put(7.5,22.275){\niceBlue\scriptsize$0$}
		\put(7.5,16.65){\niceBlue\scriptsize$0$}
		\put(7.5,11.025){\niceBlue\scriptsize$0$}
		\put(7.5,5.4){\niceBlue\scriptsize$0$}
		\put(13.125,39.15){\niceBlue\scriptsize$0$}
		\put(13.125,33.525){\niceBlue\scriptsize$0$}
		\put(13.125,27.9){\niceBlue\scriptsize$0$}
		\put(13.125,22.275){\niceBlue\scriptsize$0$}
		\put(13.125,16.65){\niceBlue\scriptsize$0$}
		\put(13.125,11.025){\niceBlue\scriptsize$0$}
		\put(13.125,5.4){\niceBlue\scriptsize$0$}
		\put(18.75,39.15){\niceBlue\scriptsize$0$}
		\put(18.75,33.525){\niceBlue\scriptsize$0$}
		\put(18.75,27.9){\niceRed\scriptsize$\frac{1}{8}$}
		\put(18.75,22.375){\niceRed\scriptsize$\frac{1}{4}$}
		\put(18.75,16.65){\niceRed\scriptsize$\frac{1}{8}$}
		\put(18.75,11.025){\niceBlue\scriptsize$0$}
		\put(18.75,5.4){\niceBlue\scriptsize$0$}
		\put(24.375,39.15){\niceBlue\scriptsize$0$}
		\put(24.375,33.525){\niceBlue\scriptsize$0$}
		\put(24.375,27.9){\niceRed\scriptsize$\frac{1}{4}$}
		\put(24.375,22.375){\niceRed\scriptsize$\frac{1}{2}$}
		\put(24.375,16.65){\niceRed\scriptsize$\frac{1}{4}$}
		\put(24.375,11.025){\niceBlue\scriptsize$0$}
		\put(24.375,5.4){\niceBlue\scriptsize$0$}
		\put(30,39.15){\niceBlue\scriptsize$0$}
		\put(30,33.525){\niceBlue\scriptsize$0$}
		\put(30,27.9){\niceRed\scriptsize$\frac{1}{8}$}
		\put(30,22.375){\niceRed\scriptsize$\frac{1}{4}$}
		\put(30,16.65){\niceRed\scriptsize$\frac{1}{8}$}
		\put(30,11.025){\niceBlue\scriptsize$0$}
		\put(30,5.4){\niceBlue\scriptsize$0$}
		\put(35.625,39.15){\niceBlue\scriptsize$0$}
		\put(35.625,33.525){\niceBlue\scriptsize$0$}
		\put(35.625,27.9){\niceBlue\scriptsize$0$}
		\put(35.625,22.275){\niceBlue\scriptsize$0$}
		\put(35.625,16.65){\niceBlue\scriptsize$0$}
		\put(35.625,11.025){\niceBlue\scriptsize$0$}
		\put(35.625,5.4){\niceBlue\scriptsize$0$}
		\put(41.25,39.15){\niceBlue\scriptsize$0$}
		\put(41.25,33.525){\niceBlue\scriptsize$0$}
		\put(41.25,27.9){\niceBlue\scriptsize$0$}
		\put(41.25,22.275){\niceBlue\scriptsize$0$}
		\put(41.25,16.65){\niceBlue\scriptsize$0$}
		\put(41.25,11.025){\niceBlue\scriptsize$0$}
		\put(41.25,5.4){\niceBlue\scriptsize$0$}
		\put(55.306,39.15){\niceBlue\scriptsize$0$}
		\put(55.306,33.525){\niceBlue\scriptsize$0$}
		\put(55.306,27.9){\niceBlue\scriptsize$0$}
		\put(55.306,22.275){\niceBlue\scriptsize$0$}
		\put(55.306,16.65){\niceBlue\scriptsize$0$}
		\put(55.306,11.025){\niceBlue\scriptsize$0$}
		\put(55.306,5.4){\niceBlue\scriptsize$0$}
		\put(60.931,39.15){\niceBlue\scriptsize$0$}
		\put(60.931,33.525){\niceBlue\scriptsize$0$}
		\put(58.931,27.9){\niceRed\scriptsize$-\frac{1}{16}$}
		\put(59.531,22.375){\niceRed\scriptsize$-\frac{1}{8}$}
		\put(58.931,16.65){\niceRed\scriptsize$-\frac{1}{16}$}
		\put(60.931,11.025){\niceBlue\scriptsize$0$}
		\put(60.931,5.4){\niceBlue\scriptsize$0$}
		\put(66.556,39.15){\niceBlue\scriptsize$0$}
		\put(64.556,33.625){\niceRed\scriptsize$-\frac{\ii}{16}$}
		\put(64.256,27.9){\niceRed\scriptsize$-\frac{1+\ii}{8}$}
		\put(65.156,22.375){\niceRed\scriptsize$-\frac{1}{4}$}
		\put(64.256,16.65){\niceRed\scriptsize$-\frac{1-\ii}{8}$}
		\put(65.956,11.125){\niceRed\scriptsize$\frac{\ii}{16}$}
		\put(66.556,5.4){\niceBlue\scriptsize$0$}
		\put(72.181,39.15){\niceBlue\scriptsize$0$}
		\put(70.781,33.625){\niceRed\scriptsize$-\frac{\ii}{8}$}
		\put(70.781,27.9){\niceRed\scriptsize$-\frac{\ii}{4}$}
		\put(72.181,22.275){\niceRed\scriptsize$0$}
		\put(72.181,16.65){\niceRed\scriptsize$\frac{\ii}{4}$}
		\put(72.181,11.125){\niceRed\scriptsize$\frac{\ii}{8}$}
		\put(72.181,5.4){\niceBlue\scriptsize$0$}
		\put(77.806,39.15){\niceBlue\scriptsize$0$}
		\put(75.806,33.625){\niceRed\scriptsize$-\frac{\ii}{16}$}
		\put(76.806,27.9){\niceRed\scriptsize$\frac{1-\ii}{8}$}
		\put(77.806,22.375){\niceRed\scriptsize$\frac{1}{4}$}
		\put(76.806,16.65){\niceRed\scriptsize$\frac{1+\ii}{8}$}
		\put(77.206,11.125){\niceRed\scriptsize$\frac{\ii}{16}$}
		\put(77.806,5.4){\niceBlue\scriptsize$0$}
		\put(83.431,39.15){\niceBlue\scriptsize$0$}
		\put(83.431,33.525){\niceBlue\scriptsize$0$}
		\put(82.831,27.9){\niceRed\scriptsize$\frac{1}{16}$}
		\put(83.431,22.375){\niceRed\scriptsize$\frac{1}{8}$}
		\put(82.831,16.65){\niceRed\scriptsize$\frac{1}{16}$}
		\put(83.431,11.025){\niceBlue\scriptsize$0$}
		\put(83.431,5.4){\niceBlue\scriptsize$0$}
		\put(89.056,39.15){\niceBlue\scriptsize$0$}
		\put(89.056,33.525){\niceBlue\scriptsize$0$}
		\put(89.056,27.9){\niceBlue\scriptsize$0$}
		\put(89.056,22.275){\niceBlue\scriptsize$0$}
		\put(89.056,16.65){\niceBlue\scriptsize$0$}
		\put(89.056,11.025){\niceBlue\scriptsize$0$}
		\put(89.056,5.4){\niceBlue\scriptsize$0$}
	\end{overpic}
	\centering
	\caption{
		The values of the discrete residues of the first two negative-power discrete monomials.
		On the left, the values of $\frac{1}{2\pi}\ddeebar\zu^{[-1]}$;
		on the right, the values of $\frac{1}{2\pi}\ddeebar\zu^{[-2]}$.}
	\label{fig: discrete residues}
\end{figure}

Using the expression in \eqref{eq: dgrad repre}, it is clear that we have
\begin{align*}
	\dgrad\Fieldxi \big( \tfrac{1}{2} \big)
	\,+\, \Null
	\,=\,
	\big( \HolCurrModexi{-1} + \AntiHolCurrModexi{-1}\big)
	\fDGFFGround
	+
	\frac{1}{2}\,
	\big( \HolCurrModexi{-2} + \AntiHolCurrModexi{-2}\big)
	\fDGFFGround\,.
\end{align*}
An identical computation yields
\begin{align*}
	\dgrad\Fieldtheta \big( \tfrac{1}{2} \big)
	\,+\, \Null
	\,=\,
	\big( \HolCurrModetheta{-1} + \AntiHolCurrModetheta{-1}\big)
	\fDGFFGround
	+
	\frac{1}{2}\,
	\big( \HolCurrModetheta{-2} + \AntiHolCurrModetheta{-2}\big)
	\fDGFFGround\,.
\end{align*}
Now we just need to use Lemma~\ref{lemma: normal order of basis} to write
\begin{align}
\label{eq: computation1}
	&
	\noQuo{
	\Big(\dgrad\Fieldxi \big( \tfrac{1}{2} \big) + \Null \Big)
	\otimes
	\Big(\dgrad\Fieldtheta \big( \tfrac{1}{2} \big) + \Null \Big)
	} \, = \,
	\big( \HolCurrModexi{-1} + \AntiHolCurrModexi{-1}\big)
	\big( \HolCurrModetheta{-1} + \AntiHolCurrModetheta{-1}\big)
	\fDGFFGround
	\\ \nonumber
	& \mspace{190mu}
	\, + \,
	\frac{1}{2}\,
	\Big[
	\big( \HolCurrModexi{-1} + \AntiHolCurrModexi{-1}\big)
	\big( \HolCurrModetheta{-2} + \AntiHolCurrModetheta{-2}\big)
	+
	\big( \HolCurrModexi{-2} + \AntiHolCurrModexi{-2}\big)
	\big( \HolCurrModetheta{-1} + \AntiHolCurrModetheta{-1}\big)
	\Big]
	\fDGFFGround
	\phantom{\bigg\vert}
	\\ \nonumber
	& \mspace{190mu}
	\, + \,
	\frac{1}{4}\,
	\big( \HolCurrModexi{-2} + \AntiHolCurrModexi{-2}\big)
	\big( \HolCurrModetheta{-2} + \AntiHolCurrModetheta{-2}\big)
	\fDGFFGround\,.
\end{align}
On the other hand, by the definition of the normal ordering, we have
\begin{align}
\label{eq: computation2}
	\noQuo{
	\Big(\dgrad\Fieldxi \big( \tfrac{1}{2} \big) + \Null \Big)
	\otimes
	\Big(\dgrad\Fieldtheta \big( \tfrac{1}{2} \big) + \Null \Big)
	}
	\, = \, & \;
	\no{
	\dgrad\Fieldxi \big( \tfrac{1}{2} \big)
	\dgrad\Fieldtheta \big( \tfrac{1}{2} \big)
	}
	\, + \, \Null
	\phantom{\Big\vert}
	\\ \nonumber
	\, = \, & \;
	\dgrad\Fieldxi \big( \tfrac{1}{2} \big)
	\dgrad\Fieldtheta \big( \tfrac{1}{2} \big)
	-
	\wick{
	\dgrad \c\Fieldxi \big( \tfrac{1}{2} \big)
	\dgrad \c\Fieldtheta \big( \tfrac{1}{2} \big)
	}
	\, + \, \Null
	\phantom{\Big\vert}
	\\ \nonumber
	\, = \, & \;
	\dgrad\Fieldxi \big( \tfrac{1}{2} \big)
	\dgrad\Fieldtheta \big( \tfrac{1}{2} \big)
	-
	\frac{4\pi}{2}
	\, + \, \Null \,,
	\\ \nonumber
	\, = \, & \;
	4\pi \, \mathsf{P}_{\OEright}
\end{align}
where we have used the values of the full-plane discrete Green's function $\Green_{\Z^2}$ to compute the Wick's contractions.
Comparing~\eqref{eq: computation1} and \eqref{eq: computation2}, it becomes clear that we have
\begin{align*}
	&
	\mathsf{P}_{\OEright} \, = \,
	\frac{1}{4\pi}\,
	\big( \HolCurrModexi{-1} + \AntiHolCurrModexi{-1}\big)
	\big( \HolCurrModetheta{-1} + \AntiHolCurrModetheta{-1}\big)
	\fDGFFGround
	\\
	& \mspace{120mu}
	\, + \,
	\frac{1}{8\pi}\,
	\Big[
	\big( \HolCurrModexi{-1} + \AntiHolCurrModexi{-1}\big)
	\big( \HolCurrModetheta{-2} + \AntiHolCurrModetheta{-2}\big)
	+
	\big( \HolCurrModexi{-2} + \AntiHolCurrModexi{-2}\big)
	\big( \HolCurrModetheta{-1} + \AntiHolCurrModetheta{-1}\big)
	\Big]
	\fDGFFGround
	\phantom{\bigg\vert}
	\\
	& \mspace{120mu}
	\, + \,
	\frac{1}{16\pi}\,
	\big( \HolCurrModexi{-2} + \AntiHolCurrModexi{-2}\big)
	\big( \HolCurrModetheta{-2} + \AntiHolCurrModetheta{-2}\big)
	\fDGFFGround\,.
\end{align*}

\newpage
\titleformat{\section}{\normalfont\Large\bfseries}{\thesection}{0pt}{}

\end{document}